%% file: chaining-main.tex
\documentclass[11pt]{article}
\pdfoutput=1
\usepackage{fullpage}           
\usepackage{amsfonts}           
\usepackage{amsthm, amssymb}             
\usepackage{xspace}             
\usepackage{graphicx}
\usepackage{adjustbox}          
\usepackage{multicol}            
\usepackage{url}
\usepackage{enumerate, enumitem}  
\usepackage{subcaption}
\usepackage[usenames]{xcolor} 

\usepackage{amsmath, hyperref, nicefrac}
\usepackage[capitalize]{cleveref}
\usepackage{thm-restate}
\usepackage{parskip}
\usepackage{mathtools}
\usepackage{multirow}
\usepackage{algorithmic}
\usepackage{algorithm}
\usepackage{array}

\colorlet{darkgreen}{green!45!black}

\newcommand{\R}{\mathbb{R}}

\newcommand{\dist}{\text{dist}}
\newcommand{\eps}{\varepsilon}
\newcommand{\opt}{\text{OPT}}
\newcommand{\cost}{\text{cost}}
\newcommand{\calS}{\mathcal{S}}

\newcommand{\calE}{\mathcal{E}}

\newcommand{\E}{\mathbb{E}}
\newcommand{\pr}{\mathbb{P}}

\newcommand{\cand}{\mathbb{C}}
\newcommand{\greedy}{\mathcal{A}}
\newcommand{\A}{\mathcal{A}}

\newcommand{\coreset}{\Omega}

\newcommand{\out}{R_O}

\DeclareMathOperator{\sign}{sign}
\DeclareMathOperator{\argmax}{argmax}
\DeclareMathOperator{\argmin}{argmin}

\newtheorem{lemma}{Lemma}
\newtheorem{observation}{Observation}
\newtheorem{theorem}{Theorem}
\newtheorem{corollary}[theorem]{Corollary}
\newtheorem{claim}[theorem]{Claim}
\newtheorem{definition}{Definition}
\newtheorem{fact}{Fact}

\newcounter{sideremark}

\usepackage{amsmath}
\title{Towards Optimal Lower Bounds for $k$-median and $k$-means Coresets}
\author{Vincent Cohen-Addad\thanks{Google Research, Zurich.} \and Kasper Green Larsen\thanks{Aarhus University} \and 
  David Saulpic\thanks{Sorbonne Universit\'e, Paris}
  \and
  Chris Schwiegelshohn$^\dagger$
}
\date{}

\begin{document}
\maketitle

\begin{abstract}
  Given a set of points in a metric space, the $(k,z)$-clustering problem consists of finding a set of $k$ points called centers, such that the sum
  of distances raised to the power of $z$ of every data point to its closest center is minimized. Special cases include the famous
  $k$-median problem ($z=1$) and
  $k$-means problem ($z=2$). The $k$-median and $k$-means problems are at the heart of modern data analysis and massive data applications
  have given raise to the notion of coreset: a small (weighted) subset of the input point set preserving the cost of any solution to the problem
  up to a multiplicative $(1\pm \varepsilon)$ factor, hence reducing from large to small scale the input to the problem.
  
  While there has been an intensive effort to understand what is the best coreset size possible for both problems in various metric spaces, 
  there is still a significant gap between the state-of-the-art upper and lower bounds. 
  In this paper, we make progress on both upper and lower bounds, obtaining tight bounds for several cases, namely:
\begin{itemize}
\item In finite $n$ point general metrics, any coreset must consist of $\Omega(k \log n / \varepsilon^2)$ points. This improves on the $\Omega(k \log n /\varepsilon)$ lower bound of Braverman, Jiang, Krauthgamer, and Wu [ICML'19] and matches the upper bounds proposed for $k$-median by Feldman and Langberg [STOC'11] and $k$-means by Cohen-Addad, Saulpic, and Schwiegelshohn [STOC'21] up to polylog factors.
\item For doubling metrics with doubling constant $D$, any coreset must consist of $\Omega(k D/\varepsilon^2)$ points. This matches the $k$-median and $k$-means upper bounds by Cohen-Addad, Saulpic, and Schwiegelshohn [STOC'21] up to polylog factors. 
\item In $d$-dimensional Euclidean space, any coreset for $(k,z)$ clustering requires $\Omega(k/\varepsilon^2)$ points.
  This improves on the $\Omega(k/\sqrt{\varepsilon})$ lower bound of Baker, Braverman, Huang, Jiang, Krauthgamer, and Wu [ICML'20]
  for $k$-median and complements the $\Omega(k\min(d,2^{z/20}))$ lower bound of Huang and Vishnoi [STOC'20].
\end{itemize}
We complement our lower bound for $d$-dimensional Euclidean space with the construction of a coreset of size $\tilde{O}(k/\varepsilon^{2}\cdot \min(\varepsilon^{-z},k))$. This improves over the $\tilde{O}(k^2 \varepsilon^{-4})$ upper bound for general power of $z$ proposed by Braverman Jiang, Krauthgamer, and Wu [SODA'21] and over the $\tilde{O}(k/\varepsilon^4)$ upper bound for $k$-median by Huang and Vishnoi [STOC'20]. In fact, ours is the first construction breaking through the $\varepsilon^{-2}\cdot \min(d,\varepsilon^{-2})$ barrier inherent in all previous coreset constructions. To do this, we employ a novel chaining based analysis that may be of independent interest.
  Together our upper and lower bounds for $k$-median in Euclidean spaces are tight up to a factor $O(\eps^{-1} \text{polylog } k/\epsilon)$. 

\end{abstract}
\thispagestyle{empty}

\setcounter{page}{1}

\input{intro}
\input{preliminaries}
\input{lower}

\input{lower_discrete}

\input{chaining}

\bibliographystyle{plain}
\bibliography{references}

\appendix
\input{appendix}

\end{document}

%% file: intro.tex
\section{Introduction}
\label{sec:intro}

A clustering is a partition of a data set $P$ such that data points in the same cluster are similar and points in different clusters are dissimilar.
Various clustering problems have become important cornerstones in combinatorial optimization and machine learning problems.
Among these, center-based clustering problems are arguably the most widely studied and used. 
Here, the data elements lie in a metric space, every cluster is associated with a center point and the cost of a data point is some function of the distance between data point and its assigned cluster.
The $(k,z)$ problem captures this and other important objectives via the cost function
$$\cost(P,\calS) := \sum_{p\in P} \underset{s\in \calS}{~\min~} d(p,s)^z,$$
where $z$ is a positive integer, $|\calS| = k$ and  $d(~,~)$ denotes the distance function.
For $z=1$, this is $k$-median problem and for $z=2$, this is the equally intensely studied $k$-means problem.

Datasets used in practice are often huge, containing hundred of millions of points, distributed, or
evolving over time. Hence, in these settings classical heuristics (such as Lloyd or k-means++) are
lapsed; the size of the dataset forbids multiple passes over the input data and finding a “compact
representation” of the input data is of primary importance. This leads to a tradeoff: the smaller the dataset, the less storage we need and the faster we can run an algorithm on the data set, but conversely the smaller the data set the more information about the orginal data will be lost. \emph{Coresets} formalize and study this tradeoff. Specifically, given a precision parameter $\varepsilon$, $k$ and $z$, an $(\eps, k, z)$ coreset $\Omega$ is a subset of $P$ with weights $w: \Omega \rightarrow \mathbb{R}$ that approximates the cost of $P$ for any candidate solution $\calS$ up to a $(1\pm \varepsilon)$ factor, namely
\[\forall \calS,~~~(1-\eps)\cost(P, \calS) \leq \sum_{p\in \Omega} w(p) \cost(p, \calS) \leq (1+\eps)\cost(P, \calS).\]

A small $(\eps, k, z)$ coreset is therefore a good compression of the initial dataset, since it preserves the cost of any possible solution. Instead of storing the full dataset, one can simply store the coreset, saving on memory footprint and speeding up performances. We note that in some definitions, an offset $\Delta$ is added to the coreset: in that case, the coreset cost of solution $\calS$ is $(1\pm \eps)\cost(P, \calS) + \Delta$. In the case where the input space is infinite (e.g., Euclidean space), the coreset points may be chosen from the whole space, and are not restricted to be part of the input.

Although numerous great work focused on improving the size of coreset constructions, our understanding of coreset lower bounds is comparatively limited, and there is a significant gap between the best upper and lower bounds on the possible coreset size. For example, even for Euclidean $k$-means, nothing beyond the trivial $\Omega(k)$ lower bound is known. In this work, we attempt to systematically obtain lower bounds for these problems. 

We pay a particular attention to Euclidean Spaces. For those, we complement our lower bound with a new coreset construction that has an optimal dependency in $1/\eps$.

\subsection{Our Results} 
\renewcommand{\arraystretch}{1.3}

\begin{figure}
\begin{tabular}{|r|m{0.245\textwidth}|c|>{\centering\arraybackslash}m{0.23\textwidth}|}
\hline
Metric Space & Best upper bound & Best lower bound & Our result\\
\hline
Discrete Metrics & $O(k\eps^{-\max(2, z)}\log n)$ \cite{CSS21}& $\Omega(k \eps^{-1} \log n)$ \cite{baker2020coresets} & $\Omega(k \eps^{-2} \log n)$*\\
\hline
with doubling dimension $D$ & $O(k\eps^{-\max(2, z)}D)$ \cite{CSS21} & - & $\Omega(k \eps^{-2} D)$*\\
\hline
Euclidean $k$-median &  $\widetilde O(k\eps^{-4})$ \cite{huang2020coresets} & $\Omega(k \eps^{-1/2})$ \cite{baker2020coresets} & $\widetilde O(k \eps^{-3})$ ~~~~~~~~ $\Omega(k \eps^{-2})$\\
\hline
Euclidean $k$-means &  $\widetilde O(k\eps^{-4})$ \cite{CSS21} & - & $\Omega(k \eps^{-2})$\\
\hline
Euclidean & $\widetilde O(k\eps^{-2 - \max(2, z)})$ \cite{CSS21}  $\widetilde O(k^2 \eps^{-4})$ \cite{braverman20minor} & $\Omega(k 2^{z/100})$ \cite{huang2020coresets}& $\widetilde O_z(k \eps^{-2}  \cdot \min(\eps^{-z}, k))$ $\Omega(k \eps^{-2})$\\
\hline
\end{tabular}
\caption{Comparison between the state-of-the-art bounds and our results. Results marked with * are tight for $k$-median and $k$-means.}
\end{figure}

\renewcommand{\arraystretch}{1}

We settle the complexity of the problem for several cases.
First, for finite $n$-point metrics, we prove the following theorem.
\begin{restatable}{theorem}{discrete}
\label{thm:discrete}
For any $0 < \eps < 1/2$, $k$ and $n \geq \eps^{-5}$ such that $\log k = O(\log n)$, there exists a finite $n$ point metric such that any $(\varepsilon,k,z)$ coreset  using offset $\Delta$ consists of at least $\Omega\left( \frac{k}{\eps^2 } \log n\right)$ points.
\end{restatable}

Our result improves over the $\Omega(k \varepsilon^{-1} \log n)$ lower bound of Baker, Braverman, Huang, Jiang, Krauthgamer, and Wu~\cite{baker2020coresets}.
For the $k$-median and $k$-means objective matches the upper bounds proposed in Feldman and Langberg \cite{FeldmanL11} and Cohen-Addad, Saulpic, and Schwiegelshohn \cite{CSS21} up to polylog$(1/\eps)$ factors. 

For metric space with doubling dimension $D$, we present a lower bound similar to that of \cref{thm:discrete}:
\begin{corollary}\label{cor:discrete}
For any $\eps, k, D$ such that $D \geq 5\log 1/\eps$ and $\log k = O(D)$, there exists a graph with doubling dimension $D$ on which any $(\eps, k, z)$-coreset using offset $\Delta$ must have size $\Omega\left(\frac{k D}{\eps^2}\right)$.
\end{corollary}
This matches up to polylog$(1/\eps)$ factors the upper bound from \cite{CSS21} for $k$-median and $k$-means.

We also study Euclidean spaces more specifically. Here, the difficulty is that centers can be placed arbitrarily in the space, and not only at input points. 
Our main results for Euclidean spaces is the following. 

\begin{theorem}[See \cref{thm:lb-euclidean} for the exact statement]
For any $0 < \eps < 1/2$ and any $k$, there exists a point set such that any $(\varepsilon,k,z)$ coreset using offset $\Delta$ consists of at least $\Omega\left( \frac{k}{\eps^2 \max\{1, z^4\}}\right)$ points.
 \end{theorem}
 
This lower bound holds for any selection of points (i.e. the coreset may use non-input points), and for any additive offset, which is a generalization initially proposed by Feldman, Schmidt, and Sohler \cite{FeldmanSS20} and which has since been used in a number of other papers, see Cohen, Elder, Musco, Musco, and Persu~\cite{CEMMP15}, Sohler and Woodruff~\cite{SohlerW18} and Cohen-Addad, Saulpic and Schwiegelshohn~\cite{CSS21b}. 
The only previously known results are the $\Omega(k/\sqrt{\varepsilon})$ bound for $k$-median by Baker, Braverman, Jiang, Krauthgamer, and Wu~\cite{baker2020coresets}, and the $\Omega(k\cdot \min(d,2^{z/20}))$ bound by Huang and Vishnoi \cite{huang2020coresets}. Thus, we obtain the first non-trivial lower bound for Euclidean $k$-means. 

We complement the lower bound with the following theorem.
\begin{theorem}
\label{thm:main}
Given a set of points $P$ in $d$-dimensional Euclidean space and any $\varepsilon>0$, there exists an $(\varepsilon,k,z)$ coreset of size $\tilde{O}(k\cdot \varepsilon^{-2} \cdot 2^{O(z\log z)}\cdot \min(\varepsilon^{-z},k))$. 
\end{theorem}

This is the first coreset construction with an optimal dependency on $\varepsilon$, at the cost of a quadratic dependency on $k$. Previously, all upper bounds either had a dependency of at least $\varepsilon^{-4}$ \cite{BravermanJKW21,CSS21,huang2020coresets} or a dependency on $d$ \cite{Chen09,FeldmanL11}.

We note that for the special case of Euclidean $k$-median, we improve the best coreset size from $O(k\cdot \varepsilon^{-4})$ to $O(k\cdot \varepsilon^{-3})$, taking a step to reduce the gap with the lower bound.


A complete overview of previous coreset bounds for Euclidean spaces and finite metrics is given in Table~\ref{table:core}. For further related work, we refer to \cref{ap:related}.

\subsection{Overview of our Techniques}
\label{sec:overview}
Our results for the Euclidean setting require several important new technical insights and we thus review them first.
We later review our approach for our lower bound for general metrics.
\paragraph{Euclidean Lower Bounds}
The lower bound proof consists of three separate steps which combined proves that any coreset for the point set $P=\{e_1,\dots,e_d\}$ in $\R^{d}$ (i.e., the standard basis of $\R^d$) must have size $\Omega(k \cdot \eps^{-2})$ (in this proof overview, we focus on $z=2$) when $d =
\Theta(k \cdot \eps^{-2})$. The basic approach is to show that any clustering of $P$ with $k$ centers has large cost, while at the same time, for any coreset $\Omega$ using $o(d)$ weighted points, there is a low cost clustering. Combining the two yields the lower bound. We carry out this proof in three steps. In the first step, we show that any clustering of $P$ using unit norm centers has cost at least $2d - O(\sqrt{dk})$. In the next step, we show that for any coreset $\Omega$ consisting of $t$ points and a weighing $w : \Omega \to \R^+$, there is a low-cost clustering using unit norm centers that has cost $2d - \Omega(\sqrt{k/t} \cdot \sum_{p \in \Omega} w(p)\|p\|_2)$. Combining this with step one implies $\sum_{p \in \Omega} w(p)\|p\|_2 = O(\sqrt{td})$. In the final step, we show that any coreset $\Omega$ must have $\sum_{p \in \Omega} w(p) \|p\|_2 = \Omega(d)$ when $d = \Theta(k \cdot \eps^{-2})$. Combining this with the previous two steps finally yields $\sqrt{td}=\Omega(d) \Rightarrow t=\Omega(d) \Rightarrow t = \Omega(k \cdot \eps^{-2})$. In the following, we elaborate on the high level ideas needed for each of the steps:

1. First, we show that any clustering of $P$ using $k$ cluster centers $c_1,\dots,c_k$ of unit norm, must have cost at least $2d - O(\sqrt{dk})$. To see this, notice that if $e_i$ is assigned to cluster center $c_j$, then the cost of $e_i$ is $\|e_i - c_j\|_2^2  = \|e_i\|_2^2 + \|c_j\|_2^2 - 2\langle e_i ,c_j \rangle = 2 -2 c_{j,i}$, where $c_{j,i}$ denotes the $i$'th coordinate of $c_j$. Any cluster center $c_j$ can thus at most reduce the cost of the clustering below $2d$ by an additive $2\sum_i c_{j,i} \leq 2 \|c_j\|_1$. Moreover, it is only ``wasteful'' to assign a value different from $0$ to $c_{j,i}$ if $e_i$ is not assigned to center $c_j$ (wasteful since $c_j$ is required to have unit norm). Thus the $k$ centers can be thought of as having disjoint supports. Thus on average, they only have $d/k$ coordinates available. By Cauchy-Schwartz (i.e. the maximum ratio between $\|c_j\|_1$ and $\|c_j\|_2$), we can argue that $\sum_j \|c_j\|_1 \leq \sqrt{d/k} \sum_j \|c_j\|_2 = \sqrt{dk}$ and the conclusion follows.

2. Next, we argue that for any coreset $\Omega$ consisting of $t$ points and a weighing $w : \Omega \to \R^+$, we can find a low-cost clustering in terms of $\sum_{p \in \Omega}w(p)\|p\|_2$ using unit norm centers. This is achieved by partitioning the points of the coreset into $k$ groups of $\ell = t/k$ points each and using one center for each group. For a group of $\ell$ points $r_1,\dots,r_{\ell}$, we choose the center as something that resembles the mean scaled to have unit norm. More precisely, we consider a random vector $u = \sum_{i=1}^\ell \sigma_i w(r_i) r_i$ for uniform random and independent signs $\sigma_i$. We can then argue that there is a fixing of the signs, such that if $u$ is scaled to have unit norm and this is repeated for all $k$ groups, the resulting cluster cost is at most $2d - \Omega(\sqrt{k/t} \sum_{p \in \Omega} w(p)\|p\|_2)$.

3. In the last step, we need to argue that any coreset $\Omega$ and weighing $w : \Omega \to \R^+$ must have $\sum_{p \in \Omega} w(p) \|p\|_2 = \Omega(d)$ when $d = \Theta(k \cdot \eps^{-2})$. This is the technically most challenging part of the proof. The basic idea for arguing this, is to exploit that $\Omega$ must be a coreset for many different clusterings of $P = \{e_1,\dots,e_d\}$. In particular, we consider the Hadamard basis over $q=d/k$ coordinates. The Hadamard basis consists of $q$ orthogonal vectors with coordinates in $\{-1/\sqrt{q},1/\sqrt{q}\}$, all having at least half of the coordinates equal to $1/\sqrt{q}$. For each vector $v$ in the basis, we consider a clustering where we use $k$ centers $c_1,\dots,c_k$ that are all copies of $v$ shifted to take up either the first $q$ coordinates in $\R^{d}$, the next $q$ coordinates and so on. Since half of the coordinates of any $v$ are $1/\sqrt{q}$, the cost of this clustering on $P$ is $2d-\Omega(d/\sqrt{q})$ (if $e_i$ is assigned to a center with the $i$'th coordinate is equal to $1/\sqrt{q}$ then the cost of $e_i$ is $2 - 2/\sqrt{q}$). Thus intuitively, the points $r_1,\dots,r_t$ in any coreset $\Omega$ also must have $\sum_{i=1}^t \max_{j=1}^k \langle r_i, c_j \rangle = \Omega(d/\sqrt{q})$. This means that on average over all $r_i$, we have $\max_{k=1}^k w(r_i)\langle r_i, c_j \rangle = \Omega(d/(t \sqrt{q}))$. The crucial observation is that we can repeat this argument for every $v$ in the basis. There are $q$ such $v$'s. Moreover, for any point $r_i$ in the coreset, the set of $q$ centers $c_{i_1}^1,\dots,c_{i_q}^q$ it is assigned to in these $q$ different clusterings are all orthogonal vectors. Thus by Cauchy-Schwartz, we must have $\sqrt{q} d/t \leq \sum_{j=1}^q \langle w(r_i) r_i ,c_{i,j}^j  \rangle = \langle w(r_i) r_i, \sum_{j=1}^q c_{i_j}^j \rangle \leq \|w(r_i)r_i\|_2 \|\sum_{j=1}^q c_{i,j}^j \|_2 = w(r_i)\|r_i\|_2 \sqrt{q}$. That is, $w(r_i) \|r_i\|_2 = \Omega(d/t)$. Summing over all $r_i$ completes the proof. Finally, let us remark where the requirement $d = \Theta(k \cdot \eps^{-2})$ enters the picture. We argued that the cost of clustering $P$ using the Hadamard basis was $2d-\Omega(d/\sqrt{q})$. In the coreset, the clustering is allowed to be a factor $(1 + \eps)$ larger. We thus require that $(2d - \Omega(d/\sqrt{q}))(1+\eps) \leq 2d - \Omega(d/\sqrt{q})$, which is satisfied when $d\eps = O(d/\sqrt{q}) \Leftrightarrow q = O(\eps^{-2})$. But $q = d/k$ and thus this translates into $d = O(k \cdot \eps^{-2})$.

\paragraph{Upper Bounds} 
Our main technical contribution is an application of chaining techniques used to analyse Gaussian processes for coreset construction, see Talagrand for an extensive introduction~\cite{talagrand1996majorizing}. To the best of our knowledge, we are not aware of any prior attempts of using chaining to improve coreset bounds directly.

For readers that may not be familiar with the technique, we now highlight how it allows us to improve over previous constructions.
For every candidate solution $\calS$, we say that $v^{\calS}$ is the cost vector associated with $\calS$, where $v^{\calS}_p$ is simply the cost of point $p$ in $\calS$.
A sampling based coreset now picks rows of $v^{\calS}$ according to some distribution and approximates $\|v^{\calS}\|_1 = \sum v^{\calS}_p$ as the weighted average of the costs of the picked points.
To show that this weighted average is concentrated, we require two ingredients.
First, we bound the variance for approximating any $\|v^{\calS}\|_1$. Suppose we make the simplifying assumption that all points less than 1 and that we are aiming for an additive error of at most $\varepsilon\cdot n$.
In this case, the variance is constant, upon which applying a Chernoff bound requires only $\textbf{Var}\cdot \varepsilon^{-2}$ samples to approximate any single $\|v^{\calS}\|_1$.

Second, we have to apply a union bound over all $v^{\calS}$. In Euclidean spaces, a naive union bound is useless, as there are infinitely many candidate solutions. To discretize $\calS$, previous work, either implicitly or explicitly, showed that there exists a small set of vectors $\mathbb{N}^{\varepsilon}$, henceforth called a net, such that for every $v^{\calS}$ there exists $v^{\calS,\varepsilon}\in \mathbb{N}^{\varepsilon}$ with $|v^{\calS,\varepsilon}_p - v^{\calS}_p|\leq \varepsilon$. Thus, an accurate estimation of $\|v\|_1 $ for all $v\in\mathbb{N}^{\varepsilon}$ is sufficient to achieve an estimation for all $v^{\calS}$.
Unfortunately, the only known bounds of $\mathbb{N}^{\varepsilon}$ are of the order $\exp(k\min(d,\varepsilon^{-2}))$, which combined with bound of the variance leads to $\log |\mathbb{N}^{\varepsilon}| \cdot \textbf{Var}\cdot \varepsilon^{-2} = k \cdot \varepsilon^{-2}\cdot \min(\varepsilon^{-2},d)$ many samples.

To improve upon this idea, we use nets at different scales, i.e. we have nets $\mathbb{N}^{1}$, $\mathbb{N}^{1/2}$, $\mathbb{N}^{1/4}$ and so on. These nets allow us to write every $v^{\calS}$ as a telescoping sum of net vectors at different scales, that is
$$ v^{\calS} = \sum_{h=0}^{\infty} v^{\calS,2^{-(h+1)}} - v^{\calS,2^{-h}},$$
where $v^{\calS,2^{-h}}$ is an element of $\mathbb{N}^{2^{-h}}$.
Instead of applying the union bound for all vectors in $\mathbb{N}^{\varepsilon}$ at once, we apply the union bound for all difference vectors at various scales, i.e. we show that for all difference vectors $v^{\calS,2^{-(h+1)}} - v^{\calS,2^{-h}}$
$$\mathbb{P}\left[|v^{\calS,2^{-(h+1)}} - v^{\calS,2^{-h}} - \mathbb{E}[v^{\calS,2^{-(h+1)}} - v^{\calS,2^{-h}}|] \geq \varepsilon \cdot n\right]$$
is small.

The reason why this improves over the naive discretization is that as the nets get finer, the difference also gets smaller, i.e. $|v^{\calS,2^{-(h+1)}}_p - v^{\calS,2^{-h}}_p| \leq 2\cdot 2^{-h}$. This difference directly affects the bound on the variance, which decreases from a constant to roughly $2^{-2h}\cdot O(1)$.
Since there are only $|\mathbb{N}^{2^{-(h+1)}}|\cdot |\mathbb{N}^{2^{-h}}| \in \exp(k\cdot  2^{-2h}\cdot O(1))$ many difference vectors, we can compensate the increase in net size by a decrease in variance, i.e. we require only 
$$ \log (|\mathbb{N}^{2^{-(h+1)}}|\cdot |\mathbb{N}^{2^{-h}}|) \cdot  \textbf{Var}\cdot \varepsilon^{-2} \approx k\cdot  2^{-2h}\cdot O(1) \cdot 2^{-2h} \cdot \varepsilon^{-2} = k \cdot \varepsilon^{-2} \cdot O(1)$$

many samples.
Applying this idea to every successive summand of the telescoping sum (or rather to every link of the chain of net vectors), leads to an overall number of samples of the order $k\cdot \varepsilon^{-2}$, ignoring polylog factors.

Unfortunately, improving the analysis from an additive approximation to a multiplicative approximation leads to several difficulties. Without using the assumption that all points cost less than $1$, the variance increases.
Indeed, contrasting to the previous work \cite{CSS21b} that used a chaining-based analysis to obtain coreset bounds for a single center and previous work \cite{CSS21} that used a chaining-inspired variance reduction technique, both of which managed to obtain constant variance, bounding the variance in this setting is highly non-trivial and requires a number of new ideas.
The lowest variance we could show for estimating $\|v^{\calS}\|_1$ is only of the order $\min(\varepsilon^{-z},k)$, leading to the (likely suboptimal) bound of $\tilde O(k\cdot\varepsilon^{-2}\cdot \min(\varepsilon^{-z},k))$ and moreover this bound on the variance is tight. Further ideas will be necessary to reach the (conjectured) optimal bound of $\Theta(k\cdot \varepsilon^{-2})$.

\paragraph{Lower Bound for discrete metric spaces} The general idea behind our lower bound is to use the tight concentration and anti-concentration
bounds on the sum of random variables.

We first build an instance for $k=1$, and combines several copies of it to obtain a lower bound for any arbitrary $k$.
Our instance for $k=1$ is such that: (1) when $|\Omega| \leq \eps^{-2}\log |C|$ there exists a center with $\cost(\Omega, c) > (1+100\eps) \cost(c)$, and (2):  for any $|\Omega| > \eps^{-2}\log |C|$ there exists a center $c$ with $\cost(\Omega, c) \in (1\pm \eps) \cost(c)$.

 To show the existence of such an instance, we consider a complete bipartite graph with nodes $P \cup C$ where there is an edge between each point of $P$ and each point of $C$, with length $1$ with probability $1/4$ and $2$ otherwise. The set of clients is $P$. For simplicity, we will assume here that the coreset weights are uniform. Making the idea work for non-uniform weights requires several other technical ingredients.

In that instance for $k=1$, the cost of a solution (with a single center, $c$) is fully determined by $n_1(c)$, the number of length $1$ edges to $c$. Indeed, $\cost(c) = 2(|P| - n_1(c)) + n_1-c) = 2|P| - n_1(c)$. Let us further assume that $n_1(c)$ is equal to its expectation, $\delta |P|$.
  For a fixed subset of points $\Omega$, the cost of the solution for $\Omega$ with uniform weights $\frac{|P|}{|\Omega|}$ verifies the same equation: it is $2|P| - n_1(\Omega, c)\cdot\frac{|P|}{|\Omega|}$, where  $n_1(\Omega, c)$ the number of length $1$ edges from $\Omega$ to $c$. Note that $\E[n_1(\Omega, c)] = \delta |\Omega|$.

Using anti-concentration inequalities, we show that $n_1(\Omega, c) > (1+200\eps) \E[n_1(\Omega, c)]$ with probability at least $\exp(-\alpha \eps^2|\Omega|)$, for some constant $\alpha$. 
When this event happens, then $\Omega$ does not preserve the cost of solution $c$: indeed, 
\begin{align*}
2|P| - n_1(\Omega, c)\cdot\frac{|P|}{|\Omega|} &> 2|P| - (1+200\eps)\delta |\Omega|\cdot\frac{|P|}{|\Omega|}\\
&= 2|P| - \delta|P| + 200\eps \delta |P| > (1+100\eps)(2|P| - n_1(c)).
\end{align*}

Since the edges are drawn independently, the coreset cost for all possible centers $c$ is independent. Hence, there exists one center with  $n_1(\Omega, c) > (1+200\eps)\delta |\Omega|$ with probability at least $1-(1-\exp(-\alpha \eps^2|\Omega|))^{|C|}$. By doing a union-bound over all possible subsets $\Omega$, one can show the following: with positive (close to $1$) probability, for any $|\Omega| \leq \eps^{-2}\log |C|$ there exists a center with $\cost(\Omega, c) > (1+100\eps) \cost(c)$. 

Using standard concentration inequality, one can show that with probability close to $1$, for any $|\Omega| > \eps^{-2}\log |C|$, there exists  a center $c$ with $\cost(\Omega, c) \in (1\pm \eps) \cost(c)$. Since the probabilities are taken on the edges randomness, those two result ensure the existence of a graph that verifies properties (1) and (2) desired for the $k=1$ instance.

Now, the full instance is made of $k$ distinct copies $X_1, ..., X_k$ of the $k=1$ instance, placed at infinite distance from each other. Let $P_i$ be the set of clients of $X_i$: the clients for the full instance are $\cup P_i$. Let $\Omega$ be a set of at most $1/100\cdot k \eps^{-2} \log n$ points: we show that $\Omega$ cannot be a coreset. By Markov's inequality, there are at least $99/100 k$ copies that contain less than $\eps^{-2} \log n$ points of $\Omega$. We say those copies are \textit{bad}, the others are \textit{good}. Consider now the solution $\calS$ defined as follows: from each $X_i$, take the center such that $\cost(\Omega \cap P_i, c) > (1+100\eps) \cost(P_i, c)$ when $X_i$ is bad, and the center such that $\cost(\Omega \cap P_i, c) \in (1\pm \eps) \cost(P_i, c)$ when $X_i$ is good. Observe also that by construction of the instance for $k=1$, the cost in each copy must lie in $[|P|, 2|P|]$. For that solution, we have:
\begin{align*}
\cost(\Omega, \calS) &= \sum \cost(\Omega \cap P_i, s_i) =  \sum_{i \text{ bad }} \cost(\Omega \cap P_i, s_i) + \sum_{i \text{ good}} \cost(\Omega \cap P_i, s_i)\\
&>\sum_{i \text{ bad }} (1+100\eps)\cost(P_i, s_i) + \sum_{i \text{ good}} (1-\eps)\cost(P_i, s_i)\\
&> \cost(\calS) + \frac{99k}{100} \cdot 100\eps |P| - \frac{k}{100} \cdot \eps 2|P| > \cost(\calS) + 98 k \eps |P| > (1+\eps) \cost(\calS).
\end{align*}
Hence, any $\Omega$ with $|\Omega| \leq 1/100\cdot k \eps^{-2} \log n$ cannot be a coreset for our instance, which concludes the proof.

\section{Related Work} \label{ap:related}

\begin{table}[H]
\begin{center}
\begin{tabular}{r|c}
Reference & Size (Number of Points) 
\\
\hline\hline
\multicolumn{2}{l}{{\bf Coreset Bounds in Euclidean Spaces}} \\
\hline\hline
\multicolumn{2}{l}{Lower Bounds} \\\hline\hline
Baker, Braverman, Huang, Jiang,  & \multirow{2}{*}{$\Omega(k\cdot \varepsilon^{-1/2})$} \\
Krauthgamer, Wu (ICML'19)~\cite{BravermanJKW19} & \\ 
\hline
Huang, Vishnoi (STOC'20)~\cite{huang2020coresets} & $\Omega(k\cdot \min(d,2^{z/20}))$  
\\
\hline
\large \textbf{This paper} & $\Omega(k \cdot \varepsilon^{-2}/z^4)$ 
\\
\hline
\hline
\multicolumn{2}{l}{Upper Bounds} \\\hline\hline
Har-Peled, Mazumdar (STOC'04)~\cite{HaM04} & $O(k\cdot \varepsilon^{-d}\cdot \log n)$ 
\\
\hline
Har-Peled, Kushal (DCG'07)~\cite{HaK07} & $O(k^3 \cdot\varepsilon^{-(d+1)})$  
\\\hline
Chen (Sicomp'09)~\cite{Chen09} & $O(k^2 \cdot d \cdot \varepsilon^{-2}\cdot \log n)$ 
\\\hline
Langberg, Schulman (SODA'10)~\cite{LS10} &$O(k^3\cdot d^2 \cdot \varepsilon^{-2})$ 
\\\hline
Feldman, Langberg (STOC'11)~\cite{FeldmanL11} & $O(k\cdot d\cdot\varepsilon^{-2z})$ 
\\\hline
Feldman, Schmidt, Sohler (Sicomp'20)~\cite{FeldmanSS20} & $O(k^3\cdot\varepsilon^{-4})$ 
\\\hline
Sohler, Woodruff (FOCS'18)~\cite{SohlerW18} & $O(k^2\cdot\varepsilon^{-O(z)})$  
\\\hline
Becchetti, Bury, Cohen-Addad, Grandoni,  & \multirow{2}{*}{$O(k\cdot\varepsilon^{-8})$}\\ Schwiegelshohn (STOC'19)~\cite{BecchettiBC0S19} & \\ 
\hline
Huang, Vishnoi (STOC'20)~\cite{huang2020coresets} & $O(k\cdot\varepsilon^{-2-2z})$ 
\\
\hline
Bravermann, Jiang, Krautgamer, Wu (SODA'21)~\cite{braverman20minor} & \multirow{1}{*}{$ O (k^2 \cdot\eps^{-4})$} 
\\
\hline
Cohen-Addad, Saulpic, Schwiegelshohn (STOC'21)~\cite{CSS21} & $\tilde O (k\cdot\eps^{-2-\max(2,z)})$ 
\\
\hline
\large \textbf{This paper} & $\tilde O (k \cdot \varepsilon^{-2} \cdot \min(\eps^{-z},k))$ 
\\
\hline\hline
\multicolumn{2}{l}{\textbf{General $n$-point metrics, $D$ denotes the doubling dimension}} \\\hline\hline
\multicolumn{2}{l}{Lower Bounds} \\\hline\hline
Braverman, Jiang, Krauthgamer, Wu (ICML'19)~\cite{BJK19} & $\Omega(k\cdot\varepsilon^{-1}\cdot\log n)$ 
\\
\hline
\large \textbf{This paper} & $ \Omega (k \cdot \varepsilon^{-2} \cdot \log n)$ 
\\
\hline
\large \textbf{This paper} & $ \Omega (k \cdot \varepsilon^{-2} \cdot D)$ 
\\
\hline \hline
\multicolumn{2}{l}{Upper Bounds} \\\hline\hline
Chen (Sicomp'09)~\cite{Chen09} & $O(k^2\cdot\varepsilon^{-2}\cdot\log^2 n)$ 
\\\hline
Feldman, Langberg (STOC'11)~\cite{FeldmanL11} & $O(k\cdot \varepsilon^{-2z}\cdot  \log n )$  
\\\hline
Huang, Jiang, Li, Wu~(FOCS'18)~\cite{HuangJLW18} & $O(k^3\cdot \varepsilon^{-2}\cdot D)$ 
\\\hline
 Cohen-Addad, Saulpic, Schwiegelshohn (STOC'21) \cite{CSS21} & $\tilde O(k\cdot \varepsilon^{-\max(2,z)}\cdot D)$ 
\\
\hline
 Cohen-Addad, Saulpic, Schwiegelshohn (STOC'21) \cite{CSS21} & $\tilde O(k\cdot \varepsilon^{-\max(2,z)}\cdot \log n)$ 
\\
\hline
\end{tabular}
\end{center}
\caption{Comparison of coreset sizes for $(k,z)$-Clustering in Euclidean spaces. 
\cite{BravermanJKW19} only applies to $k$-median,
\cite{HaK07,HaM04} only applies to $k$-means and $k$-median, and\cite{BecchettiBC0S19,FeldmanSS20} only applies to $k$-means. 
\cite{SohlerW18} runs in exponential time, which has been addressed by Feng, Kacham, and Woodruff \cite{FKW19}.
Aside from~\cite{HaK07,HaM04}, the algorithms are randomized and succeed with constant probability. Any dependency on $2^{O(z\log z)}$, as well as polylog factors have been omitted in the upper bounds.}
\label{table:core}
\end{table}

For the most part, related work on coresets for $k$ clustering in Euclidean spaces are given in Table~\ref{table:core}. 
A closely related line of research focusses on dimension reduction for $k$-clustering objectives, particularly $k$-means.
Starting with~\cite{DrineasFKVV04}, a series of results \cite{BecchettiBC0S19,BoutsidisMD09,BoutsidisZD10,BoutsidisZMD15,CEMMP15,Cohen-AddadS17,FeldmanSS20,FKW19,KuK10,MakarychevMR19,SohlerW18} explored the possibility of using dimension reduction methods for $k$-clustering, with a particular focus on principal component analysis (PCA) and random projections. The problem of dimension reduction, at least with respect to these techniques has been mostly resolved by now: Cohen, Elder, Musco, Musco, and Persu~\cite{CEMMP15} proved tight bounds of $\lceil k/\varepsilon\rceil$ for PCA and Makarychev, Makarychev and Razenshteyn~\cite{MakarychevMR19} gave a bound of $O(\varepsilon^{-2}\log k/\varepsilon)$ for random projections, which nearly matches the lower bound by Larsen and Nelson~\cite{LarsenN17}. The arguably most important technique for combining dimension reduction with coresets is the recent work on terminal embeddings, see~\cite{ChN21,ElkinFN17,MahabadiMMR18}. 
Notably, Narayanan and Nelson~\cite{NaN18} gave an optimal bound of $O(\varepsilon^{-2}\log n)$. 
We will discuss specifics on terminal embeddings in Section~\ref{sec:nets}.

While Euclidean spaces are doubtlessly the most intensively studied metric, a number of further metrics have also been considered, including finite metrics~\cite{Chen09,CSS21,FeldmanL11}, doubling metrics~\cite{CSS21,HuangJLW18}, and graph metrics~\cite{baker2020coresets,BravermanJKW21,CSS21}.
Coresets also feature prominently in streaming literature, see~\cite{BravermanFLR19,BravermanFLSY17,BravermanLLM16,FGSSS13,FrahlS2005} for results with a special focus on various streaming models.
Other related work considers generalizations of $k$-median and $k$-means by either adding capacity constraints~\cite{BFS20,Cohen-AddadL19,HuangJV19,SSS19}, generalizing the notion of centers to subspaces \cite{BJKW21,FeldmanL11,FeldmanMSW10}, time series \cite{huang2021coresets} or sets \cite{JubranTMF20} or considering more general objective functions~\cite{BachemLL18,BravermanJKW19}. Coresets have also been studied for many other problems: we cite non-comprehensively decision trees \cite{JSNF21}, kernel methods \cite{JKLZ21,KarninL19,PhillipsT20}, determinant maximization \cite{IndykMGR20}, diversity maximization \cite{IndykMMM14}, shape fitting problems~\cite{AHV04,Chan09}, linear regression~\cite{BoutsidisDM13,HuangSV20,TukanMF20}, logistic regression~\cite{huggins2016coresets,MunteanuSSW18}, Gaussian mixtures~\cite{LucicFKF17}, dependency networks~\cite{MK18}, or low‐rank approximation~\cite{maalouf2019fast}. The interested reader is referred to~\cite{AHK05,Feldman20,MunteanuS18} and similar surveys for more pointers to coreset literature.

In terms of approximation guarantee, the best known approximation ratio for general metrics is 2.67 due to Byrka et al.~\cite{BPRST15}, improving over the
result of 2.71 of Li and Svensson~\cite{LiS16} while computing a better than $1+2/e$-approximation has been shown to be NP-hard by Guha and
Khuller~\cite{GuK99}.
In Euclidean spaces of arbitrary dimension, the best known approximation is 2.408 and 5.957 for $k$-median and $k$-means, respectively, due to a recent result of
Cohen-Addad et al.~\cite{CohenAddadEsfandiari22} who improved over
the work of Grandoni et al.~\cite{GrandoniORSV22} and Ahmadian et al.~\cite{AhmadianNSW20}.
The best known hardness of approximation is 1.73 and 1.27 for $k$-means and $k$-median assuming the Johnson-Coverage Hypothesis or
1.17 and 1.07 respectively assuming P $\neq$ NP~\cite{doi:10.1137/1.9781611977073.63} (see also~\cite{Cohen-AddadS19,Cohen-AddadSL21,LeeSW17}).
For graphs excluding a fixed-minor, the problem is NP-Hard~\cite{MarxP15} and a PTAS is known~\cite{CKM19,Cohen-AddadPP19}.
For doubling metrics, the problem is NP-Hard
(even in the plane~\cite{DBLP:journals/siamcomp/MegiddoS84}) and a linear-time approximation scheme when the dimension is considered constant
is known~\cite{Cohen-AddadFS21,Cohen-Addad18,KoR07}.

\subsection{Roadmap}
The proof of the Euclidean lower bound for $k$-Means is given in Section~\ref{sec:lower}. The proof for general powers is given in Appendix~\ref{ap:euclidean}.
The lower bounds for finite metrics and doubling metrics are given in Section~\ref{sec:lower-discrete}.
The proof of the upper bound is given in Section~\ref{sec:upper}.

%% file: preliminaries.tex
\section{Preliminaries}
\paragraph{General Preliminaries}

Given two points $p$ and $c$ in some metric space with distance function $\dist$, the $(k,z)$-clustering cost of $p$ to $c$ is $\cost(p,c)=\dist^z(p,c).$ The $\ell_p$ norm of a $d$ dimensional vector $x$ is defined as $\|x\|_p := \sqrt[p]{\sum_{i=1}^d |x|_i^p}$. If the value of $p$ is unspecified, it is meant to be the Euclidean norm $p=2$.
Given a set of point $P$ with weights $w : P \rightarrow \mathbb{R}^+$ on a metric space $I$ and a solution $\calS$, we define $\cost_I(P, \calS) := \sum_{p \in P} w(p)\cost(p, \calS)$.

\begin{definition}
Let $(X,\dist)$ be a metric space, let $P\subset X$ be a set of \emph{clients} and let $\coreset$ be a set of points with weights $w : \coreset \rightarrow \R^+$ and a constant $\Delta$.
$\coreset$ is an $(\eps,k,z)$-coreset using offset $\Delta$ if for any set $\calS \subset X$, $|\calS| = k$, 
\[\left\vert\sum_{p \in P} \cost(p, \calS) - \left(\Delta + \sum_{p \in \coreset} w(p)\cost(p, \calS)\right) \right\vert \leq \eps \sum_{p \in P} \cost(p, \calS) \]
$\coreset$ is a $(\eps,k,z)$-coreset using offset $\Delta$ with additive error $E$ if for any set $\calS \subset X$, $|\calS| = k$, 
\[\left\vert\sum_{p \in P} \cost(p, \calS) - \left(\Delta + \sum_{p \in \coreset} w(p)\cost(p, \calS) \right) \right\vert \leq \eps \sum_{p \in P} \cost(p, \calS) + E.\]
\end{definition}

The offset $\Delta$ is often $0$ for most coreset constructions, with a few exceptions~\cite{CEMMP15,FeldmanSS20,SohlerW18}. In our algorithm, $\Delta=0$. The lower bounds hold for any choice of $\Delta$.

%% file: lower.tex
\section{Lower Bounds in Euclidean Spaces for $k$-Means}
\label{sec:lower}

We first prove the bound for $k$-means, i.e. for $z=2$. The
generalization to arbitrary powers is made in appendix: the proof idea is exactly alike, but a few new technicalities arise.

\subsection{$k$-Means}

As mentioned in the proof outline in Section~\ref{sec:overview}, we
proceed in three steps. First we show that any clustering of
$e_1,\dots,e_d$ using $k$ cluster centers of unit norm must have cost
at least $2d-O(\sqrt{dk})$. Next, we show that for any coreset $\Omega$ of $t$ points and
weights $w : \Omega \to \R^+$, there
is a clustering that has cost at most $2d - \Omega(\sqrt{k/t} \cdot
\sum_{p \in \Omega} w(p)\|p\|_2)$. Combined with step one, this
shows that $\sum_{p \in \Omega} w(p)\|p\|_2 = O(\sqrt{t/k}
\sqrt{dk}) = O(\sqrt{td})$. Finally we show that $\Omega$
must satisfy $\sum_{p \in \Omega} w(p)\|p\|_2 = \Omega(d)$ when $d =
\Theta(k \cdot \eps^{-2})$. Combining all of these implies $\sqrt{td}
= \Omega(d) \Rightarrow t = \Omega(d) = \Omega(k \cdot \eps^{-2})$.

For technical reasons, we consider the point set $e_1,\dots,e_d$ as
residing in $\R^{2d}$ and not $\R^d$. The reason for this, is that we
need to be able to find a vector that is orthogonal to all $e_i$ and
all points in a coreset $\Omega$ (see proof of Lemma~\ref{lem:FandSquare}). If the size of the coreset is $t <
d$, then such a vector exists in $\R^{2d}$.

\paragraph{Step One.}
We start by showing that any clustering of $e_1,\dots,e_d$ using $k$
centers of unit norm must have large cost:

\begin{lemma}
  \label{lem:standardhard}
For any $d$, consider the point set $P = \{e_1,\dots,e_d\}$ in
$\R^{2d}$. For any set of $k$ centers $c_1,\dots,c_k \in \R^{2d}$ with
unit norm, it holds that $\sum_{i=1}^d \min_{j=1}^k \|e_i - c_j \|_2^2 \geq 2d-2\sqrt{dk}$.
\end{lemma}
\begin{proof}
  We see that
  \begin{eqnarray*}
    \sum_{i=1}^d \min_{j=1}^k \|e_i - c_j \|_2^2 &=& \sum_{i=1}^d \min_{j=1}^k \|e_i\|_2^2 + \|c_j \|_2^2 -2\langle e_i, c_j \rangle \\
                                                 &=& 2d -2 \sum_{i=1}^d \max_{j=1}^k\langle e_i, c_j \rangle \\
                                                 &=& 2d -2 \sum_{j=1}^k \sum_{i :  j=\argmax_h \langle e_i, c_h \rangle } \langle e_i, c_j \rangle.
  \end{eqnarray*}
  Now, for each $c_j$, define $\hat{c}_j$ to equal $c_j$, except that we set the $i$'th coordinate to $0$ if $j \neq \argmax_h \langle e_i, c_h \rangle$. Then:
  \begin{eqnarray*}
    2d -2 \sum_{j=1}^k \sum_{i :  j=\argmax_h \langle e_i, c_h \rangle } \langle e_i, c_j \rangle 
    &=& 2d -2 \sum_{i=1}^d \sum_{j=1}^k \langle e_i, \hat{c}_j \rangle \\
    &=& 2d -2 \sum_{i=1}^d \langle e_i, \sum_{j=1}^k \hat{c}_j \rangle \\
    &\geq& 2d - 2\|\sum_{j=1}^k \hat{c}_j\|_1.
  \end{eqnarray*}
  By Cauchy-Schwartz, we have $\|\sum_{j=1}^k \hat{c}_j\|_1 \leq \| \sum_{j=1}^k \hat{c}_j\|_2 \cdot \sqrt{d}$. Since the $\hat{c}_j$'s are orthogonal and have norm at most $1$, we have $\|\sum_{j=1}^k \hat{c}_j \|_2 \leq \sqrt{k}$. Thus we conclude $\sum_{i=1}^d \min_{j=1}^k \|e_i - c_j \|_2^2 \geq 2d-2\sqrt{dk}$. 
\end{proof}

\paragraph{Step Two.}
Next, we show that for any coreset $\Omega$ of $t$ points and
weights $w : \Omega \to \R^+$, there
is a clustering that has cost at most $2d - \Omega(\sqrt{k/t} \cdot
\sum_{p \in \Omega} w(p)\|p\|_2)$. To prove this, we start by
considering the case of using a single cluster center to cluster
$\ell$ weighted points:

\begin{lemma}
  \label{lem:highips}
  Let $r_1,\dots,r_\ell \in \R^{2d}$ and let $w_1,\dots,w_\ell \in \R^+$. There exists a unit vector $v$ such that $\sum_{i=1}^\ell w_i |\langle r_i, v\rangle| \geq  \frac{\sum_{i=1}^t w_i \|r_i\|_2}{\sqrt{\ell}}$.
\end{lemma}

\begin{proof}
Consider the random vector $u = \sum_{i=1}^\ell w_i\sigma_i r_i$ where
the $\sigma_i$ are i.i.d. uniform Rademachers ($-1$ and $+1$ with
probability $1/2$).
  We see that
  \begin{eqnarray*}
    \sum_{i=1}^\ell w_i |\langle r_i, u\rangle| &=& \sum_{i=1}^\ell w_i \left| \sum_{j=1}^\ell w_j \sigma_j \langle r_i, r_j \rangle\right| \\
                                             &=& \sum_{i=1}^\ell w_i \left| \sum_{j=1}^\ell w_j \sigma_i \sigma_j \langle r_i, r_j \rangle\right| \\
                                             &\geq& \sum_{i=1}^\ell w_i \sum_{j=1}^\ell w_j \sigma_i \sigma_j \langle r_i, r_j \rangle \\
                                             &=& \|u\|_2^2.                                                
  \end{eqnarray*}
  We may then define the unit vector $v = u/\|u\|_2$ (with $v=0$ when $u=0$) and conclude that
$$
\sum_{i=1}^\ell w_i |\langle r_i, v \rangle| \geq \|u\|_2.
$$
Since $\E[\|u\|_2^2] = \sum_{i=1}^\ell w_i^2  \|r_i\|_2^2$ we conclude that there must exist a unit vector $v$ with
$$
\sum_{i=1}^\ell w_i |\langle r_i, v \rangle| \geq \sqrt{\sum_{i=1}^\ell w_i^2 \|r_i\|_2^2}.
$$
By Cauchy-Schwartz, we have:
$$
\sum_{i=1}^\ell | 1 \cdot w_i \|r_i\|_2 | \leq \sqrt{\sum_{i=1}^\ell w_i^2 \|r_i\|_2^2 } \cdot \sqrt{\sum_{i=1}^\ell 1 } = \sqrt{\sum_{i=1}^\ell w_i^2 \|r_i\|_2^2 } \cdot \sqrt{\ell}
$$
which finally implies
$$
\sum_{i=1}^\ell w_i |\langle r_i, v \rangle | \geq \frac{\sum_{i=1}^\ell w_i \|r_i\|_2}{\sqrt{\ell}}.
$$
\end{proof}

We can now extend this to using $k$ centers of unit norm to cluster
$t$ weighted points:

\begin{lemma}
  \label{lem:largereduction}
  Let $r_1,\dots,r_t \in \R^{2d}$ and let $w_1,\dots,w_t \in \R^+$. For any
  positive even integer $k$, there
  exists a set of $k$ unit vectors $v_1,\dots,v_k$ such that
  $\sum_{i=1}^t -2 w_i \max_{j=1}^k \langle r_i, v_j \rangle \leq
  -\sqrt{2k/t} \cdot \sum_{i=1}^t w_i \|r_i\|_2$ and moreover, for
  all $i$ we have $\max_{j=1}^k \langle r_i, v_j \rangle \geq 0$.
\end{lemma}

\begin{proof}
  Partition $r_1,\dots,r_t$ arbitrarily into $k/2$ disjoint groups $G_1,\dots,G_{k/2}$ of at most $2t/k$ vectors each. For each group $G_j$, apply Lemma~\ref{lem:highips} to find a unit vector $u_j$ with $\sum_{r_i \in G_j} w_i |\langle r_i, u_j \rangle| \geq \frac{\sum_{r_i \in G_j} w_i\|r_i\|_2}{\sqrt{2t/k}}$. Let $v_{2j-1} = u_j$ and $v_{2j}=-u_j$. Since we always add both $u_j$ and $-u_j$, it holds for all $r_i$ that $\max_{j=1}^k \langle r_i , v_j \rangle = \max_{j=1}^k |\langle r_i, v_j \rangle|$. We therefore conclude (notice the $\leq$ rather than $\geq$ due to the negation):
\begin{eqnarray*}
  \sum_{i=1}^t -2w_i \max_{j=1}^k \langle r_i, v_j \rangle &=& \sum_{i=1}^t -2w_i \max_{j=1}^k |\langle r_i, v_j \rangle| \\
                                                           &\leq& \sum_{j=1}^{k/2} \sum_{r_i \in G_j} -2w_i |\langle r_i, u_j\rangle| \\
                                                           &\leq& -2\sum_{j=1}^{k/2} \frac{\sum_{r_i \in G_j} w_i \|r_i\|_2}{\sqrt{2t/k}} \\
                                                           &=& -\frac{ \sqrt{2} \sum_{i=1}^t w_i \|r_i\|_2}{\sqrt{t/k}}.
                                                               \end{eqnarray*}
\end{proof}

With this established, we now combine this with step one to show that
for any coreset $\Omega$ with $t$ points, we must have $\sum_{p \in \Omega} w(p)\|p\|_2 = O(\sqrt{t/k}
\sqrt{dk}) = O(\sqrt{td})$. This is established in two smaller steps:

\begin{lemma}
  \label{lem:FandSquare}
  For any $d$, consider the point set $P=\{e_1,\dots,e_d\}$ in $\R^{2d}$. Let $r_1,\dots,r_t \in \R^{2d}$ and let $w_1,\dots,w_t \in \R^+$ be an $\eps$-coreset for $P$, using offset $\Delta$ and with $t < d$. Then we must have $\Delta + \sum_{i=1}^t w_i (\|r_i\|_2^2+1) \in (1 \pm \eps)2d$.
\end{lemma}

\begin{proof}
  Since $t+d < 2d$ there exists a unit vector $v$ that is orthogonal to all $r_i$ and all $e_j$.
Consider placing all $k$ centers at $v$. Then the cost of clustering
$P$ with these centers is $2d$. It therefore must hold that $\Delta +
\sum_{i=1}^t w_i (\|r_i\|_2^2 + \|v\|_2^2 - 2\langle r_i
,v\rangle) = \Delta+\sum_{i=1}^t w_i (\|r_i\|_2^2+1) \in (1\pm
\eps)2d$.
\end{proof}



\begin{lemma}
  \label{lem:uppernorm}
  For any $d$ and any $k > 1$, let $P=\{e_1,\dots,e_d\}$ in $\R^{2d}$. Let $r_1,\dots,r_t \in \R^{2d}$ and let $w_1,\dots,w_t \in \R^+$ be an $\eps$-coreset for $P$ with $t < d$, using offset $\Delta$. Then
  $$
  \sum_{i=1}^t w_i \|r_i\|_2 \leq \frac{4 \eps d + 2 \sqrt{dk}}{\sqrt{2k/t}}.
  $$
\end{lemma}

\begin{proof}
  By Lemma~\ref{lem:largereduction}, we can find $k$ unit vectors $v_1,\dots,v_k$ such that $\sum_{i=1}^t -2w_i \max_{j=1}^k \langle r_i, v_j \rangle \leq -\sqrt{2k/t} \cdot \sum_{i=1}^t w_i \|r_i\|_2$. By Lemma~\ref{lem:standardhard}, it holds that $\sum_{p \in P} \min_{j=1}^k \|p - v_j\|_2^2 \geq 2d-2\sqrt{dk}$. Since points $r_1,\dots,r_t$ with respective weights $w_1,\dots,w_t$ and offset $\Delta$ form an $\eps$-coreset for $P$, we must have 
\begin{eqnarray*}
  (1-\eps)(2d-2\sqrt{dk})&\leq& \Delta + \sum_{i=1}^t \min_{j=1}^k w_i \|r_i - v_j\|_2^2\\
                                &=& \Delta + \sum_{i=1}^t w_i (\|r_i\|_2^2 + \|v_j\|_2^2 - 2 \max_{j=1}^k \langle r_i,v_j \rangle) \\
                                &=& \Delta + \sum_{i=1}^t w_i (\|r_i\|_2^2 + 1)- 2 \sum_{i=1}^t w_i \max_{j=1}^k \langle r_i,v_j \rangle \\
                                &\leq& \Delta + \sum_{i=1}^t w_i (\|r_i\|_2^2 + 1) -\sqrt{2k/t} \cdot \sum_{i=1}^t w_i \|r_i\|_2.
\end{eqnarray*}
By Lemma~\ref{lem:FandSquare}, this is at most
\begin{eqnarray*}
  &\leq& (1+\eps)2d  -\sqrt{2k/t} \cdot \sum_{i=1}^t w_i \|r_i\|_2.
\end{eqnarray*}
We have therefore shown that
\begin{eqnarray*}
  (1-\eps)(2d-2\sqrt{dk}) &\leq& (1+\eps)2d  -\sqrt{2k/t} \cdot \sum_{i=1}^t w_i \|r_i\|_2 \Rightarrow \\
  \sqrt{2k/t} \cdot \sum_{i=1}^t w_i \|r_i\|_2 &\leq& (1+\eps)2d- (1-\eps)(2d-2\sqrt{dk}) \Rightarrow \\
  \sqrt{2k/t} \cdot \sum_{i=1}^t w_i \|r_i\|_2 &\leq& 4 \eps d + (1-\eps)2\sqrt{dk}\Rightarrow \\
 \sum_{i=1}^t w_i \|r_i\|_2 &\leq& \frac{4 \eps d + 2 \sqrt{dk}}{\sqrt{2k/t}}.
\end{eqnarray*}
\end{proof}

\paragraph{Step Three.}
Finally we show that any coreset $\Omega$
must satisfy $\sum_{p \in \Omega} w(p)\|p\|_2 = \Omega(d)$ when $d =
\Theta(k \cdot \eps^{-2})$:

\begin{lemma}
  \label{lem:lowernorm}
   For any $0 < \eps < 1/2$ and any positive even integer $k$, let $d
   = k/(36 \eps^2)$ and let $P=\{e_1,\dots,e_d\}$ in $\R^{2d}$. Let
   $r_1,\dots,r_t \in \R^{2d}$ and let $w_1,\dots,w_t \in \R^+$ be an
   $\eps$-coreset for $P$ with $t < d$, using offset $\Delta$. Then $\sum_{i=1}^t w_i
   \|r_i\|_2 \geq d/6$.
 \end{lemma}

 \begin{proof}
   Consider the Hadamard basis $h_1,\dots,h_q$ on $q = 1/(36 \eps^2)$
   coordinates, i.e. the set of rows in the normalized Hadamard
   matrix. This is a set of $q$ orthogonal unit vectors with all
   coordinates in $\{-1/\sqrt{q},1/\sqrt{q}\}$.  All $h_i$ except $h_1$ have equally
   many coordinates that are $-1/\sqrt{q}$ and $1/\sqrt{q}$ and $h_1$
   have all coordinates $1/\sqrt{q}$. Now partition the
   first $d$ coordinates into $k$ groups $G_1,\dots,G_{k}$ of $q$
   coordinates each. For any $h_i$, consider the $k$ centers
   $v^i_1,\dots,v^i_k$ obtained as follows: For each group $G_j$ of
   $q$ coordinates, copy $h_i$ into those coordinates to obtain the
   vector $v^i_j$. We must have that $\sum_{h=1}^d\min_{j=1}^k
     \|e_h - v^i_j\|_2^2 = \sum_{h=1}^d \min_{j=1}^k \|e_h\|_2^2 +
     \|v^i_j\|_2^2 - 2\langle e_h, v^i_j \rangle$. Since $k>1$, there
     is always a $j$ such that $\langle e_h, v^i_j\rangle =
     0$. Moreover, for $i=1$, we have $\max_{j=1}^k \langle e_h, v_j^i
     \rangle = 1/\sqrt{q}$ and for $i \neq 1$, it holds that precisely half
     of all $e_h$ have $\max_{j=1}^k \langle e_h, v_j^i \rangle =
     1/\sqrt{q}$. Thus we have $\sum_{h=1}^d \min_{j=1}^k \|e_h -
     v^i_j\|_2^2 \leq (d/2)2 + (d/2)(2-2/\sqrt{q}) = 2d -
     d/\sqrt{q}$. Thus: 
   \begin{eqnarray*}
     (1+\eps)(2d-d/\sqrt{q}) &\geq& \Delta + \sum_{h=1}^t w_h(\|r_h\|_2^2 + 1 -2 \max_{j=1}^{k} \langle r_h, v^i_j \rangle)
   \end{eqnarray*}
   By Lemma~\ref{lem:FandSquare}, this is at least
   \begin{eqnarray*}
     &\geq& (1-\eps)2d - 2 \sum_{h=1}^t w_h \max_{j=1}^{k} \langle
            r_h, v^i_j \rangle.
   \end{eqnarray*}
   We have thus shown
   \begin{eqnarray*}
      (1+\eps)(2d-d/\sqrt{q}) &\geq& (1-\eps)2d - 2 \sum_{h=1}^t w_h \max_{j=1}^{k} \langle
                                     r_h, v^i_j \rangle \Rightarrow \\
     4\eps d - (1+\eps) d/\sqrt{q} &\geq& - 2 \sum_{h=1}^t w_h \max_{j=1}^{k} \langle
                                          r_h, v^i_j \rangle \Rightarrow \\
    \sum_{h=1}^t w_h \max_{j=1}^{k} \langle
            r_h, v^i_j \rangle &\geq& (1+\eps) d/(2
                                      \sqrt{q}) -2\eps d \Rightarrow
     \\
     \sum_{h=1}^t w_h \max_{j=1}^{k} \langle
            r_h, v^i_j \rangle &\geq& d/(2
                                      \sqrt{q}) -2\eps d.
   \end{eqnarray*}
   Now consider any $r_h$ with weight $w_h$. Collect the vectors
   $u^i_{h}$ such that $u^i_h = v^i_{j^*}$ with $j^* = \argmax_j \langle r_h , v^i_j \rangle$. By construction, all these $q$ vectors are orthogonal (either disjoint support or distinct vectors from the Hadamard basis). By Cauchy-Schwartz, we then have $\langle w_h r_h, \sum_{i=1}^q u^i_h \rangle \leq w_h\|r_h\|_2 \|\sum_{i=1}^q u^i_h\|_2 = w_h\|r_h\|_2\sqrt{q}$. We then see that
   \begin{eqnarray*}
     dq/(2\sqrt{q}) -2\eps dq &\leq& \sum_{i=1}^q \sum_{h=1}^t w_h \max_{j=1}^{k} \langle r_h, v^i_j \rangle \\
                &=& \sum_{h=1}^t \sum_{i=1}^q w_h \langle r_h , u_h^i \rangle \\
                &=& \sum_{h=1}^t \langle  w_hr_h , \sum_{i=1}^q  u_h^i \rangle \\
     &\leq& \sum_{h=1}^t w_h \|r_h\|_2 \sqrt{q}.
   \end{eqnarray*}
   We have thus shown $\sum_{h=1}^t w_h \|r_h \|_2 \geq d/2 -2\eps d
   \sqrt{q} = d/2 -  2 \eps d/(6 \eps) = d/2 - d/3 = d/6$.
 \end{proof}

 \paragraph{Combining it All.}

 \begin{theorem}
   For any $0 < \eps < 1/2$ and any positive even integer $k$, let $d
   = k/(36 \eps^2)$ and let $P=\{e_1,\dots,e_d\}$ in $\R^{2d}$. Let
   $r_1,\dots,r_t \in \R^{2d}$ and let $w_1,\dots,w_t \in \R^+$ be an
   $\eps$-coreset for $P$, using offset $\Delta$. Then $t  \geq \eps^{-2} k
   /180$.
 \end{theorem}

 \begin{proof}
   If $t \geq d$, then we are done. Otherwise, we
   combine Lemma~\ref{lem:uppernorm} and Lemma~\ref{lem:lowernorm},
 to get:
\begin{eqnarray*}
  d/6 &\leq& \sum_{i=1}^t w_i \|r_i\|_2 \\
  &\leq& \frac{4 \eps d +
            2 \sqrt{dk}}{\sqrt{2k/t}} \\
      &=& \frac{4 \eps d + 6 \eps d}{\sqrt{2k/t}} \\
  &=& \frac{10 \eps d}{\sqrt{2k/t}}.
\end{eqnarray*}
This finally implies:
$$
t \geq \eps^{-2} k /180.
$$
 \end{proof}

%% file: lower_discrete.tex
\newcommand{\numpoints}{n_U}
\newcommand{\tc}{\tilde{c}}
\newcommand{\sizecenter}{|C|}

\section{Lower Bounds For Discrete Metrics}
\label{sec:lower-discrete}
We show in this section \cref{thm:discrete}, that we recall here for convenience:
\discrete*

To prove the theorem, we create a \emph{subinstance} that implies a lower bound for the case $k=1$.
The general lower bound for arbitrary $k$ then naturally combines several copies of the subinstance.
The key technical part of our proof is the use of some Azuma-Hoeffding type concentration inequality,
but where the concentration probability is \textit{lower bounded}. The results we use are developed
in \cref{sec:concentration-lb}. We present the subinstance in \cref{sec:subinstance},
and the general lower bound in \cref{sec:combine-sub}.

\subsection{Technical lemmas}\label{sec:concentration-lb}
Our proof relies on Lemma~\ref{lem:lb-proba}, which we prove using the following
result from~\cite{fan2015sharp}.
\begin{lemma}[Equation 2.11 in \cite{fan2015sharp}]\label{lem:sumtight}
Let $\xi_1, ..., \xi_m$ be independent centered random variables, and $\tilde \eps$ such that
\[\forall i, k \geq 3 ~ |\E[\xi_i^k]| \leq \frac{1}{2}k!\tilde {\eps}^{k-2}\E[\xi_i^2].\]
Let $\sigma^2 = \sum \E[\xi_i^2]$, and $S_m = \sum_{i=1}^m \xi_i$.

Then, for all $0 \leq x \leq 0.1 \frac{\sigma}{\tilde \eps}$, 
\[\Pr[S_m \geq x \sigma] \geq \left(1-\Phi\left(x(1 - c x \frac{\tilde \eps}{\sigma}\right)\right) \cdot \left(1 - c(1+x)\frac{\tilde \eps}{\sigma}\right),\]
where $c$ is an absolute positive constant and $\Phi$ is the standard normal distribution function.
\end{lemma}

\begin{lemma}\label{lem:lb-proba}
Let $X_1, ..., X_m$ be independent Bernouilli random variables with expectation $p\leq1/4$, $\eps > 0$ and $w_1, ..., w_m$ be some positive weights, such that $\max w_i \leq \gamma \cdot \frac{\sum w_i}{\eps m}$, for some $\gamma$. Let $\mu = p \cdot \sum w_i$. Then, there exists a constant $\beta$ such that
\[\Pr\left[\sum w_i X_i - \mu > \eps \mu\right] \geq \exp(-\frac{\beta}{\gamma^2} \eps^2 mp)\]
\end{lemma}
\begin{proof}
Define $\xi_i = w_i X_i - p w_i$. We show that the variables $\xi_i$ verify the conditions of \cref{lem:sumtight}. They are independent and centered, and:
\begin{align*}
\E[\xi_i^2] &=  p(w_i - pw_i)^2 + (1-p)(pw_i)^2) \\
&= w_i^2 \left(p - 2p^2 + p^3 + p^2 - 2p^3 + p^4 \right)\\
&= w_i^2 (p-p^2-p^3 + p^4) \geq \frac{w_i^2 p}{2},
\end{align*}
using $p \leq 1/4$. The $k$-th moment verifies:
\begin{align*}
\left|\E[\xi_i^k]\right| &= w_i^k \cdot \left(p\cdot (1 - p)^k + (1-p)\cdot (-p)^k\right) \leq w_i^k p,
\end{align*}
hence $\xi_i$ verifies the condition of \cref{lem:sumtight} with $\tilde \eps = \max_i w_i$. We want to apply that lemma to $x$ of the order $\eps \frac{\mu}{\sigma}$: therefore, we need to bound that quantity. Note that 
\begin{equation}
\sigma^2 \geq \frac{p}{2} \sum w_i^2 \geq \frac{p}{2} \cdot \frac{(\sum w_i)^2}{ m},
\end{equation}
 and so by the assumptions of the lemma $\frac{\sigma}{\tilde \eps} \geq \frac{ \eps \sqrt {mp}}{\gamma \sqrt 2}$. Furthermore, 
 \begin{equation}\label{eq:musigma}
 \frac{\mu}{\sigma} \leq \frac{p \sum w_i}{\sqrt{\frac{p}{2m}} \sum w_i} \leq \sqrt{2mp}
 \end{equation}

Now, let $x := \frac{\eps}{10 \gamma c \sqrt 2} \cdot \frac{\mu}{\sigma}$. Thus, $x$ verifies $x  \leq \frac{\eps}{10 \gamma c \sqrt 2} \cdot \sqrt{2pm} \leq 0.1 \frac{\sigma}{c \tilde \eps}$ and so applying \cref{lem:sumtight} we obtain:
\begin{align*}
\Pr\left[\sum w_i X_i - \mu > \eps \mu\right]&\geq \left(1-\Phi\left(x(1 - c x \frac{\tilde \eps}{\sigma})\right)\right) \cdot \left(1 - c(1+x)\frac{\tilde \eps}{\sigma}\right)\\
 &\geq  \left(1-\Phi\left(0.9x\right)\right) \cdot 0.9\\
 &= 0.9 \cdot \Pr[\mathcal{N}(0, 1) \geq 0.9x]\\
 &\geq 0.9 \cdot \frac{1}{2}\left(1-\sqrt{1-e^{-(0.9x)^2}}\right)\\
 &\geq \exp(-\frac{\beta}{\gamma^2} \eps^2 \frac{\mu^2}{\sigma^2})\\
 &\geq \exp(-\frac{\beta}{\gamma^2} \eps^2 mp),
\end{align*}
where $\beta$ is some absolute constant, and where the last line uses \cref{eq:musigma}.
\end{proof}


\subsection{A subinstance for the case $k=1$}\label{sec:subinstance}
We now turn to proving a lower bound for the case where $k=1$. This is going
to be our building block in the next subsection where we generalize the result
to arbitrary $k$. Let $\delta = 1/4$ be a parameter.

\begin{definition}
A subinstance $U_\delta$ is defined as follows.
Let $C$ be a set of $n$ \emph{candidate centers} and $P$ a set of of $\numpoints$ \emph{clients}. The 
metric on the ground set $P \cup C$ is defined according to the following probability distribution.

For each pair $(p, c) \in P \times C$, 
\begin{equation}
\dist(p, c) =  \begin{cases}
 1 \text{ with probability } \delta\\
 2^{1/z} \text{ otherwise}
\end{cases}
\end{equation}
Distances between any pair of points $p,p' \in P$ or $c, c' \in C$ is set to $2^{1/z}$. 

\cref{fig:subinstance} illustrates the definition.
\end{definition}

\begin{figure}
\centering
\includegraphics[scale=0.9]{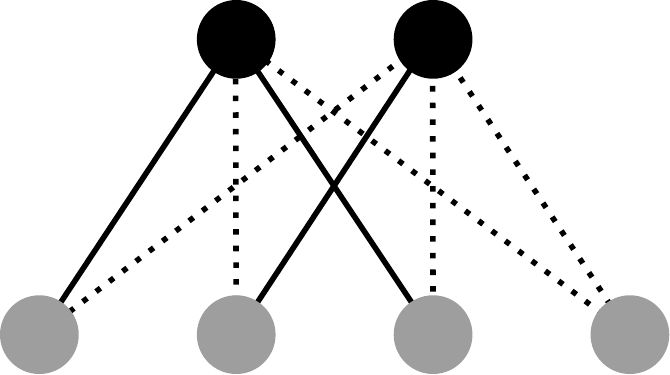}
\caption{Illustration of an instance $U_\delta$. Dashed edges have length $2^{1/z}$, black ones have length $1$.}
\label{fig:subinstance}
\end{figure}

Since any complete graph with edge length only $1$ or $\ell \leq 2$ defines a metric space, it immediately follows that $(P\cup C, \dist)$ is a metric space. 

The important properties of the subinstance are summarized in the following lemma. We say that a set of weights is \emph{$\eps$-rounded} if all weights are multiples of $\eps$.

\begin{lemma}
  \label{lem:subinstance}
There exists a constant $\eta$ and an instance $U_{\delta} = (P, C, \dist)$ with $|P| = \numpoints =  \eps^{-2}\log |C|$, $\delta \leq 1/4$ and $|C| \geq \eps^{-5}$, the following holds. For any subset $\Omega \subseteq P$ with $\eps/2$-rounded weights $w_x$ being such that $\sum w_x \in (1\pm 1/2) \numpoints$,
  we have:
  \begin{enumerate}
  \item If $|\Omega| < \eps^{-2} \eta \log |C| $, there exists a center
    $\tc \in C$ such that
    \[\sum_{x \in \Omega:\dist(x, \tc) = 1}w_x > (1+200\eps)\delta \numpoints\] and  $|x\in P:~\dist(x, \tc) = 1| \geq \delta \numpoints$ 
  \item If $|\Omega| \ge \eps^{-2} \eta  \log |C| $, there exists a  center $c^* \in C$ such that
    \[\sum_{x \in \Omega:\dist(x, c^*) = 1}w_x \geq (1-\eps)\delta \numpoints\]
     and $|x\in P:~\dist(x, c^*) = 1| \geq \delta \numpoints$.
  \end{enumerate}
\end{lemma}
\begin{proof}
  We use the probabilistic method: we will show that, when $U_\delta$ is generated according to the process defined above, the two properties of the lemma
  hold with some positive probability. This is enough to ensure the existence of an instance $U_\delta$ verifying them.

  We start by proving the first item. Fix some arbitrary subset of clients $\Omega$ of size at most $\eps^{-2} \eta \log |C| $, with weight $w_x, \forall x \in \Omega$ and a candidate center $c \in C$. Let $w_1(c, \Omega) := \sum_{x \in \Omega:\dist(x, c) = 1}w_x$ denote the (weighted) number of edges of length 1 from $\Omega$ to $c$. The expected value of $w_1(c, \Omega)$ over the random choice of edges is $\delta \cdot \numpoints$. We aim at applying \cref{lem:lb-proba} on the variable $w_1(c, \Omega)$. This cannot be done directly, as we have no control on $\max w_x$. Hence, we partition the points of $\Omega$ into five groups:
  \begin{itemize}
  \item $\Omega_1 := \{x\in \Omega:~w_x < \eps\}$
  \item $\Omega_2 := \{x\in \Omega:~w_x \in [\eps, 1)\}$
  \item $\Omega_3 := \{x\in \Omega:~w_x \in [1, \eps^{-1})\}$
  \item $\Omega_4 := \{x\in \Omega:~w_x \in [\eps^{-1}, 10  \log (1/\delta) \cdot \eps^{-2})\}$
  \item $\Omega_5 := \{x\in \Omega:~w_x \geq 20  \log (1/\delta) \cdot \eps^{-2}\}$
  \end{itemize}

  We will show that, $\forall i \in \{2, ..., 5\}$, $w_1(c, \Omega_i)$ exceeds its expectation by a factor $(1+205\eps)$
  with large probability, and that $w_1(c, \Omega_1)$ is negligible.

First, note that since $\sum_{x \in \Omega} w_x \leq (1+1/2)\numpoints = \frac{3}{2} \eps^{-2} \log |C|$, it must be that 
\[|\Omega_5| \leq \frac{\log |C|}{10 \log(1/\delta)}.\]
Hence, $c$ is connected with length $1$ to all points of $\Omega_5$ with probability $\delta^{|\Omega_5|} \geq \exp\left(-\log|C| / 10\right) = |C|^{-1/10}$.

Now, on each group $\Omega_2,\Omega_3,\Omega_4$, the maximum weight cannot be more than $20  \log (1/\delta) \cdot\eps^{-1}$ times the average.
  
For $i \in \{2,3,4\}$, $w_1(c, \Omega_i)$ is the sum of $m=|\Omega_i| \leq \eps^{-2} \eta \log |C| $ random variables $X_{x}$, for $x\in \Omega_i$, with $X_x = 0$ with probability $(1-\delta)$ and $X_x = w_x$ with probability $\delta = 1/4$. Hence, \cref{lem:lb-proba} gives that:
  
\begin{align*}
\Pr\left[w_1(c, \Omega_i) \geq (1+205\eps)\cdot \E[w_1(c, \Omega_i)] \right] &>  \exp(-\frac{\beta}{\log(1/\delta)^2} \eps^2 \delta |\Omega_i|)\\
&\geq \exp(- \log |C| / 10),
\end{align*}
for some absolute constant $\beta$ given by \cref{lem:lb-proba} and $\eta \le \frac{ \log(1/\delta)^2}{10 \beta \delta}$.

Finally, to deal with $\Omega_1$, we note that $\E[w_1(c, \Omega_1)] \leq \eps \delta \numpoints$. Hence, $\sum_{i=2}^5 \E[w_1(c, \Omega_i)] \geq \E[w_1(c, \Omega)] - \eps \delta \numpoints$, and 
\begin{eqnarray*}
\sum_{i=2}^5 w_1(c, \Omega_i) \geq (1+205\eps) \sum_{i=2}^5 \E[w_1(c, \Omega_i)] & \Rightarrow & w_1(c, \Omega) \geq (1+200\eps) \delta \numpoints.
\end{eqnarray*}

Since all groups are disjoint, the variables $w_1(c, \Omega_i)$ are independents and we can combine the previous equations to get:
\begin{align*}
\Pr\left[\sum_{x \in \Omega:\dist(x, \tc) = 1}w_x \geq (1+200\eps)\cdot \delta \numpoints \right] &>  |C|^{-3/ 10}.
\end{align*}

Since the length of the edges are chosen independently, the probability that there exists no center $\tc$ with $\sum_{x \in \Omega:\dist(x, \tc) = 1}w_x \geq (1+200\eps)\cdot \delta \numpoints$ is at most
\begin{align*}
\left(1-|C|^{-3/10}\right)^{|C|} &= \exp\left(|C| \log (1-|C|^{-3/10}) \right)\\
&\leq \exp(-|C|^{7/10}).
\end{align*}
And hence with probability at least $1-\exp(-|C|^{7/10})$ there is a center $\tc$ with $\sum_{x \in \Omega:\dist(x, \tc) = 1}w_x \geq (1+200\eps)\cdot \delta \numpoints$. 

To conclude the proof of the first bullet, it remains to do a union-bound over all possible weighted subset $\Omega$. Such an $\Omega$ consists of at most $\numpoints$ different points, with $\eps/2$-rounded weights in $[0, (1+1/2)\numpoints]$. Hence, there are at most $\frac{4}{\eps}\numpoints$ many
different weights.

Therefore, there are $\left(\frac{4\numpoints}{\eps}\right)^{\numpoints}$ many possible weighted subset $\Omega$ with $\eps/2$-rounded weights, i.e., 
\[\exp\left(\eps^{-2}\log |C| \cdot \log\left(2\eps^{-3}\log |C|\right)\right).\]

We can conclude that there exists a center $\tc \in C$ with $\sum_{x \in \Omega:\dist(x, \tc) = 1}w_x \geq (1+200\eps)\cdot \delta \numpoints$
with probability at least 
\begin{align*}
1 - \exp\left(\eps^{-2}\log |C| \cdot \log\left(2\eps^{-3}\log |C|\right)\right) \cdot \exp(-|C|^{7/10}) & \geq 
\frac{99}{100}
\end{align*}
by our choice of $|C|$. Furthermore, $\Pr[|x\in P:~\dist(x, \tc) = 1| \geq \delta \numpoints] \geq 1/2$, because $|x\in P:~\dist(x, \tc) = 1|$ follows a binomial law with mean $\delta \numpoints$. This concludes the proof of the first bullet.

We now turn to the second bullet of the claim, for which the proof is a more standard application of Azuma inequality. Fix some coreset $\Omega$ of size at least $\eps^{-2} \eta \log |C| $, and a center $c$. 
We have,
\begin{align*}
\Pr\left[w_1(c, \Omega) \notin (1\pm \eps)\cdot \delta \numpoints \right] &\leq \exp(-2 \eps^2 \delta^2 \frac{\numpoints^2}{\sum w_i^2})\\
&\leq \exp(- 2 /4 \cdot \delta^2 \eps^2 )\\
&\leq \exp(-1/2\cdot \delta^2\eps^2),
\end{align*}
where the second inequality uses $\numpoints^2 \geq 1/4\left(\sum w_i\right)^2 \geq 1/4\cdot \sum w_i^2$.
 
Since those events are independent for different centers $c$, the probability that there exists no center $c \in C$ with $w_1(c, \Omega) \in (1\pm\eps)\cdot \delta \numpoints$ is at most $\exp(-1/2\cdot \delta^2\eps^2|C|)$.

Hence, a union-bound over the $\left(\frac{4\numpoints}{\eps}\right)^{\numpoints}$ many possible weighted subset $\Omega$
 ensures that the following holds with probability at most $1-\left(\frac{4\numpoints}{\eps}\right)^{\numpoints} \cdot \exp(-1/2\cdot \delta^2\eps^2|C|) \geq 99/100$: For any $\Omega$  there exists a center $c$ with $w_1(c, \Omega) \in (1\pm\eps)\cdot \delta |\Omega|$ as desired.
\end{proof}

\subsection{Combining the subinstances}\label{sec:combine-sub}

\begin{figure}[H]
\centering
\includegraphics[scale=0.8]{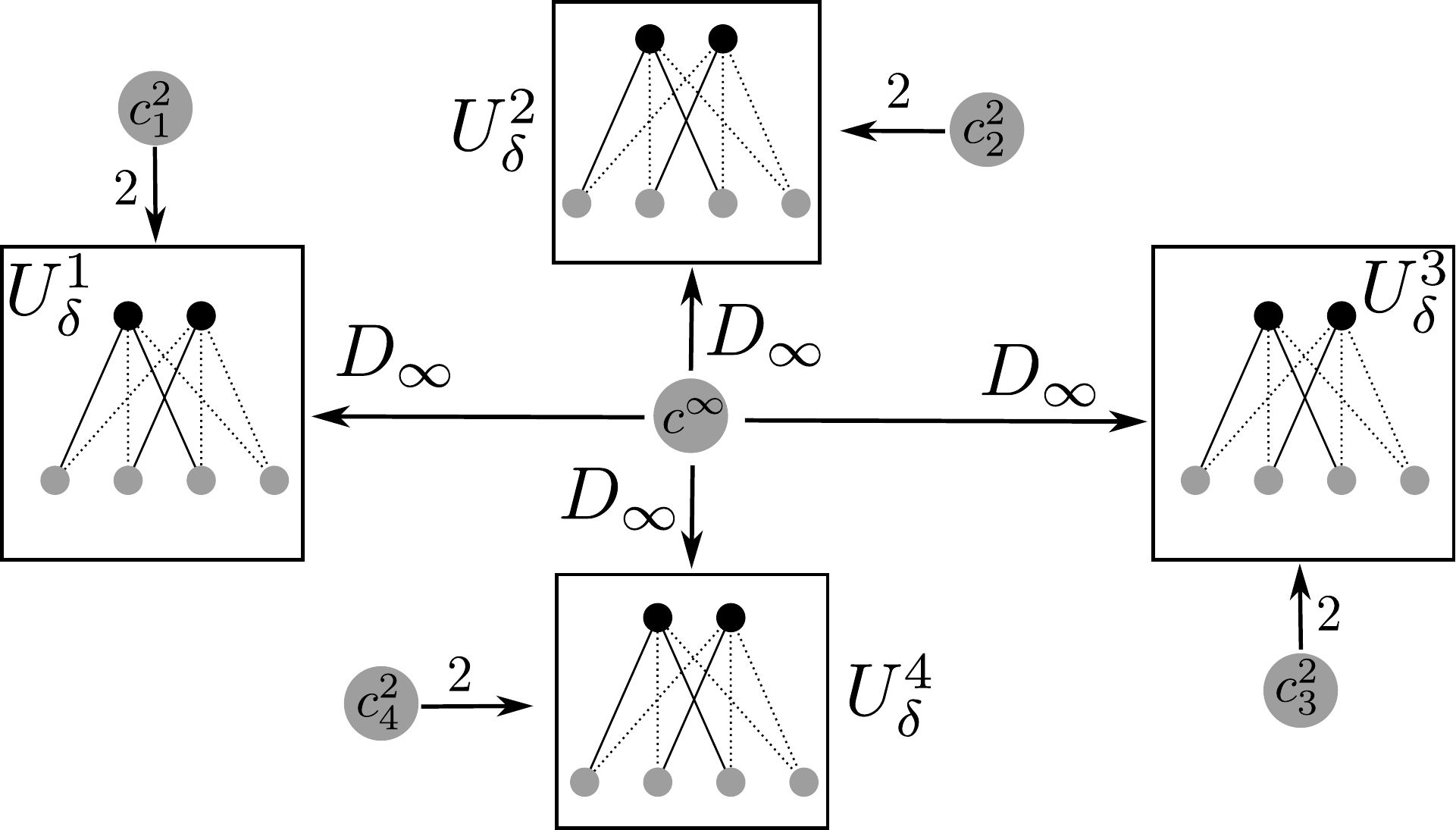}
\caption{Illustration of a full instance, in the case $z=1$. The subinstance are inside squares, and there is an edge from a node to a square when the node is linked to every point of the subinstance, with the distance written on the edge. $D_\infty$ is set to be $4^i \cdot \frac{\numpoints \cdot k}{\eps}$. The node $c^{4\cdot \infty}$ is not represented.}
\label{fig:full}
\end{figure}

We now conclude the proof of the lower bound for the $(\eps, k, z)$ coreset using offset $\Delta$.
We consider $k$ copies of the subinstance given by \cref{lem:subinstance}, $U_{\delta}^1,\ldots,U_{\delta}^k$, where the set of clients in each subinstance has size $\numpoints = 10 \eps^{-2} \log n$, and the set of candidate centers has size $\sizecenter$ such that $\sizecenter \geq \eps^{-5}$ and $\log\left(\sizecenter \cdot k\right) = O(\log \sizecenter) $.
In total, there are $k \sizecenter$ many candidate centers, and $k \numpoints$ many different clients. The subinstances are numbered from $1$ to $k$, and connected together in a star-graph metric centered at an arbitrary point $c^\infty$, where all points are at distance $\frac{\numpoints \cdot k}{\eps}$ of $c^\infty$. There is some additional candidate centers:  $c^{4\cdot \infty}$, at distance $4 \cdot \frac{\numpoints \cdot k}{\eps}$ of every client, and for subinstance $i$ there is a center $c^2_i$, at distance $2^{1/z}$ from every client of the subinstance.
\cref{fig:full} illustrates that construction.

We can now turn to the proof of the theorem. For this, we start with three claims: The first one shows that the total weight of the coreset must be very close to the number of point in the instance. The second shows that the offset $\Delta$ must be negligible, and the third that the coreset weight in each subinstance is close to $\numpoints$, the number of point in a subinstance.

\begin{claim}\label{claim:totweight}
If $\Omega$ is an $\eps$-coreset with offset $\Delta$ for the instance, then the total weight verifies $w(\Omega) \in (1\pm 2\eps) k\numpoints$.
\end{claim}
\begin{proof}
Consider the solution consisting only of one center placed at $c^\infty$. Let $D_\infty = \frac{\numpoints \cdot k}{\eps}$. This solution has cost $\cost(c^\infty) = k \numpoints \cdot D_\infty ^z$, and $\cost(\Omega, c^\infty) = w(\Omega) \cdot D_\infty ^z$. 
Hence, 
\[\Delta + w(\Omega)\cdot D_\infty^z \in (1\pm \eps)k \numpoints \cdot D_\infty^z.\]

Similarly, considering the solution that places only one center at $c^{4\cdot \infty}$ gives 

\[\Delta + w(\Omega) 4^z D_\infty^z \in (1\pm \eps)k \numpoints \cdot 4^z D_\infty^z.\]

Substracting those two equations yields:
\begin{gather*}
(4^z -1) w(\Omega)\cdot D_\infty^z \in ((4^z -1)  \pm (4^z+1)\eps) k \numpoints \cdot D_\infty^z,
\end{gather*}
and so $w(\Omega) \in (1\pm 2\eps) k \numpoints$.
\end{proof}
\begin{claim}\label{claim:delta}
If $\Omega$ is an $\eps$-coreset with offset $\Delta$ for the instance, then $|\Delta| \leq 3\eps k\cdot \numpoints$.
\end{claim}
\begin{proof}
Consider the solution $\calS^2 = \{c^2_i,~\forall i\}$. We have $\cost(\calS^2) = 2k\numpoints$ and $\cost(\Omega, \calS^2) = 2w(\Omega)\in (1\pm2\eps)\cost(\calS^2)$, using \cref{claim:totweight}. Since $|\Delta + \cost(\Omega, \calS^2) - \cost(\calS^2)| \leq \eps \cost(\calS^2)$, it must be that $|\Delta| \leq 3\eps \cost(\calS^2) = 3\eps k \numpoints$.
\end{proof}

\begin{claim}\label{claim:weight-sub}
If $\Omega$ is an $\eps$-coreset with offset $\Delta$ for the instance, then in every subinstance, the sum of the coreset weights is in $(1\pm 1/2)\numpoints$.
\end{claim}
\begin{proof}
Assume towards contradiction that, in some subinstance, say subinstance $i$, the coreset mass is not in $(1\pm 1/2)\numpoints$, and consider a solution $\calS$ that places one center in each subinstance but subinstance $i$. Suppose w.l.o.g. that the subinstance is overweighted: the coreset places a total weight larger than $3/2\cdot \numpoints$ in it. The cost of the solution is a most 
\begin{align*}
\cost(\calS) &\leq \underbrace{2 (k-1)\numpoints
}_{\text{for subinstances that contain a center}} + \underbrace{\numpoints \cdot \left(k\numpoints \eps^{-1}\right)^z}_{\text{for the overweighted subinstance}}\\
&\leq (1+\eps)\left(k \cdot \numpoints^2 \eps^{-1}\right)^z,
\end{align*}
 while the cost in the coreset verifies
\begin{align*}
\Delta + \cost(\Omega, \calS) &> -3\eps k \numpoints + 3/2\cdot \numpoints \cdot \left(k\numpoints\eps^{-1}\right)^z\\
& \qquad \text{ (using \cref{claim:delta} and keeping only the cost of the overweighted subinstance)}\\
&> (1+\eps)\cdot (1+\eps)\left(k \cdot \numpoints^2 \eps^{-1}\right)^z\\
&> (1+\eps)\cost(\calS),
\end{align*}
hence contradicting the fact that $\Omega$ is an $\eps$-coreset with offset $\Delta$.

The proof of the case where some subinstance is underweighted is done exactly alike.
\end{proof}

We can now turn to the proof of the theorem.
\begin{proof}[Proof of \cref{thm:discrete}]
Assume toward contradiction that there exists an $\eps$-coreset with offset $\Delta$ of
size smaller than $\frac{\eta}{10} \cdot k \eps^{-2} \log |C|$, where $\eta$ is the constant of \cref{lem:subinstance}. 

First, this implies the existence of an $2\eps$-coreset with $\eps$-rounded weights, simply by rounding each weight to the closest multiple of $\eps$. 

Using \cref{claim:weight-sub}, we can apply \cref{lem:subinstance} on each subinstance.  The total coreset size is $\frac{\eta}{10} \cdot k \eps^{-2} \log |C|$: that means that there are at least $k/10$
subinstances for which the coreset contains no more than $\eta \eps^{-2} \log |C|$ many different points. We refer to these subinstances
as the \textit{bad} subinstances.
Using \cref{lem:subinstance}, we construct a solution $\calS$ by taking the center given by bullet 1 for the bad subinstances, i.e.: center $\hat{c}$ as per the notation of \Cref{lem:subinstance}, and bullet 2 for the others, i.e.:
center $c^*$ as per the notation of \Cref{lem:subinstance}. The cost of that solution is $n_1 + 2 (k\numpoints - n_1) = 2k\numpoints - n_1$, where $n_1$ the number of edges of length $1$ from the clients to $\calS$. Similarly, the cost of $\calS$ for the coreset is $2\cdot w(\Omega) - w_1(\calS, \Omega)$, where $w(\Omega)$ is the total coreset weight and $w_1(\calS, \Omega)$ the weighted number of length 1 edges from $\Omega$ to $\calS$. By construction of $\calS$, $w_1(\calS, \Omega)$ verifies

\begin{gather*}
w_1(\calS, \Omega) \geq k/10 \cdot (1+200\eps)\delta \numpoints + 9k/10 \cdot (1-\eps) \delta \numpoints
> (1+19\eps) \delta \cdot k\numpoints
\end{gather*}

Furthermore, using properties of \cref{lem:subinstance},
$n_1 \leq  \delta k\numpoints$. Hence, the cost of $\calS$ in the coreset satisfies 
\begin{align*}
\Delta + 2\cdot w(\Omega) -  w_1(\calS, \Omega) &< 3\eps k \numpoints + 2\cdot (1+2\eps)k\numpoints - (1+19\eps) \delta \cdot k\numpoints\\
& \leq  (2k\numpoints -  n_1) +\eps k \numpoints \cdot(7-38\delta) < (1-\eps)(2k|P| - n_1),
\end{align*}
where the last inequality uses $\delta = 1/4$, so that $(38 \delta - 7)k \numpoints \geq 2 k \numpoints$. Therefore the cost of the coreset for $\calS$ is smaller than a $(1-\eps)$ factor times
the cost of $P$ for $\calS$, a contradiction that concludes the proof.
\end{proof}

A simple corollary of that proof is a lower bound for metric with bounded doubling dimension. Since any $n$ points metric has doubling dimension $O(\log n)$, the metric constructed has doubling dimension $D = O(\log n) = O(\log |C|)$, which implies \cref{cor:discrete}.

%% file: chaining.tex
\section{Algorithm}
\label{sec:upper}

Throughout this section, we use the following notation.
We use $\|P\|_0$ to denote the distinct number of points in $P$.
For a solution $\calS$, we define the $|P|$ dimensional cost vector $v^{\calS}$ induced by $\calS$ as
$$ v^{\calS}_p = \cost(p,\calS).$$
Hence, $\|v^{\calS}\|_1=\cost(P,\calS)$. 


We will also make use the following lemma to have a weaker version of the triangle inequality for $k$-Means and more general powers of distances. See Appendix A from Makarychev, Makarychev, and Razenshteyn~\cite{MakarychevMR19} for a proof.

\begin{lemma}[Triangle Inequality for Powers]
\label{lem:weaktri}
Let $a,b,c$ be an arbitrary set of points in a metric space with distance function $d$ and let $z$ be a positive integer. Then for any $\varepsilon>0$
\begin{align*}
d(a,b)^z &\leq (1+\varepsilon)^{z-1} d(a,c)^z + \left(\frac{1+\varepsilon}{\varepsilon}\right)^{z-1} d(b,c)^z\\
\left\vert d(a,b)^z - d(a,c)^z\right\vert &\leq \varepsilon \cdot d(a,c)^z + \left(\frac{z+\varepsilon}{\varepsilon}\right)^{z-1} d(b,c)^z.
\end{align*}
\end{lemma}

We also require Bernstein's inequality:
\begin{theorem}[Bernstein's Inequality]
\label{thm:bernstein}
Let $X_1,\ldots X_{\delta}$ be non-negative independent random variables. Let $S=\sum_{i=1}^{\delta} X_i$.
If there exists an almost-sure upper bound $M\geq X_i$, then 
$$\pr\left[\left\vert S - \mathbb{E}[S] \right\vert \geq t\right] \leq \exp\left(- \frac{t^2}{2\sum_{i=1}^{\delta} \text{Var}[X_i] + \frac{2}{3}\cdot M\cdot t}\right).$$
\end{theorem}

\subsection{Preprocessing and General Outline}

We make the following three assumptions:
\begin{description}
\item[Assumption 1] The number of distinct points $\|P\|_0$ is in $\text{poly}(k/\varepsilon)$.
\item[Assumption 2] The dimension $d$ of the points is in $O(\varepsilon^{-2}\log \|P\|_0) = O\left(\varepsilon^{-2} \log \frac{k}{\varepsilon} \cdot \text{poly }z\right)$.
\item[Assumption 3] The point set is unweighted.
\end{description}
Assuming these simplifies the presentation significantly. The first assumption can be justified by computing a (potentially weighted) coreset in preprocessing. Coresets of size $\tilde{O}(k^2\cdot \varepsilon^{-4}\cdot 2^{O(z)})$ are known to exists for all $(k,z)$ clustering objectives~\cite{CSS21}, which is sufficient for our purposes. 

The second assumption follows from a result on terminal embeddings due to Narayanan and Nelson~\cite{NaN18}. We will discuss this result in more detail in Section~\ref{sec:nets}. Suffice to say here is that there exists a coreset-preserving embedding from an arbitrary dimension to the desired target dimension. 

The final assumption follows by scaling the weights and rounding them to integers. Each weight is then treated as a multiplicity of a point. Note that this does not increase the distinct number of points. For a proof of the validity of such an operation, we refer to Corollary 2.3~\cite{CSS21}.

We now describe the algorithm.
We first compute some constant factor approximation $\greedy$ for the entire instance.\footnote{A bicriteria approximation that uses $O(k)$ centers and yields a constant factor approximation would also be possible. See~\cite{MakarychevMSW16} for state of the art bounds on bicriteria approximations for $k$-median and $k$-means. For higher powers, see Mettu and Plaxton~\cite{MettuP04} for a $2^{O(z)}$ approximation.}
Let $C_i$ be the $i$th cluster induced by $\greedy$. The average cost of $C_i$ is $\Delta_{C_i} = \frac{\cost(C_i, \greedy)}{|C_i|}$. For all $i, j$, the \textit{ring} $R_{i,j}$ is the set of points $p \in C_i$ such that $2^j \Delta_{C_i} \leq \cost(p, \greedy) \leq 2^{j+1} \Delta_{C_i}.$
The \textit{inner rings} $R_I(C_i) := \cup_{j \leq z\log(\eps/z)} R_{i,j}$ (resp. \textit{outer rings} $R_O(C_i) := \cup_{j > 2z\log(z/\eps)}R_{i,j}$) of a cluster $C_i$ consists of the points of $C_i$ with cost at most $\left(\nicefrac \eps z\right)^z \Delta_{C_i}$ and resp. at least $\left(\nicefrac z \eps \right)^{2z} \Delta_{C_i}$. The \textit{main rings} $R_M(C_i)$ consists of all the other points of $C_i$. For each $j$, $R_j$ is defined to be $\cup_{i=1}^{k} R_{i,j}$.
We then partition the input point set into the following groups.

\begin{itemize}
\item For each $j$, the rings $R_{i,j}$ are gathered into \textit{groups} $G^M_{j,b}$:
\begin{align*}
G^M_{j, b} := \Big\lbrace p ~\vert~  \exists i,~p\in R_{i,j} \text{ and }& \\
\left(\frac{\varepsilon}{4z}\right)^{z}\cdot    \frac{\cost(R_j,\greedy)}{k} \cdot 2^b  &\leq \cost(R_{i,j},\greedy) \leq \left(\frac{\varepsilon}{4z}\right)^{z}  \cdot 2^{b+1}\cdot \frac{\cost(R_j,\greedy)}{k}\Big\rbrace.
\end{align*}
\item  For any $j$,  let $G^M_{j, min} := \cup_{b \leq 0} G^M_{j, b}$ be the union of the cheapest groups, and $G^M_{j, max} := \cup_{b \geq z\log{\frac{4z}{\eps}}} G^M_{j, b}$ be the union of the most expensive ones. We define $G^M := \bigcup_j G^M_{j, max} \cup \bigcup_b G^M_{j,b} \setminus G^M_{j, min}$.
\item The points in the outer rings are also partitioned into \emph{outer groups}:
\begin{align*}
G^O_{b} = \Big\lbrace p ~\mid ~  \exists i,~ p\in C_i \text{ and }& \\
\left(\frac{\varepsilon}{4z}\right)^{z}\cdot    \frac{\cost(\out^\greedy,\greedy)}{k} \cdot 2^b  &\leq \cost(\out(C_i),\greedy)  \leq \left(\frac{\varepsilon}{4z}\right)^{z}  \cdot 2^{b+1}\cdot \frac{\cost(\out^\greedy,\greedy)}{k}\Big\rbrace.
\end{align*}
We denote by $P^{G^O_b}:=\{p\in P~|~ p\in C \wedge C\cap G^O_b \neq \emptyset\}$ the set all points in clusters intersecting with $G$.
\item We let as well $G^O_{min} = \cup_{b \leq 0} G^O_b$ and $G^O_{max} = \cup_{b \geq z\log{\frac{4z}{\eps}}} G^O_b$. We define $G^O :=  G^O_{max} \cup \bigcup_b G^O_{b} \setminus G^O_{min}$. 
\end{itemize}

The set of all groups is denoted by
$\mathcal{G} := G^M \cup G^O$. We sometimes abuse notation and also use $\mathcal{G}$ to denote the set of points in the groups $G^M\cup G^O$, i.e. $P \cap \mathcal{G} = \{p\in P~|~p\in G\in\mathcal{G}\}$ and $P \setminus \mathcal{G} = \{p\in P~|~p\in G\notin\mathcal{G}\}$.
We summarize the group partitioning scheme with the following two facts.
\begin{fact}
\label{fact:groupnumber}
There exist at most $O(z^2\log^2 z/\varepsilon)$ groups in $\mathcal{G}$.
\end{fact}
\begin{fact}
\label{fact:groupdisj}
All groups are pairwise disjoint. Moreover, every cluster $C$ induced by $\greedy$ intersects with at most one group $G\in G^O$.
\end{fact}

The final algorithm now consists of sensitivity sampling for all groups $G\in \mathcal{G}$.
Specifically, we pick a point $p\in G$ with probability $\frac{\cost(p,\greedy)}{\cost(G,\greedy)}$. We repeat this $\delta$ times, where $\delta$ is the size of the desired coreset. For each picked point $p$, we set the weight equal to $w_p:=\frac{\cost(G,\greedy)}{\delta\cdot \cost(p,\greedy)}$. 
For every cluster $C_i$, we weigh the center $c_i$ with the number of points in $C_i\cap \left(\bigcup_j G_{j,\min} \cup G^{O}_{\min}\right)$.
The entire coreset construction then consists of steps required to satisfy the three initial assumptions followed by the sampling procedure, see Algorithm~\ref{alg}.

\begin{algorithm}[tb]
\begin{algorithmic}
\STATE Compute a $O(2^z)$ approximation $\greedy$ to $P$.\\
\STATE Preprocess the instance such that Assumptions 1-3 hold.\\
\STATE Partition the points into groups $\mathcal{G} = \left(\bigcup_{j}G_{j,\max}\cup \bigcup_{b}G(j,b)\setminus G_{j,\min}\right) \cup \left(G^O_{\max} \cup \bigcup_{b} G^O_{b}\setminus G^O_{\min}\right)$. \\
\FORALL{Groups $G\in\mathcal{G}$} 
\STATE Sample $\delta\in k\cdot \log \frac{k}{\varepsilon} \cdot \varepsilon^{-2}\cdot 2^{O(z\log (1+z))}\cdot \log^{3} \varepsilon^{-1}\cdot \min(\varepsilon^{-z},k)$ points $\coreset_G$ proportionate to $ \frac{\cost(p,\greedy)}{\cost(G,\greedy)}$, and weighted by $\frac{\cost(G,\greedy)}{\delta\cdot\cost(p,\greedy)}$.\\
\ENDFOR
\FORALL{$c_i\in \greedy$}
\STATE{Weigh $c_i\in \greedy$ by the number of points not in $\mathcal{G}\cap C_i$}
\ENDFOR
\STATE Output $\coreset = \greedy \cup \bigcup_G \Omega_G $.
\end{algorithmic}
   \caption{Euclidean Coreset Construction}
   \label{alg}
\end{algorithm}

For every group $G\in \mathcal{G}$, we will prove that the sampling yields an $(\varepsilon,E)$ coreset with $E = \varepsilon\cdot \cost(G,\greedy)$.

Given a solution $\calS$, the basic estimator for the error is 
$$D^\coreset_\calS(G):=\left\vert \sum_{p\in \coreset} w_p \cdot \cost(p,\calS) - \cost(G,\calS)\right\vert.$$
If for all solutions $\calS$ we have a coreset of group $G\in \mathcal{G}$, we can compose the coresets of each group such that we have a coreset for $P$.
Specifically, we will prove Theorem~\ref{thm:main} by proving the following three lemmas.

The first lemma states that we can use the centers of $\greedy$ as proxies for all points not in $\mathcal{G}$.
The second and third lemmas informally give the bounds such that sensitivity sampling for every group $G\in G^M$ and respectively $G\in G^O$ yield corsets.

\begin{lemma}
\label{lem:move}
Let $P$ be a set of points and let $\calS$ be an arbitrary solution. Then 
$$ \left\vert \cost(P\setminus \mathcal{G},\calS) - \sum_{C_i} |\mathcal{G} \cap C_i| \cdot \cost(c_i,\calS) \right\vert \leq \varepsilon\cdot \left(\cost(P,\calS) + \cost(P,\greedy)\right).$$
\end{lemma}

\begin{lemma}
\label{lem:coresetgroupmain}
Let $P$ be a set of points and let $G\subset G^M$ be a group.
Then there exist absolute constants $\gamma_1>0$ such that the sampling procedure of Algorithm~\ref{alg} with $\delta \geq  k\cdot \log \frac{k}{\varepsilon} \cdot \varepsilon^{-2}\cdot 2^{\gamma_1 \cdot z\log (1+z)}\cdot \log^{3} \varepsilon^{-1}\cdot \min(\varepsilon^{-z},k)$ yields
$$\mathbb{E}\underset{\calS}{\text{ sup }}\left[\frac{1}{\cost(G,\calS) + \cost(G,\greedy)}\cdot D^\coreset_{\calS}(G)\right] \leq \varepsilon.$$
\end{lemma}

\begin{lemma}
\label{lem:coresetgroupouter}
Let $P$ be a set of points and let $G\subset \mathcal{G}^O$ be a group.
Then there exist absolute constants $\gamma_2>0$ such that the sampling procedure of Algorithm~\ref{alg} with $\delta \geq  k\cdot \log \frac{k}{\varepsilon} \cdot \varepsilon^{-2}\cdot 2^{\gamma_2 \cdot z\log (1+z)}\cdot \log^{3} \varepsilon^{-1}$ yields
$$\mathbb{E}\underset{\calS}{\text{ sup }}\left[\frac{1}{\cost(P^G,\calS) + \cost(P^G,\greedy)}\cdot D^\coreset_{\calS}(G)\right] \leq \varepsilon.$$
\end{lemma}

First, we show that this lemma implies our main theorem.

\begin{proof}[Proof of Theorem~\ref{thm:main}]
For every group $G\in \mathcal{G}$, let $\Omega_G$ be the set of points returned by the sampling routine and let $\coreset_{\mathcal{G}}$ be the union of the output of all sampling routines.
We consider 
\begin{eqnarray*}
& &\mathbb{E}~\underset{\calS}{\sup}\left[ \frac{1}{\cost(P,\calS) + \cost(P,\greedy)} D_{\calS}^{\coreset_{\mathcal{G}}}(P\cap \mathcal{G})\right] \\
& = &\mathbb{E}~\underset{\calS}{\sup}\left[ \frac{1}{\cost(P,\calS) + \cost(P,\greedy)} \left\vert \sum_{G\in \mathcal{G}}\sum_{p\in \coreset_{G}} w_p \cdot \cost(p,\calS) - \cost(G,\calS)\right\vert \right] \\
&\leq & \mathbb{E}~\underset{\calS}{\sup}\left[ \frac{1}{\cost(P,\calS) + \cost(P,\greedy)} \sum_{G\in \mathcal{G}} D_{\calS}^{\coreset_{G}}(G)\right] \\
& \leq &  \mathbb{E}~\underset{\calS}{\sup}\left[ \sum_{G\in G^M} \frac{\cost(G,\calS) + \cost(G,\greedy)}{\cost(P,\calS) + \cost(P,\greedy)} \cdot \mathbb{E}~\underset{\calS}{\sup}\left[\frac{1}{\cost(G,\calS) + \cost(G,\greedy)} \cdot D_{\calS}^{\coreset_{G}}(G)\right]\right. \\
& & \left.  + \sum_{G\in G^O} \frac{\cost(P^G,\calS) + \cost(P^G,\greedy)}{\cost(P,\calS) + \cost(P,\greedy)} \cdot \mathbb{E}~\underset{\calS}{\sup}\left[\frac{1}{\cost(P^G,\calS) + \cost(P^G,\greedy)} \cdot D_{\calS}^{\coreset_{G}}(G)\right]\right] \\
(\cref{lem:coresetgroupmain}) & \leq &  \mathbb{E}~\underset{\calS}{\sup}\left[ \sum_{G\in G^M} \frac{\cost(G,\calS) + \cost(G,\greedy)}{\cost(P,\calS) + \cost(P,\greedy)} \cdot  \varepsilon \right. \\
(\cref{lem:coresetgroupouter}) & & + \left. \sum_{G\in G^O} \frac{\cost(P^G,\calS) + \cost(P^G,\greedy)}{\cost(P,\calS) + \cost(P,\greedy)} \cdot \varepsilon \right] \\
&\leq & \mathbb{E}~\underset{\calS}{\sup} \left[ \varepsilon + \varepsilon\right] = 2\varepsilon
\end{eqnarray*}
Due to Markov's inequality, we have with probability at least $3/4$ that $ D_{\calS}^{\coreset}(P\cap \mathcal{G}) \leq 8\cdot \varepsilon \cdot \left(\cost(P,\calS) + \cost(P,\greedy)\right)$ for all $\calS$.
Combining this with \cref{lem:move}, we then have for all $\calS$
\begin{eqnarray*}
D_{\calS}^\Omega(P) &\leq &  D_{\calS}^{\coreset_{\mathcal{G}}}(P\cap \mathcal{G}) + \left\vert \cost(P\setminus \mathcal{G},\calS) - \sum_{C_i} |\mathcal{G} \cap C_i| \cdot \cost(c_i,\calS) \right\vert \\
&\leq & 9\cdot \varepsilon \cdot \left(\cost(P,\calS) + \cost(P,\greedy)\right).
\end{eqnarray*}
Rescaling $\varepsilon$ by a factor $9\cdot \left(1+ \frac{\cost(P,\greedy)}{\opt}\right) \in 2^{O(z)}$ yields the desired accuracy.
What is left is to prove the space bound. The maximum number of samples in any group required by \cref{lem:coresetgroupmain} and \cref{lem:coresetgroupouter} is in $O(k\cdot \log \frac{k}{\varepsilon} \cdot \varepsilon^{-2} \cdot \log^3 \varepsilon^{-1} \cdot 2^{O(z\log z)}\cdot \min(\varepsilon^{-z},k))$. Due to Fact~\ref{fact:groupnumber}, the overall coreset therefore has size $O(k\cdot \log \frac{k}{\varepsilon} \cdot \varepsilon^{-2} \cdot \log^5 \varepsilon^{-1} \cdot 2^{O(z\log z)})$.
\end{proof}

The remainder of this section will now focus on the proofs of \cref{lem:coresetgroupmain} and \cref{lem:coresetgroupouter}.
Our main analysis tool will be a chaining argument. To do this, we require two things: (i) a reduction to a Gaussian process and (ii) controlling the variance of said Gaussian process.
The proof of \cref{lem:move} is standard in this line of research and included in the appendix for completeness sake.

\subsection{Setting up a Gaussian process}
The chaining arguments we use for proving \cref{lem:coresetgroupmain} and \cref{lem:coresetgroupouter}, while similar, are distinct enough that each lemma requires it's own notation and approach. We will focus on \cref{lem:coresetgroupmain}, as it arguably the more interesting and important step.
The differences for \cref{lem:coresetgroupouter} are discussed at the end of this section.

Unless mentioned otherwise, we let the group $G$ be in $G^M$.
For proving~\cref{lem:coresetgroupmain}, we need to have a handle on $\sum_{p\in G \cap \Omega}w_p \cost(p, \calS)$, to show that $D^\coreset_\calS(G)$ is concentrated around zero. 
We will not try to work directly with the basic cost estimator $\sum_{p\in P \cap \Omega}v^{\calS}_p\cdot w_p$, since it has a too large variance.
We denote the cost vector $v^{G,\calS}_p = \begin{cases} v^{\calS}_p & \text{if }p\in G \\
0 & \text{else}\end{cases}.$
We will split the cost vector $v^{G,\calS}$ into two vectors for which we have separate estimators, for which we will be able show strong concentration.

To define those estimators, let us first characterize the clusters of the initial solution $\greedy$ as follow.
\begin{itemize}
\item We say that a cluster $C_i\cap G$ induced by $\greedy$ is \emph{huge} if there exists a point $p\in C_i\cap G$ such that $\cost(p,\calS)\geq \left(\frac{4z}{\varepsilon}\right)^z \cdot \cost(p,\greedy)$. The set of huge clusters induced by $\calS$ in $G$ are denoted by $H_{G,\calS}$.
\end{itemize}

Instead of estimating $\|v^{G,\calS}\|_1$ directly, we now split $v^{G,\calS}$ in two vectors for which we carry out the estimation separately. 
We, define the $|P|$-dimensional vector $u^{G,\calS}$ with entries
\begin{equation}
\label{eq:defq}
u^{G,\calS}_p := \begin{cases} 
\cost(p,\calS) &\text{if }p\in C\cap G \text{ and } C\in H_{G,\calS}   \\
0 &\text{otherwise}\end{cases}.
\end{equation}
Clearly $v^{G,\calS} = v^{G,\calS}-u^{G,\calS} + u^{G,\calS}$, but even more importantly 
\begin{equation}
\label{eq:costsplit}
\|v^{G,\calS}\|_1 = \|v^{G,\calS}-u^{G,\calS}\|_1 + \|u^{G,\calS}\|_1
\end{equation}
as none of the entries of the considered vectors are negative. 
 
For a group $G\in G^O$, we also characterize the clusters by a type.
\begin{itemize}
\item We say that a cluster $C_i\cap G$ induced by $\greedy$ is \emph{far} if there exists a point $p\in C_i\cap G$ such that $\cost(p,\calS)\geq 4^z \cdot \cost(p,\greedy)$. The set of far clusters induced by $\calS$ in $G$ are denoted by $F_{G,\calS}$.
\end{itemize}
Again, we split the cost vector $v^{G,\calS}$ into two parts. Here we define
\begin{equation}
\label{eq:defq2}
u^{G,\calS}_p := \begin{cases}
\cost(p,\calS) &\text{if }p\in C\cap G \text{ and } C\in F_{G,\calS}   \\
0 &\text{otherwise}\end{cases}.
\end{equation}
As above, Equation~\ref{eq:costsplit} holds for this definition of $u^{G,\calS}$.

We will estimate $\|v^{G,\calS}-u^{G,\calS}\|_1$ in both cases by means of controlling a Gaussian process. Estimating $\|u^{G,\calS}\|_1$ is done via more straightforward methods.

To show that $\sum_{p\in \coreset} w_p \cdot\left(v^{G,\calS}_p - u^{G,\calS}_p\right)$ is concentrated around its expectation $\|v^{G,\calS} - u^{G,\calS}\|_1$, we introduce a notion of nets for cost vectors defined as follows.

\begin{definition}\label{def:clusteringnets}
Let $I$ be a metric space, $P$ a set of points and two positive integers $k$ and $z$, and let $\alpha > 0$ be a precision parameter. Given some solution $\greedy$, suppose that $G$ is a group of $P$.  Let $\mathbb{C}\subset I^k$ be a (potentially infinite) set of candidate $k$-clusterings. 
We say that a set of cost vectors $\mathbb{N}\subset \R^{|P|}$ is an $(\alpha,k,z)$-clustering net if for every $\calS\in \mathbb{C}$ there exists a vector $v \in \mathbb{N}$ such that the following condition holds.

For all $p \in C\cap G$ such that $C\cap G$ is not huge and not far, 
$$|\cost(p,\calS) - v_p| \leq \alpha\cdot (\cost(p, \calS) + \cost(p, \greedy)).$$
For all $p \in C\cap G$ such that $C\cap G$ is either huge or far, 
$$v_p = 0.$$
\end{definition}

The existence of small clustering nets is given by Lemma~\ref{lem:netsize} and~\ref{lem:netsizelarge} in Section~\ref{sec:nets} further below. Before we prove these lemmas, we first describe how this allows us to use a Gaussian process.

Consider a sequence of $|P|$ dimensional vectors $v^{\calS,1}, v^{\calS,2}, \ldots$ such that $v^{\calS,h}$ is the vector approximating the cost vector $v^{G,\calS}-u^{G,\calS}$ of $\calS$ from a $(2^{-h},k,z)$ clustering net $\mathbb N_h$.
Let us now consider our estimator of $\|v^{G,\calS}-u^{G,\calS}\|_1$ defined as follows. 
\begin{eqnarray*}
Y_{G,p,\calS} &:= & \left(\sum_{h=1}^{\infty} w_p\cdot \left(v^{\calS,h+1}_p - v^{\calS,h}_p \right)\right) + w_p\cdot v^{\calS,1}_p  \\
Y_{G,\calS} &:=& \sum_{p\in \Omega} Y_{G,p}
\end{eqnarray*}

The following fact shows that this sum telescopes, and that the expectation of $Y_{G,\calS}$ remains $\|v^{G,\calS}-u^{G,\calS}\|_1$.

\begin{fact}
$\E_{\coreset}\left[Y_{G,\calS}\right] = \|v^{G,\calS} - u^{G,\calS}\|_1.$
\end{fact}
\begin{proof}
For a fixed point $p$, it holds that $\lim_{h \rightarrow \infty} v^{h}_p = v^{G,\calS}_p - u^{G,\calS}_p$. 
Hence, the infinite sum is well defined and we have:
\begin{align*}
&\sum_{h=1}^{\infty} w_p\left(v^{\calS,h+1}_p - v^{\calS,h}_p\right) +w_p \cdot v^{\calS,1}_p  \\
=~& w_p  v^{G,\calS}_p -u^{G,\calS} \text{ since the sum telescopes}\\
\end{align*}
Hence, summing over all points $p \in \coreset$ and taking the expecation concludes the lemma.
\end{proof}

Using this fact, we can estimate $\|v^{G,\calS} - u^{G,\calS}\|_1$ by $Y_{G, \calS}$. 

To prove \cref{lem:coresetgroupmain}, we in particular wish to show for $G\in G^{M}$
$$\mathbb{E}_{\coreset}\underset{\calS}{\text{sup}} \left|\frac{1}{\cost(G,\calS) + \cost(G,\greedy)}\cdot (Y_{G,\calS} - \mathbb{E}[Y_{G,\calS}])\right| \leq \varepsilon.$$
Analogously, for \cref{lem:coresetgroupouter}, we wish to show for $G\in G^O$
$$\mathbb{E}_{\coreset}\underset{\calS}{\text{sup}} \left|\frac{1}{\cost(P^G,\calS) + \cost(P^G,\greedy)}\cdot (Y_{G,\calS} - \mathbb{E}[Y_{G,\calS}])\right| \leq \varepsilon.$$
Unfortunately, it is difficult to apply the chaining framework with weighted Boolean variables. This is usually addressed using the following symmetrization argument. We pick $\delta$ independent standard normal Gaussian random variables $\xi_1,\ldots,\xi_{\delta}\sim \mathcal{N}(0,1)$ and analyse the following random variables for the respective cases $G\in G^M$ and $G\in G^O$ 
\begin{eqnarray*}
X_{G,\calS} &:=& \frac{\sum_{p\in \Omega} \left(\sum_{h=1}^{\infty} \xi_p \cdot w_p\cdot \left(v^{\calS,h+1}_p - v^{\calS,h}_p\right)\right) + \xi_p\cdot w_p\cdot v^{\calS,1}(p)}{\cost(G,\calS) + \cost(G,\greedy)}, \\
X_{G,\calS} &:=& \frac{\sum_{p\in \Omega} \left(\sum_{h=1}^{\infty} \xi_p \cdot w_p\cdot \left(v^{\calS,h+1}_p - v^{\calS,h}_p \right)\right) + \xi_p\cdot w_p\cdot v^{\calS,1}_p }{\cost(P^G,\calS) + \cost(P^G,\greedy)}.
\end{eqnarray*}
The following lemma is due to Rudra and Wootters~\cite{RudraW14}, see also the book by Ledoux and Talagrand~\cite{LT2013} for more general statements.
\begin{lemma}[Appendix B.3 of ~\cite{RudraW14}]
\label{lem:symmetrization}
Let $T=\frac{1}{\cost(G,\calS) + \cost(G,\greedy)}$ or $T=\frac{1}{\cost(P^G,\calS) + \cost(P^G,\greedy)}$.
Then
$\mathbb{E}_{\coreset}\underset{\calS}{\text{sup}} \left|\sum_{p\in \coreset}^{\delta} T\cdot (Y_{G,p,\calS}-\mathbb{E}[Y_{G,p,\calS}])\right| \leq \sqrt{2\pi}\cdot\mathbb{E}_{\coreset}\mathbb{E}_{\xi}\underset{\calS}{\text{sup}} |X_{G,\calS}|$.
\end{lemma}

With these, we now prove the following lemmas.

\begin{lemma}
\label{lem:q}
Let $G\in G^M$.
Suppose $\delta = \gamma_3\cdot k\cdot \log \frac{k}{\varepsilon}\cdot \varepsilon^{-2}$ for some absolute constant $\gamma_5$. 
Then
$$ 
\mathbb{E}_{\coreset}~\underset{\calS}{\sup}\left[\left\vert\frac{\sum_{p\in \Omega} w_p \cdot u^{G,\calS}_p - \|u^{G,\calS}\|_1}{\cost(G,\greedy)+\cost(G,\calS)}\right\vert\right] \leq \varepsilon.
$$
\end{lemma}

\begin{lemma}
\label{lem:process}
Let $G\in G^M$.
Suppose $\delta =  2^{\gamma_4\cdot z\log (1+z)}\cdot k\cdot \log \frac{k}{\varepsilon}\cdot \varepsilon^{-2}\cdot \log^{3} \varepsilon^{-1}\cdot \min(\varepsilon^{-z},k)$ for absolute constants $\gamma_6$ and $\gamma_7$. Then 
$$\mathbb{E}_{\coreset}~\mathbb{E}_{\xi}~\underset{\calS}{\sup}~\vert X_{G,\calS}\vert \leq \varepsilon.$$
\end{lemma}

\begin{lemma}
\label{lem:qO}
Let $G\in G^O$.
Suppose $\delta = \gamma_5\cdot k\cdot \log \frac{k}{\varepsilon}\cdot \varepsilon^{-2}$ for some absolute constant $\gamma_9$. 
Then
$$ 
\mathbb{E}_{\coreset}~\underset{\calS}{\sup}\left[\left\vert\frac{\sum_{p\in \Omega} w_p \cdot u^{G,\calS}_p - \|u^{G,\calS}\|_1}{\cost(P^G,\greedy)+\cost(P^G,\calS)}\right\vert\right] \leq \varepsilon.
$$
\end{lemma}

\begin{lemma}
\label{lem:processO}
Let $G\in G^O$.
Suppose $\delta =  2^{\gamma_6\cdot z\log (1+z)}\cdot k\cdot \log \frac{k}{\varepsilon}\cdot \varepsilon^{-2}\cdot \log^{3} \varepsilon^{-1}$ for absolute constants $\gamma_8$ and $\gamma_9$. Then 
$$\mathbb{E}_{\coreset}~\mathbb{E}_{\xi}~\underset{\calS}{\sup}~\vert X_{G,\calS}\vert \leq \varepsilon.$$
\end{lemma}

The proofs of \cref{lem:q} and \cref{lem:qO} are in Section~\cref{sec:huge}, the proofs of \cref{lem:process} and \cref{lem:processO} is split into proving the existence of sufficiently small nets (Section~\ref{sec:nets}) and analysing the variance of the Gaussian process.
For now, we show why these lemmas imply \cref{lem:coresetgroupmain} and \cref{lem:coresetgroupouter}.

\begin{proof}[Proof of \cref{lem:coresetgroupmain}]
As mentioned in Equation~\ref{eq:costsplit}, we have $\cost(G,\calS) = \|v^{G,\calS}\|_1 = \|v^{G,\calS}-u^{G,\calS}\|_1 + \|u^{G,\calS}\|_1$. Due to \cref{lem:symmetrization}, we have
$$\mathbb{E}_{\coreset}\underset{\calS}{\text{ sup }}\left[\frac{1}{\cost(G,\calS)+\cost(G,\greedy)} \cdot |Y_{G,\calS} - \mathbb{E}[Y_{G,\calS}]|\right] \leq \sqrt{2\pi} \cdot \mathbb{E}_{\coreset}~\mathbb{E}_{\xi}~\underset{\calS}{\text{sup}}\vert X_{G,\calS} \vert.$$
Plugging in the bound from \cref{lem:process}, we therefore have
$$\mathbb{E}_{\coreset}\underset{\calS}{\text{ sup }}\left[\frac{1}{\cost(G,\calS)+\cost(G,\greedy)} \cdot |Y_{G,\calS} - \mathbb{E}[Y_{G,\calS}]|\right] \leq \sqrt{2\pi} \varepsilon.$$ 
Then
\begin{eqnarray*}
& &\mathbb{E}_{\coreset}\underset{\calS}{\text{ sup }}\left[\frac{1}{\cost(G,\calS)+\cost(G,\greedy)} D^{\coreset}_{\calS}(G)] \right] \\
&=& \mathbb{E}_{\coreset}\underset{\calS}{\text{sup}}\left[\left\vert \frac{\|v^{G,\calS}-u^{G,\calS}\|_1 + \|u^{G,\calS}\|_1 - \sum_{p\in \Omega}\left( w_p \cdot \left(v^{G,\calS}_p-u^{G,\calS}_p\right) + w_p \cdot u^{G,\calS}_p\right)}{\cost(G,\calS)+ \cost(G,\greedy)}\right\vert\right] \\
&\leq & \mathbb{E}_{\coreset}\underset{\calS}{\text{sup}}\left[\left\vert \frac{\sum_{p\in \Omega} w_p \cdot u^{G,\calS}_p - \|u^{G,\calS}\|_1}{\cost(G,\calS)+ \cost(G,\greedy)}\right\vert\right] \\
& & + \mathbb{E}_{\coreset}\underset{\calS}{\text{sup}}\left[\left\vert \frac{\sum_{p\in \Omega}  w_p \cdot (v^{G,\calS}_p - u^{G,\calS}_p) - \|v^{G,\calS}-u^{G,\calS}\|_1}{\cost(G,\calS)+ \cost(G,\greedy)}\right\vert\right] \\
& &\leq \mathbb{E}_{\coreset}\underset{\calS}{\text{sup}}\left[\left\vert \frac{\sum_{p\in \Omega} w_p \cdot u^{G,\calS}_p - \|u^{G,\calS}\|_1}{\cost(G,\calS)+ \cost(G,\greedy)}\right\vert\right] 
+ \mathbb{E}_{\coreset}\underset{\calS}{\text{sup}} \left|\sum_{p\in \coreset}^{\delta}  \frac{Y_{G,p,\calS}-\mathbb{E}[Y_{G,p,\calS}]}{\cost(G,\calS)+ \cost(G,\greedy)}\right|\\
(\cref{lem:q}) & \leq & \varepsilon + \sqrt{2\pi} \varepsilon.
\end{eqnarray*}
Rescaling $\varepsilon$ yields the claim.
\end{proof}

The proof of \cref{lem:coresetgroupouter} is completely analogous. For completeness sake, we repeat the steps.

\begin{proof}[Proof of \cref{lem:coresetgroupouter}]
As mentioned in Equation~\ref{eq:costsplit}, we have $\cost(G,\calS) = \|v^{G,\calS}\|_1 = \|v^{G,\calS}-q^{G,\calS}\|_1 + \|u^{G,\calS}\|_1$. Due to \cref{lem:symmetrization}, we have
$\mathbb{E}_{\coreset}\underset{\calS}{\text{ sup }}\left[\frac{1}{\cost(P^G,\calS)+\cost(P^G,\greedy)} \cdot |Y_{G,\calS} - \mathbb{E}[Y_{G,\calS}]|\right] \leq \sqrt{2\pi} \cdot \mathbb{E}_{\coreset}~\mathbb{E}_{\xi}~\underset{\calS}{\text{sup}}\vert X_{G,\calS} \vert$.
Plugging in the bound from \cref{lem:processO}, we therefore have
\newline
$\mathbb{E}_{\coreset}\underset{\calS}{\text{ sup }}\left[\frac{1}{\cost(P^G,\calS)+\cost(P^G,\greedy)} \cdot |Y_{G,\calS} - \mathbb{E}[Y_{G,\calS}]|\right] \leq \sqrt{2\pi} \varepsilon$. Then
\begin{eqnarray*}
& &\mathbb{E}_{\coreset}\underset{\calS}{\text{ sup }}\left[\frac{1}{\cost(P^G,\calS)+\cost(P^G,\greedy)} D^{\coreset}_{\calS}(G)] \right] \\
&=& \mathbb{E}_{\coreset}\underset{\calS}{\text{sup}}\left[\left\vert \frac{\|v^{G,\calS}-u^{G,\calS}\|_1 + \|u^{G,\calS}\|_1 - \sum_{p\in \Omega} w_p \cdot \left(v^{G,\calS}_p-u^{G,\calS}_p\right) + w_p \cdot u^{G,\calS}_p}{\cost(P^G,\calS)+ \cost(P^G,\greedy)}\right\vert\right] \\
&\leq & \mathbb{E}_{\coreset}\underset{\calS}{\text{sup}}\left[\left\vert \frac{\sum_{p\in \Omega} w_p \cdot u^{G,\calS}_p - \|u^{G,\calS}\|_1}{\cost(P^G,\calS)+ \cost(P^G,\greedy)}\right\vert\right] \\
& & + \mathbb{E}_{\coreset}\underset{\calS}{\text{sup}}\left[\left\vert \frac{\sum_{p\in \Omega} w_p \cdot (v^{G,\calS}_p - u^{G,\calS}_p) - \|v^{G,\calS}-u^{G,\calS}\|_1}{\cost(P^G,\calS)+ \cost(P^G,\greedy)}\right\vert\right] \\
\cref{lem:qO} & \leq & \varepsilon + \sqrt{2\pi} \varepsilon.
\end{eqnarray*}
Rescaling $\varepsilon$ yields the claim.
\end{proof}

%

\subsection{A Structural Lemma}
\label{sec:struct}

We will use the property for $G^M$ that we have a good estimator for the size of every cluster of $\greedy$. We will frequently use this property in subsequent sections.
By definition of groups, we have for every point $p$ of any cluster $C$ with a non-empty intersection with $G\in G^M$
\begin{equation}
\label{eq:ksize1}
\cost(G,\greedy) \leq 2k\cdot \cost(C\cap G,\greedy) \leq 4k\cdot |C\cap G|\cdot \cost(p,\greedy).
\end{equation}
We first show that, given we sampled enough points, $|C\cap G|$ is well approximated for every cluster $C$. This lemma will also be used later for bounding the supremum of $X_{G,\calS}$ in the proof of \cref{lem:process}. 
We define event $\mathcal{E}_G$ to be for all clusters $C$,
\begin{equation*}
\sum_{p\in C \cap G \cap \coreset} w_p = \sum_{p\in C \cap G \cap \coreset} \frac{\cost(G,\greedy)}{\delta\cdot \cost(p,\greedy)} = (1\pm \varepsilon) \cdot |C\cap G|.
\end{equation*}

\begin{lemma}
\label{lem:ksize}
Let $G\in G^M$. 
We have that with probability at least $1 - k\cdot \exp\left(-\frac{\varepsilon^2}{9\cdot k}\delta\right)$, event $\calE_G$ happens.
\end{lemma}

The proof is similar to the one used in Lemma 4.4 from \cite{CSS21}. The main difference is, due to using a slightly different sampling distribution, Hoeffding's inequality is insufficient and we have to rely on Bernstein's inequality.

\begin{proof}[Proof of Lemma~\ref{lem:ksize}]
First, observe that $\mathbb{E}[\sum_{p\in C \cap G \cap \coreset} w_p] = |C\cap G|$.
We will bound both the variance as well as $M$ in order to apply Bernstein's inequality.
Let $\coreset$ be the set of sampled points and let $p_j$ be the $j$th point in the sample with respect to some arbitrary but fixed ordering.
Consider the random variable $w_{p_j,C} = \begin{cases}w_{p}&\text{if }p_j=p\in C \cap G \\ 0 & \text{else}\end{cases}$. 
Then 
\begin{eqnarray}
\nonumber
\text{Var}[w_{p_j,C}] &\leq & \mathbb{E}[w_{p_j,C}^2] = \sum_{p\in C\cap G} w_p^2 \cdot \pr[p\in \coreset] = \sum_{p\in C\cap G} \frac{\cost(G,\greedy)}{\delta^2\cdot \cost(p,\greedy)} \\
\label{eq:ksize2}
(Eq.~\ref{eq:ksize1})&\leq & \sum_{p\in C\cap G} \frac{2k\cost(C,\greedy)}{\delta^2 \cost(p,\greedy)} \leq \sum_{p\in C\cap G} \frac{4k\cdot |C\cap G|}{\delta^2} \leq \frac{4k\cdot |C\cap G|^2}{\delta^2}
\end{eqnarray}
For the maximum upper bound, we have again due to Equation~\ref{eq:ksize1}
\begin{equation}
\label{eq:ksize3}
w_{p_j,C } = \frac{\cost(P,\greedy)}{\delta\cdot \cost(p_j,\greedy)} \leq \frac{4k\cdot |C \cap G|}{\delta} 
\end{equation}
Thus, combining Equation~\ref{eq:ksize2} and~\ref{eq:ksize3} with Bernstein's inequality, we have
\begin{eqnarray*}
\pr\left[\left\vert\sum_{p\in C_i\cap \coreset} w_p - |C \cap G| \right\vert > \varepsilon\cdot |C \cap G|\right] &\leq & \exp\left(-\frac{\varepsilon^2\cdot |C \cap G|^2}{2\delta\cdot \frac{4k\cdot |C \cap G|^2}{\delta^2}+ \frac{2}{3}\frac{4k\cdot |C \cap G|}{\delta} \cdot \varepsilon\cdot |C \cap G|}\right) \\
&\leq & \exp\left(-\frac{\varepsilon^2}{9\cdot k} \cdot \delta\right).
\end{eqnarray*}
The lemma now follows by taking a union bound over all clusters in $\greedy$.
\end{proof}

\subsection{Existence of Small Clustering Nets}
\label{sec:nets}

For a set of points $P$, a set of points $\mathcal{N}_{\varepsilon}$ is an \emph{$\varepsilon$-net} of $P$ if for every point $x\in P$ there exists some point $y\in \mathcal{N}_{\varepsilon}$ with $\|x-y\|\leq \varepsilon$.
The existence of small nets in Euclidean spaces is given by the following statement.
\begin{lemma}[Lemma 5.2 of~\cite{VershyninEK12}]
\label{lem:Euclideannets}
For the unit $d$-dimensional Euclidean ball centered around the origin, there exists an $\varepsilon$-net of cardinality $(1+2/\varepsilon)^d$.
\end{lemma}

We further will crucially rely on terminal embeddings defined as follows. A terminal embedding of a set $P\in \mathbb{R}^d$ is a mapping $f:\mathbb{R}^{d}\rightarrow \mathbb{R}^m$ such that 
\[\forall x\in P,~\forall y\in \mathbb{R}^d,~(1-\varepsilon)\cdot\|x-y\|_2 \leq \|f(x)-f(y)\|_2 \leq (1+\varepsilon)\cdot \|x-y\|_2.\]
The statement is closely related to the classic Johnson-Lindenstrauss lemma. The crucial generalization is that the pairwise distances between any point of $\mathbb{R}^d$ and any point of $P$, rather than just the pairwise distances of points in $P$, are preserved.

\begin{theorem}[Theorem 1.1 of \cite{NaN18}]
\label{thm:terminal}
For any point set $P$ in $\mathbb{R}^d$, there exists a terminal embedding $f:\mathbb{R}^d \rightarrow \mathbb{R}^m$ with $m\in O(\varepsilon^2 \log \|P\|_0)$.
\end{theorem}
The target dimension here is optimal for a wide range of parameters (see Larsen and Nelson for a matching lower bound~\cite{LarsenN17}).
Using both of these statements, we now show the existence of small clustering nets.

\begin{lemma}
\label{lem:netsize}
Let $k, z$ be two positive integers, $G$ be a group and $\greedy$ be a solution to $(k, z)$-clustering.
Define $\cand$ to be the set of possible candidate centers.
For all $\alpha \leq 1/2$, there exists an $(\alpha,k,z)$-clustering net $\mathbb{N}$ of $\cand$ with 
$$|\mathbb{N}|\leq \exp\left(\gamma_{7} \cdot z^2 \cdot k\cdot \log \|P\|_0 / \alpha^2 \cdot \log\frac{1}{\alpha\cdot \eps})\right),$$
where $\gamma_{7}$ is an absolute constant. 
\end{lemma}
\begin{proof}
Let $f$ be a terminal embedding of $P$ into $O(\left(\frac{z}{\alpha}\right)^2\log \|P\|_0)$ dimensions given by Theorem~\ref{thm:terminal}.
Given a solution $\calS$, we then have for any $p\in P$ 
\begin{equation*}
\cost(p,f(\calS)) =\left(\min_{s\in \calS} \|f(p)-f(s)\|\right)^z = \left((1\pm \alpha/O(z))\cdot \min_{s\in \calS} \|p-s\|\right)^z = (1\pm \alpha)\cdot \cost(p,\calS)
\end{equation*}

Let $B$ be an arbitrary subset of the clusters induced by $\greedy$. Here $B$ is meant to contain the clusters that are not in $H_{G,\calS}$ for a given candidate solution $\calS$, but the exact interpretation of $B$ is not important for the proof.
We will show that for every $B$, there exists an $(\alpha,k,z)$ clustering net $\mathbb{N}_{B}$ of size 
\begin{equation}
\label{eq:netbound}
|\mathbb{N}_{B}| \in \exp\left(O\left(z^2 / \alpha^2 k\cdot \log \|P\|_0  \cdot \log\left(\frac{z}{\alpha\cdot \varepsilon}\right)\right) \right).
\end{equation}
Since there are at most $2^k$ subsets $B$, the overall size of the clustering net is then $\sum_{B} |\mathbb{N}_{B}| \leq 2^k\cdot \exp\left(O\left(z^2 \cdot k\cdot \log \|P\|_0 / \alpha^2  \cdot \log\left(\frac{z}{\alpha\cdot \varepsilon}\right)\right) \right) = \exp\left(O\left(z^2 \cdot k\cdot \log \|P\|_0 / \alpha^2  \cdot \log\left(\frac{z}{\alpha\cdot \varepsilon}\right)\right) \right).$

We now justify Equation~\ref{eq:netbound}.
We take an $\frac{\alpha}{z}\cdot \dist(p,\greedy)$-net of the Euclidean ball centered around $p\in C\in B$ with radius $8\cdot \left(\frac{4z}{\eps}\right) \cdot \dist(p,\greedy)$. Such a net has size at most 
$$\exp\left(O\left(z^2  \log \|P\|_0 / \alpha^2  \cdot \log\left(\frac{z}{\alpha\cdot \varepsilon}\right)\right) \right)$$
due to Lemma~\ref{lem:Euclideannets}.

%
%
We now take the union of all $\frac{\alpha}{z}\cdot \dist(p,\greedy)$-nets of all points $p\in C\in B$. This yields a total number of $\|P\|_0\cdot \exp\left(K_z \log \|P\|_0 \cdot \alpha^{-2} \cdot \log \frac{1}{\alpha\cdot\eps}\right)$ nets points. We set $\mathbb{N}_B$ to be set of all subsets of size $k$ of the union of nets. Clearly, $|\mathbb{N}_B| = \exp\left(K_z \cdot k\cdot  \log \|P\|_0 / \alpha^2 \cdot \log\frac{1}{\alpha\cdot \eps}\right)$ as desired in Equation~\ref{eq:netbound}. What is left to show is that $\mathbb{N}_B$ is an $(O(\alpha),k,z)$ clustering net. The lemma then follows by rescaling $\alpha$.

Let $\calS \in \mathbb{C}$ be a set of $k$ centers. Consider the set $N$ of $k$ net points defined as follows : $N:= \cup_{s \in \calS} \{n_s : n_s \text{ is the closest net point to } f(s)\}$. 
Define the cost vector $v^N$ such that \\ 
$v^N(p)~=~\begin{cases}\cost(f(p), N)  & \text{if } p\in C\notin H_{G,\calS} \\
0 &\text{else}\end{cases}$.
Let $p$ be a point from a cluster $C\notin H_{G,\calS}$.
By definition of $H_{G,\calS}$, this implies 
\begin{equation}
\label{eq:netsize1}
\cost(p, \calS) \leq \left(\frac{4z}{\eps}\right)^z \cost(p, \greedy).
\end{equation}
We need to show that $|\cost(p,\calS) - v_{N_{\calS}}(p)| \leq \alpha\cdot (\cost(p, \calS) + \cost(p, \greedy))$.

Let $s$ be the center closest to $p$ in $\calS$, and $c$ be $p$'s closest center in $\greedy$. The terminal embedding ensures
$$\|f(p)-f(s)\|_2 \leq (1+\alpha /O(z))\|p-s\|_2 \leq (1+\alpha / O(z))^2 \cdot \|f(p)-f(s)\|_2.$$ 
Then,
\begin{align*}
\|f(c)-f(s)\|_2 &\leq \|f(c)-f(p)\|_2 + \|f(p)-f(s)\|_2\\
&\leq (1+\alpha/z)\cdot \|c-p\|_2 + (1+\alpha/z)\cdot\|p-s\|_2\\
&\leq (1+\alpha /O(z)) \left(1+\left(\frac{4z}{\eps}\right)\right) \cdot \|p-c\|_2\\
&\leq 4\cdot \left(\frac{4z}{\eps}\right) \|p-c\|_2
\end{align*}
where third inequality holds due to Equation~\ref{eq:netsize1}.
Hence, $f(s) $ is in the ball of radius $8\left(\frac{4z}{\eps}\right)  \dist(p,c)$ centered around $p$ which implies
$$\vert \dist(f(p),f(\calS)) - \dist(f(p),n_s) \vert \leq \left(\frac{2\alpha}{z}\right) \cdot \dist(f(p),f(\greedy)).$$
We therefore have
\begin{eqnarray}
\notag
& &|\cost(f(p),f(\calS)) - \cost(f(p),n_s)| \\
\notag
 &= & |\dist(f(p),f(\calS)) - \dist(f(p),n_s)| \cdot \dist(f(p),f(\calS))^{z-1} \cdot \sum_{i=0}^{z-1} \left(\frac{\dist(f(p),n_s)}{\dist(f(p),f(\calS))}\right)^i  \\
 \notag
&\leq & \left(\frac{2\alpha}{z}\right) \cdot \dist(f(p),f(\greedy)) \cdot \dist(f(p),f(\calS))^{z-1} \cdot z\cdot  \left(1+ \left(\frac{2\alpha}{z}\right)\cdot \frac{\dist(f(p),f(\greedy))}{\dist(f(p),f(\calS))}\right)^{z-1} \\
\notag
&\leq & \left(\frac{2\alpha}{z}\right) \cdot \dist(f(p),f(\greedy))\cdot z\cdot  \left(\dist(f(p),f(\calS))+ \left(\frac{2\alpha}{z}\right)\dist(f(p),f(\greedy))\right)^{z-1} \\
\notag
&\leq & \left(\frac{2\alpha}{z}\right) \cdot \dist(f(p),f(\greedy))\cdot z\cdot  \left(1+ \left(\frac{2\alpha}{z}\right)\right)\cdot \max\left(\dist(f(p),f(\calS)),\dist(f(p),f(\greedy))\right)^{z-1} \\
\label{eq:diffbound}
&\leq & 4\alpha \cdot \left(\cost(f(p),f(\greedy) + \cost(f(p),f(\calS)\right).
\end{eqnarray}
This implies that
\begin{align*}
|\cost(p, \calS) - v_N(p)| &\leq |\cost(f(p), f(\calS)) - \cost(f(p), n_s)| + O(\alpha) \cdot \cost(p, \calS)\\
(Eq.~\ref{eq:diffbound})&\leq 4\alpha \cdot \left(\cost(f(p),f(\greedy) + \cost(f(p),f(\calS)\right) + O(\alpha)\cdot \cost(p, \calS)\\
&= O(\alpha)\cdot (\cost(p, \calS) + \cost(p, \greedy)).
\end{align*}
Thus, up to a rescaling of $\alpha$ by constant factors, we have the desired accuracy and thereby proving Equation~\ref{eq:netbound}.
\end{proof}

\begin{lemma}
\label{lem:netsizelarge}
Let $k, z$ be two positive integers, $G$ be a group and $\greedy$ be a solution to $(k, z)$-clustering. Suppose the points of $G$ lie in $d$ dimensional Euclidean space.
Define $\cand$ to be the set of possible candidate centers.
For all $\alpha \leq 1/2$, there exists an $(\alpha,k,z)$-clustering net $\mathbb{N}$ of $\cand$ with 
$$|\mathbb{N}|\leq \exp\left(\gamma_{8}\cdot k\cdot d \cdot z\log(4z/(\alpha\cdot\eps))\right),$$
where $\gamma_{8}$ is an absolute constant.  
\end{lemma}
\begin{proof}
The construction is essentially identical to that of Lemma~\ref{lem:netsize}. The main difference is that we now take nets of the $d$-dimensional Euclidean ball centered around every point $p$. These nets has size at most 
$$\exp\left(O\left( d \cdot z\log(4z/(\alpha\cdot\eps)) \right)\right) = \exp\left(\gamma_{11}\cdot d \cdot z\log(4z/(\alpha\cdot\eps))\right)$$
due to Lemma~\ref{lem:Euclideannets}.

The remaining arguments from Lemma~\ref{lem:netsize} are not affected by this change.
\end{proof}

Recall that we assumed that $d\in O(\varepsilon^{-2}\log \|P\|_0)$, which is a consequence of Theorem~\ref{thm:terminal}. We will describe the chaining procedure in more detail in Section~\ref{sec:chaining-proof}. For those familiar with chaining: this assumption on $d$, combined with the bound of Lemma~\ref{lem:netsizelarge} will ensure that the chain converges after only a small number of steps.

\subsection{Proofs of \cref{lem:process} and \cref{lem:processO}}\label{sec:chaining-proof}

We focus on the proof of \cref{lem:process}. The proof of \cref{lem:processO} follows along the same lines, but is far simpler. For completeness, we repeat the arguments at the end of the section.

\begin{proof}[Proof of \cref{lem:process}]
In the following, let $G\in G^M$.
We recall the random variable
\begin{eqnarray*}
X_{G,\calS} &:=& \frac{1}{\cost(G,\greedy) + \cost(G,\calS)}\sum_{p\in \Omega} \left(\sum_{h=1}^{\infty} \xi_p \cdot w_p\cdot \left( v^{\calS,h+1}_p - v^{\calS,h}_p \right)\right) + \xi_p\cdot w_p\cdot v^{\calS,1}_p .
\end{eqnarray*}

Define
\begin{eqnarray*}
X_{G,\calS,0} &:=& \frac{1}{\cost(G,\greedy) + \cost(G,\calS)}\sum_{p\in \Omega} \xi_p \cdot w_p\cdot v^{\calS,1}_p \\
X_{G,\calS,h} &:=& \frac{1}{\cost(G,\greedy) + \cost(G,\calS)}\sum_{p\in \Omega} \xi_p \cdot w_p\cdot \left( v^{\calS,h+1}_p - v^{\calS,h}_p \right).
\end{eqnarray*}

We have $\mathbb{E}_{\coreset}\mathbb{E}_{\xi} \underset{\calS}{\text{ sup }} |X_{G,\calS}| \leq \sum_{h=0}^{\infty} \mathbb{E}_{\coreset}\mathbb{E}_{\xi} \underset{\calS}{\text{ sup }} |X_{G,\calS,h}|.$ The number of vectors $v^{\calS,h}$ are bounded via \cref{lem:netsize} and \ref{lem:netsizelarge}. The primary remaining challenge is to control the variance of $X_{G,\calS,h}$.

\begin{lemma}
\label{lem:var}
Let $G\in G^M$. Fix a solution $\calS$ and let $\beta_1,\beta_2>0$ be absolute constants. Then $X_{G,\calS,h}$ is Gaussian distributed with mean $0$. The variance of $X_{G,\calS,h}$ is always at most
$$ \sum_{p\in \Omega} \left(\frac{w_p\cdot \left( v^{\calS,h+1}_p  - v^{\calS,h}_p\right)}{\cost(G,\greedy) + \cost(G,\calS)}\right)^2 \in \delta^{-1} \cdot 2^{\beta_1 z\log (1+ z)} 2^{-2h} \cdot \varepsilon^{-2z}.$$
Furthermore, conditioned on event $\calE_G$, the variance of $X_{G,\calS,h}$ is at most
$$ \sum_{p\in \Omega} \left(\frac{w_p\cdot \left( v^{\calS,h+1}_p - v^{\calS,h}_p  \right)}{\cost(G,\greedy) + \cost(G,\calS)}\right)^2 \in \delta^{-1} \cdot 2^{\beta_2 z\log(1+z)} 2^{-2h} \cdot \min(\varepsilon^{-z},k).$$
\end{lemma}
\begin{proof}
We recall the standard fact that if $\xi_p\sim \mathcal{N}(0,1)$, then $\sum_p a_p\cdot \xi_p$ is Gaussian distributed with mean $0$ and variance $\sum_p a_p^2$.

For any $p\in C\notin H_{G,\calS}$, we have $\cost(p,\calS) \leq \left(\frac{4z}{\varepsilon}\right)^{z} \cost(p, \greedy)$. A terminal embedding with target dimension $O(z^2 2^{2h} \log \|P\|_0)$ preserves the cost up to a factor $(1\pm 2^{-h})$, i.e. we have $(1-2^{-h}) \cdot \cost(p,\calS) \leq v^{\calS,h}_p \leq (1+2^{-h}) \cost(p,\calS)$. Therefore
\begin{eqnarray}
\nonumber
& &\sum_{p\in \Omega} \left(\frac{w_p\cdot \left( v^{\calS,h+1}_p - v^{\calS,h}_p  \right)}{\cost(G,\greedy) + \cost(G,\calS)}\right)^2 \\
\nonumber
& = & \sum_{p\in \Omega} \left(\frac{w_p\cdot \left( v^{\calS,h+1}_p - \cost(p,\calS) +\cost(p,\calS) - v^{\calS,h}_p  \right)}{\cost(G,\greedy) + \cost(G,\calS)}\right)^2 \\
\nonumber
& \leq & \sum_{p\in \Omega} \left(\frac{w_p\cdot  2^{-h-1} \cdot \cost(p,\calS)}{\cost(G,\greedy) + \cost(G,\calS)}\right)^2 \\
\label{eq:variance1}
& \leq & \sum_{p\in \Omega} 2^{-2h+2}\left(\frac{\frac{\cost(G,\greedy)}{\delta \cost(p,\greedy)}\cdot \cost(p,\calS) }{\cost(G,\greedy) + \cost(G,\calS)}\right)^2.
\end{eqnarray}
To prove the first bound, we now merely recall since $p\in C\notin H_{G,\calS}$, we have $\frac{\cost(p,\calS) }{\cost(p,\greedy)} \leq \left(\frac{4z}{\varepsilon}\right)^{z}$.

For the second bound, we first consider an arbitrary cluster $C$ and let $q=\underset{p\in C}{~\argmin~} \cost(p,\calS)$. Then we have for any solution $\calS$
\begin{equation}
\label{eq:variance2}
\cost(p,\calS) \leq (\dist(q,\calS)+\dist(p,q))^z \leq  \frac{2^z \cdot\cost(C\cap G,\calS) + 4^z \cdot \cost(C\cap G,\greedy)}{|C|}.
\end{equation}

We first focus on the case $\varepsilon^{z}<k$. 
Continuing from Equation \ref{eq:variance1} and combining with Equation \ref{eq:variance2}, we have 
\begin{eqnarray*}
& & \sum_{p\in \Omega} 2^{-2h+2}\left(\frac{\frac{\cost(G,\greedy)}{\delta \cost(p,\greedy)}\cdot \cost(p,\calS) }{\cost(G,\greedy) + \cost(G,\calS)}\right)^2 \\
& \leq & \frac{\cost(G,\greedy)}{\delta\cdot (\cost(G,\greedy) + \cost(G,\calS))^2} \cdot 2^{-2h+2} \cdot \sum_{p\in \Omega} \frac{\cost(G,\greedy)}{\delta \cost(p,\greedy)}\cdot \cost(p,\calS) \cdot \left(\frac{4z}{\varepsilon}\right)^{z} \\
& \leq & \frac{\cost(G,\greedy)\cdot 2^{-2h+2} \cdot \left(\frac{4z}{\varepsilon}\right)^{z}}{\delta\cdot (\cost(G,\greedy) + \cost(G,\calS))^2}  \cdot \sum_{C}  \frac{2^z \cdot \cost(C\cap G,\calS) + 4^z \cdot \cost(C\cap G,\greedy)}{|C|}\sum_{p\in C\cap \Omega} \frac{\cost(G,\greedy)}{\delta \cost(p,\greedy)} \\
& \leq & \frac{\cost(G,\greedy)\cdot 2^{-2h+2} \cdot \left(\frac{4z}{\varepsilon}\right)^{z}}{\delta\cdot (\cost(G,\greedy) + \cost(G,\calS))^2}  \cdot \sum_{C}  2^z \cdot \cost(C\cap G,\calS) + 4^z \cdot \cost(C\cap G,\greedy) \\
& \leq & \frac{\cost(G,\greedy)\cdot 2^{-2h+2} \cdot \left(\frac{4z}{\varepsilon}\right)^{z}}{\delta\cdot (\cost(G,\greedy) + \cost(G,\calS))^2}  \cdot 8^z\cdot (\cost(G,\greedy) + \cost(G,\calS)) \\
& \leq & \frac{2^{-2h+2} \cdot \left(\frac{4z}{\varepsilon}\right)^{z}\cdot 8^z}{\delta}   
\end{eqnarray*}
where the third inequality uses event $\mathcal{E}_G$.

We now obtain a bound depending on $k$. Recall for any point $p\in C\cap G$, we have 
\begin{equation}
\label{eq:variance3}
\cost(G,\greedy) \leq 2k\cdot \cost(C\cap G,\greedy) \leq 4k\cdot \cost(p,\greedy)\cdot |C|
\end{equation}
by definition of the groups. Then
\begin{eqnarray*}
& & \sum_{p\in \Omega} 2^{-2h+2}\left(\frac{\frac{\cost(G,\greedy)}{\delta \cost(p,\greedy)}\cdot \cost(p,\calS) }{\cost(G,\greedy) + \cost(G,\calS)}\right)^2 \\
& \leq & \frac{1}{\delta\cdot (\cost(G,\greedy) + \cost(G,\calS))^2} \cdot 2^{-2h+2} \cdot \sum_{C} \sum_{p\in C\cap \Omega} \frac{4\cdot \cost(G,\greedy) \cdot |C|\cdot k}{\delta \cdot \cost(p,\greedy)}\cdot \cost^2(p,\calS)\\
& \leq & \frac{2^{-2h+4} \cdot k}{\delta\cdot (\cost(G,\greedy) + \cost(G,\calS))^2}  \cdot \sum_{C} |C| \cdot \left(\frac{2^z \cdot\cost(C\cap G,\calS) + 4^z \cdot \cost(C\cap G,\greedy)}{|C|}\right)^2 \\
& & ~~~~~~~~~~~~~~~~~~~~~~~~~~~~~~~~~~~~~~~~~~~~~~~~\cdot \sum_{p\in C\cap \Omega} \frac{\cost(G,\greedy)}{\delta \cost(p,\greedy)} \\
& \leq & \frac{2^{-2h+4} \cdot k}{\delta\cdot (\cost(G,\greedy) + \cost(G,\calS))^2}  \cdot \sum_{C} \left(2^z \cdot\cost(C\cap G,\calS) + 4^z \cdot \cost(C\cap G,\greedy)\right)^2 \\
& \leq & \frac{2^{-2h+4} \cdot k}{\delta\cdot (\cost(G,\greedy) + \cost(G,\calS))^2}  \cdot \left(\sum_{C} 2^z \cdot\cost(C\cap G,\calS) + 4^z \cdot \cost(C\cap G,\greedy)\right)^2 \\
& \leq & \frac{2^{-2h+4} \cdot k \cdot 64^z}{\delta\cdot (\cost(G,\greedy) + \cost(G,\calS))^2}  \cdot \left(\cost(G,\greedy) +  \cost(G,\calS)\right)^2 \\
& \leq & \frac{2^{-2h+4} \cdot k \cdot 64^z}{\delta} 
\end{eqnarray*}
where the first inequality uses Equation~\ref{eq:variance3}, the second inequality uses Equation~\ref{eq:variance2} and the third inequality uses Equation~\ref{eq:variance1}.
\end{proof}

We now combine the bound on the variance with union bounds for $X_{G,\calS,h}$.
Let $\delta$ be as given in the statement of \cref{lem:process}.

We will use the following lemma for bounding the expected maximum of independent Gaussians.  
\begin{lemma}[Lemma 2.3 of Massart~\cite{massart2007}]
\label{lem:minichain}
Let $g_i\thicksim\mathcal{N}(0,\sigma_i^2)$, $i\in [n]$ be Gaussian random variables and suppose $\sigma_i\leq \sigma$ for all $i$. Then
$$ \mathbb{E}[\underset{i\in [n]}{\max} |g_i|] \leq 2\sigma\cdot \sqrt{2\ln n}.$$
\end{lemma}

We use the following fact:
\begin{fact}
$$\mathbb{E}_{\Omega,\xi}[\sup_{\calS} X_{G,\calS,h}] = \pr_{\Omega}[\calE_G]\cdot \mathbb{E}_{\xi} [\sup_{\calS} X_{G,\calS,h}~|~\calE_G] + \pr_{\Omega}[\overline{\calE_G}]\cdot \mathbb{E}_{\xi} [\sup_{\calS} X_{G,\calS,h}~|~\overline{\calE_G}].$$

$$\mathbb{E}_{\Omega,\xi}[\sup_{\calS} X_{\calS}] = \pr_{\Omega}[\calE_G]\cdot \mathbb{E}_{\coreset, \xi} [\sup_{\calS} X_{G,\calS,h}~|~\calE_G] + \pr_{\Omega}[\overline{\calE_G}]\cdot \mathbb{E}_{\coreset, \xi} [\sup_{\calS} X_{G,\calS,h}~|~\overline{\calE_G}].$$
\end{fact}
\begin{proof}
Since $\calE_G$ is independent of $\xi$, the law of total expectation gives
\begin{align*}
\mathbb{E}_{\Omega,\xi}[\sup_{\calS} X_{G,\calS,h}] &= \pr_{\coreset,\xi}[\calE_G] \cdot \E_{\coreset,\xi} [\sup_{\calS} X_{G,\calS,h}~|~\calE_G]
+ \pr_{\coreset,\xi}[\overline \calE_G]\cdot \E_{\coreset,\xi}[\sup_{\calS} X_{G,\calS,h}~|~\overline \calE_G]\\
&= \pr_{\coreset}[\calE_G] \cdot \E_{\coreset,\xi} [\sup_{\calS} X_{G,\calS,h}~|~\calE_G]
+ \pr_{\coreset}[\overline \calE_G]\cdot \E_{\coreset,\xi}[\sup_{\calS} X_{G,\calS,h}~|~\overline \calE_G]
\end{align*}
\end{proof}

We will bound the expectation, first assuming the (more likely case) that event $\calE_G$ holds, then assuming the (more unlikely case) that event $\calE_G$ does not hold.

Condition on $\calE_G$. We simply upper bound $\pr[\calE_G]$ with $1$. Assume that the points lie in dimension $\alpha\cdot \log \|P\|_0 \cdot \varepsilon^{-2}$ and let $t\in  O(\log (\alpha \cdot \varepsilon^{-2}))$.
We will show that the contribution of the $X_{G,\calS,h}$ with $h\leq t$ to the expectation is at most $\varepsilon/\log (1/\varepsilon)$, and then bound the expectation for all remaining $h\geq t$.

First, we recall that, conditioned on event $\calE_G$ and due to Lemma~\ref{lem:var} and by our choice of $\delta$, we have
$$\textbf{Var} [X_{h,\calS} ~|~ \calE_G] \leq \beta_3 \cdot \delta^{-1} \cdot 2^{\beta_4 z\log z} 2^{-2h} \cdot \min(\varepsilon^{-z},k).$$ for absolute constants $\beta_3$ and $\beta_4$.

For the number of distinct $X_{G,\calS,h}$, we have an upper bound of $|\mathbb{N}_{h-1}|\cdot |\mathbb{N}_{h}|\leq |\mathbb{N}_{h}|^2$, where $|\mathbb{N}_{h}|$ is the size of an $(2^{-h},k,z)$-clustering net. Due to Lemma~\ref{lem:netsize}, this is at most $\exp(2\cdot\gamma_{7} \cdot z^2 \cdot k\cdot \log \|P\|_0 \cdot 2^{2h} \cdot \log\frac{1}{2^{-h}\cdot \eps})$, which by the upper bound on $h$ is at most $\exp(2\cdot\gamma_{7} \cdot z^2 \cdot k\cdot \log \|P\|_0 \cdot 2^{2h} \cdot \log \frac{\alpha}{\varepsilon^{-3}})$.
Therefore, using Lemma~\ref{lem:minichain} 

\begin{eqnarray}
\nonumber
& & \mathbb{E}[\sup_{\calS} |X_{G,\calS,h}| ~|~\calE_G] \\
\nonumber
 & \leq & 2 \sqrt{\textbf{Var} [X_{G,\calS,h} ~|~ \calE_G]} \sqrt{2\ln (|\mathbb{N}_{h-1}|\cdot |\mathbb{N}_{h}|)}\\
 \nonumber
&\leq & 2 \sqrt{\delta^{-1} \cdot 2^{\beta_1 z\log z} 2^{-2h} \cdot \min(\varepsilon^{-z},k)} \cdot \sqrt{2\cdot\gamma_{7} \cdot z^2 \cdot k\cdot \log \|P\|_0 \cdot 2^{2h} \cdot \log\log \frac{\alpha}{\varepsilon^{-3}}} \\
 \nonumber
 &=& 2\cdot \left(\frac{2\cdot\gamma_{7} \cdot z^2 \cdot k\cdot \log \|P\|_0 \cdot 2^{2h} \cdot \log  \frac{\alpha}{\varepsilon^{-3}} \cdot 2^{\beta_1 z\log z} 2^{-2h} \cdot \min(\varepsilon^{-z},k)}{ 2^{\gamma_4\cdot z\log (1+z)}\cdot k\cdot \log \frac{k}{\varepsilon}\cdot \varepsilon^{-2}\cdot \log^{3} \varepsilon^{-1}\cdot \min(\varepsilon^{-z},k))}\right)^{1/2}\\
 \label{eq:chainappl1}
  &\leq &  \beta \cdot \frac{\varepsilon}{\log 1/\varepsilon},
\end{eqnarray}
where the final inequality holds for a sufficiently large choice of the constants $\gamma_4$.

We now assume that $h\geq t$. This time, using Lemma~\ref{lem:netsizelarge}, we have a net of size at most $|\mathbb{N}_h|\leq \exp\left(\gamma_{8}\cdot k\cdot d \cdot z\log(4z/(2^{-h}\eps))\right)$. Furthermore, using the assumption that the points lie in dimension $\log \|P\|_0 \cdot \varepsilon^{-2} \in O(\log k/\varepsilon) \cdot \varepsilon^{-2}$, we have 

\begin{eqnarray}
\nonumber
& & \mathbb{E}[\sup_{\calS} |X_{G,\calS,h}| ~|~\calE_G] \\
 \\
\nonumber
 & \leq & 2 \sqrt{\textbf{Var} [X_{G,\calS,h} ~|~ \calE_G]} \sqrt{2\ln (|\mathbb{N}_{h-1}|\cdot |\mathbb{N}_{h}|)}\\
 \nonumber
 &\leq & 
 2 \cdot \left(\frac{2\cdot\gamma_{8}\cdot k\cdot d \cdot z\log(4z/(2^{-h}\eps)) \cdot \beta_3 \cdot 2^{\beta_4 z\log z} 2^{-2h} \cdot \min(\varepsilon^{-z},k) }{2^{\gamma_4\cdot z\log (1+z)}\cdot k\cdot \log \frac{k}{\varepsilon}\cdot \varepsilon^{-2}\cdot \log^{3} \varepsilon^{-1}\cdot \min(\varepsilon^{-z},k))}\right)^{1/2}  \\
  \nonumber
 &\leq & 
 2 \cdot \left(\frac{2\cdot\gamma_{8}\cdot k\cdot \alpha \cdot \|P\|_0\varepsilon^{-2} \cdot z\log(4z/(2^{-h}\eps)) \cdot 2^{\beta_2 z\log (1+z)} 2^{-2h} \cdot \min(\varepsilon^{-z},k) }{ 2^{\gamma_4\cdot z\log (1+z)}\cdot k\cdot \log \frac{k}{\varepsilon}\cdot \varepsilon^{-2}\cdot \log^{3} \varepsilon^{-1}\cdot \min(\varepsilon^{-z},k))}\right)^{1/2}  \\
 \label{eq:chainappl2}
  &\leq &  \beta \cdot \left(\frac{\log h}{2^{2h} \log^2 1/\varepsilon}\right)^{1/2},
\end{eqnarray}
where again the final inequality holds for a sufficiently large choice of the constant $\gamma_4$ and with the assumption $\|P\|_0 = \text{poly}(k/\varepsilon)$.

Summing up Equations~\ref{eq:chainappl1} and~\ref{eq:chainappl2} for all $h$, we then obtain
\begin{eqnarray}
\nonumber
\mathbb{E}[\sup_{\calS} |X_{\calS}|~|~\calE_G] &\leq & \sum_{h=0}^{\infty} \mathbb{E}[\sup_{\calS} |X_{h,\calS}| ~|~\calE_G] \\
\nonumber
&=&\sum_{h=0}^{t-1} \mathbb{E}[\sup_{\calS} |X_{h,\calS}| ~|~\calE_G] + \sum_{h=t}^{\infty} \mathbb{E}[\sup_{\calS} |X_{h,\calS}| ~|~\calE_G] \\
\nonumber
(\text{Eq. }\ref{eq:chainappl1}~\ref{eq:chainappl2}) &\leq &\sum_{h=0}^{t-1} \beta\cdot \frac{\varepsilon}{\log 1/\varepsilon} + \sum_{h=t}^{\infty} \beta\cdot  \left(\frac{ \log h }{2^{2h} \log^{2} 1/\varepsilon }\right)^{1/2} \\
 &\overset{t  = \log (\alpha/\varepsilon^2)}{\leq } & \beta\cdot \varepsilon + \beta\cdot \varepsilon \sum_{h=0}^{\infty}  \left(\frac{ \log (h+t) }{\beta \cdot 2^{2h} \log^{2} 1/\varepsilon }\right)^{1/2} \in  O(\varepsilon). 
    \label{eq:chainappl0}
\end{eqnarray}

Now, we move onto the case $\overline{\calE_G}$. Due to Lemma~\ref{lem:ksize}, we have
\begin{equation}
\label{eq:chainappprob}
\pr[\overline{\calE_G}] \leq k\cdot z^2\log^2(z/\varepsilon)\exp(-O(1)\cdot\frac{\varepsilon^2}{k}\delta) \leq \varepsilon^{2z},
\end{equation}
where the second upper bound follows from our choice of sufficiently large constants in the definition of $\delta$\footnote{The bound is not close to tight and can be any power of $\varepsilon$. The stated bound happens to be sufficient here.}

Using the worse variance bound from Lemma~\ref{lem:var} for the variance in equations~\ref{eq:chainappl1} and~\ref{eq:chainappl2}, we then have for $h\leq t$
\begin{eqnarray}
\label{eq:chainappl3}
 \mathbb{E}[\sup_{\calS} |X_{h,\calS}| ~|~\overline{\calE_G}] 
 \leq  \beta \frac{\varepsilon}{\log 1/\varepsilon} \cdot \varepsilon^{-z}.
\end{eqnarray}
and for $h>t$
\begin{eqnarray}
\label{eq:chainappl4}
 \mathbb{E}[\sup_{\calS} |X_{h,\calS}| ~|~\overline{\calE_G}] 
 \leq   \beta \cdot \left(\frac{\log h}{2^{2h} \log^2 1/\varepsilon}\right)^{1/2} \cdot \varepsilon^{-z}.
\end{eqnarray}
Using an analogous calculation to those in the derivation of Equation~\ref{eq:chainappl0} using Equations \ref{eq:chainappl3} and \ref{eq:chainappl4}, we now obtain
\begin{equation}
\label{eq:chainappl00}
\mathbb{E}[\sup_{\calS} |X_{\ell,\calS}|~|~\overline{\calE_G}] \in O(\varepsilon^{-z}).
\end{equation}

Combining Equations~\ref{eq:chainappl0},~\ref{eq:chainappprob}, and~\ref{eq:chainappl00}, we finally have
\begin{eqnarray*}
\pr_{\Omega}[\calE_G]\cdot \mathbb{E}_{\xi} [\sup_{\calS} X_{\ell,\calS}~|~\calE_G] + \pr_{\Omega}[\overline{\calE_G}]\cdot \mathbb{E}_{\xi} [\sup_{\calS} X_{\ell,\calS}~|~\overline{\calE_G}] \in  O(\varepsilon + \varepsilon^{2z}\cdot  \varepsilon^{-z}) = O(\varepsilon).
\end{eqnarray*}

\end{proof}

We now repeat these calculations for \cref{lem:processO}. The main difference is the variance bound, which improves over what we were able to prove for \cref{lem:process}.
The remaining arguments differ only in the calculations and are omitted.

\begin{proof}[Proof of \cref{lem:processO}]

Let $G\in G^O$.
As for \cref{lem:processO}, we bound the random variable
\begin{eqnarray*}
X_{G,\calS} &:=& \frac{1}{\cost(G,\greedy) + \cost(G,\calS)}\sum_{p\in \Omega} \left(\sum_{h=1}^{\infty} \xi_p \cdot w_p\cdot \left( v^{\calS,h+1}_p - v^{\calS,h}_p \right)\right) + \xi_p\cdot w_p\cdot v^{\calS,1}_p .
\end{eqnarray*}

by bounding
\begin{eqnarray*}
X_{G,\calS,0} &:=& \frac{1}{\cost(G,\greedy) + \cost(G,\calS)}\sum_{p\in \Omega} \xi_p \cdot w_p\cdot v^{\calS,1}_p \\
X_{G,\calS,h} &:=& \frac{1}{\cost(G,\greedy) + \cost(G,\calS)}\sum_{p\in \Omega} \xi_p \cdot w_p\cdot \left( v^{\calS,h+1}_p - v^{\calS,h}_p \right).
\end{eqnarray*}

\begin{lemma}
\label{lem:varO}
Let $G\in G^O$ and let $\beta_3>0$ be a constant. Fix a solution $\calS$. Then $X_{G,\calS,h}$ is Gaussian distributed with mean $0$. The variance of $X_{G,\calS,h}$ is always at most
$$ \sum_{p\in \Omega} \left(\frac{w_p\cdot \left( v^{\calS,h+1}_p  - v^{\calS,h}_p\right)}{\cost(G,\greedy) + \cost(G,\calS)}\right)^2 \in  \frac{16^z\cdot 2^{-2h+2}}{\delta}.$$
\end{lemma}
\begin{proof}
As in the proof of \cref{lem:var}, we use that if $\xi_p\sim \mathcal{N}(0,1)$, then $\sum_p a_p\cdot \xi_p$ is Gaussian distributed with mean $0$ and variance $\sum_p a_p^2$.

For any $p\in C\notin H_{G,\calS}$, we have $\cost(p,\calS) \leq \left(\frac{4z}{\varepsilon}\right)^{z} \cost(p\greedy)$. A terminal embedding with target dimension $O(z^2 2^{2h} \log \|P\|_0)$ preserves the cost up to a factor $(1\pm 2^{-h})$, i.e. we have $(1-2^{-h}) \cdot \cost(p,\calS) \leq v^{\calS,h}_p \leq (1+2^{-h}) \cost(p,\calS)$. Therefore
\begin{eqnarray}
\nonumber
& &\sum_{p\in \Omega} \left(\frac{w_p\cdot \left( v^{\calS,h+1}_p - v^{\calS,h}_p  \right)}{\cost(G,\greedy) + \cost(G,\calS)}\right)^2 \\
\nonumber
& = & \sum_{p\in \Omega} \left(\frac{w_p\cdot \left( v^{\calS,h+1}_p - \cost(p,\calS) +\cost(p,\calS) - v^{\calS,h}_p  \right)}{\cost(G,\greedy) + \cost(G,\calS)}\right)^2 \\
\nonumber
& \leq & \sum_{p\in \Omega} \left(\frac{w_p\cdot  2^{-h-1} \cdot \cost(p,\calS)}{\cost(G,\greedy) + \cost(G,\calS)}\right)^2 \\
\label{eq:varianceO1}
& \leq & \sum_{p\in \Omega} 2^{-2h+2}\left(\frac{\frac{\cost(G,\greedy)}{\delta \cost(p,\greedy)}\cdot \cost(p,\calS) }{\cost(G,\greedy) + \cost(G,\calS)}\right)^2.
\end{eqnarray}

By definition, none of the points with non-zero coordinates in the cost vector $v$ are far, i.e. $\cost(p,\calS)\leq 4^z \cdot \cost(p,\greedy)$.

Therefore, Equation~\ref{eq:varianceO1} becomes
$$\sum_{p\in \Omega} 16^z \cdot 2^{-2h+2}\frac{1}{\delta^2} = \frac{16^z\cdot 2^{-2h+2}}{\delta}.$$
\end{proof}

The remaining calculations are completely analogous to that of \cref{lem:process}, albeit with a significantly better (and simpler) bound on the variance.
\end{proof}

\subsection{Estimating $\|u^{G,\calS}\|_1$ (Proofs of \cref{lem:q} and \cref{lem:qO})}
\label{sec:huge}

\begin{lemma}
\label{lem:khuge}
Let $G\in G^M$ be a group.
Condition on event $\calE_G$. Suppose $\varepsilon<1/4$. Then,
for any solution $\calS$, and point $p\in C$ with $C\in H_{G,\calS}$,  we have:
\[\left\vert\cost(C\cap G,\calS) - |C\cap G| \cdot \cost(p,\calS) \right\vert \leq  10\cdot\eps\cdot \cost(C, \calS).\]
\end{lemma}
\begin{proof}
Let $p,p'\in C\in H_{\calS}$. 
First, we require an upper bound on $\cost(p,p')$. We have due to Lemma~\ref{lem:weaktri} and since $\cost(p,\greedy) \leq 2\cdot \cost(p',\greedy)$
$$\cost(p,p')  \leq 2^{z-1} \left(\cost(p,\greedy) + \cost(p',\greedy)\right)\leq 2^{z+1}\cost(p',\greedy).$$
Let $p'\in C$ be a point such that $\cost(p,\calS)>\left(\frac{4z}{\varepsilon}\right)^z \cost(p',\greedy)$.
We now give upper and lower bounds for $\cost(p,\calS)$ in terms of $\cost(p',\calS)$, again using Lemma~\ref{lem:weaktri}.
For the upper bound:
\begin{eqnarray*}
\cost(p,\calS) &\leq & (1+\varepsilon) \cdot \cost(p',\calS) + \left(\frac{z+\varepsilon}{\varepsilon}\right)^{z-1}\cost(p,p') \\
&\leq & (1+\varepsilon) \cdot \cost(p',\calS) + \left(\frac{z+\varepsilon}{\varepsilon}\right)^{z-1}2^{z+1}\cdot \cost(p',\greedy)  \\
&\leq & (1+\varepsilon) \cdot \cost(p',\calS) + \left(\frac{z+\varepsilon}{\varepsilon}\right)^{z-1}2^{z+1}\cdot \left(\frac{\varepsilon}{4z}\right)^z\cost(p',\calS) \\
&\leq & (1+2\varepsilon) \cdot \cost(p',\calS) 
\end{eqnarray*}
For the lower bound:
\begin{eqnarray*}
\cost(p',\calS) &\leq & (1+\varepsilon) \cdot \cost(p,\calS) + \left(\frac{z+\varepsilon}{\varepsilon}\right)^{z-1}\cost(p,p') \\
&\leq & (1+\varepsilon) \cdot \cost(p,\calS) + \left(\frac{z+\varepsilon}{\varepsilon}\right)^{z-1}2^{z+1}\cdot \cost(p',\greedy)  \\
&\leq & (1+\varepsilon) \cdot \cost(p,\calS) + \left(\frac{z+\varepsilon}{\varepsilon}\right)^{z-1}2^{z+1}\cdot \left(\frac{\varepsilon}{4z}\right)^z\cost(p',\calS) \\
&\leq & (1+\varepsilon) \cdot \cost(p,\calS) + \varepsilon\cdot\cost(p',\calS) \\
\Rightarrow \cost(p,\calS) &\geq & \frac{1-\varepsilon}{1+\varepsilon} \cdot \cost(p',\calS) \geq (1-2\varepsilon)\cdot \cost(p',\calS)
\end{eqnarray*}
Thus we have
$\cost(C_i,\calS)=\sum_{p\in C_i}\cost(p,\calS) = (1\pm 2\varepsilon)\cdot |C_i|\cdot \cost(p',\calS)$. Conditioned on event $\calE$, we now have $\sum_{p\in C_i\cap \coreset} w_p = (1\pm\varepsilon)\cdot |C_i|$, hence
\begin{eqnarray*}
\sum_{p\in C_i\cap \coreset} w_p\cdot \cost(p,\calS) &\leq & \sum_{p\in C_i\cap \coreset} w_p \cdot (1 + 2\varepsilon)\cdot \cost(p',\calS) = (1+\varepsilon)\cdot (1+ 2\varepsilon) \cdot |C_i|\cdot \cost(p',\calS) \\
& \leq & \cost(C_i,\calS) \cdot \frac{(1+\varepsilon)\cdot (1+2\varepsilon)}{1-2\varepsilon}
\end{eqnarray*}
and analogously for the lower bound
\begin{eqnarray*}
\sum_{p\in C_i\cap \coreset} w_p\cdot \cost(p,\calS) &\geq & \sum_{p\in C_i\cap \coreset} w_p \cdot (1 - 2\varepsilon)\cdot \cost(p',\calS) = (1-\varepsilon)\cdot (1- 2\varepsilon) \cdot |C_i|\cdot \cost(p',\calS) \\
& \geq & \cost(C_i,\calS) \cdot \frac{(1-\varepsilon)\cdot (1-2\varepsilon)}{1+2\varepsilon}.
\end{eqnarray*}
The final bound follows by observing for $\varepsilon<1/4$, we have $\frac{(1+\varepsilon)\cdot (1+2\varepsilon)}{1-2\varepsilon}\leq 1+10\varepsilon$ and $ \frac{(1-\varepsilon)\cdot (1-2\varepsilon)}{1+2\varepsilon} \geq 1-10\varepsilon$.
\end{proof}

\begin{proof}[Proof of Lemma~\ref{lem:q}]
We have
\begin{eqnarray}
\nonumber
& &\mathbb{E}_{\coreset}~\underset{\calS}{\text{sup}}\left[\left\vert\frac{\sum_{p\in \Omega} w_p \cdot u^{\calS}(p) - \|u^{\calS}\|_1}{\cost(G,\greedy)+\cost(G,\calS)}\right\vert\right] \\
\label{eq:q1}
&\leq & \mathbb{E}_{\coreset}~\underset{\calS}{\text{sup}}\left[\left\vert\frac{\sum_{p\in \Omega} w_p \cdot u^{G,\calS}(p) - \|u^{G,\calS}\|_1}{\cost(G,\greedy)+\cost(G,\calS)}\right\vert~|~\calE_G\right]\cdot \pr_\coreset[\calE_G] \\
\label{eq:q2}
& & + \mathbb{E}_{\coreset}~\underset{\calS}{\text{sup}}\left[\left\vert\frac{\sum_{p\in \Omega} w_p \cdot u^{G,\calS}(p) - \|u^{G,\calS}\|_1}{\cost(G,\greedy)+\cost(G,\calS)}\right\vert~|~\overline{\calE_G}\right]\cdot \pr_\coreset[\overline{\calE_G}] 
\end{eqnarray}
We first consider the term \ref{eq:q1}. A trivial upper bound for $\pr_\coreset[\calE_G]$ is $1$. Using Lemma~\ref{lem:khuge} we have
$$\sum_{C\in H_{G,\calS}} w_p\cdot u^{G,\calS}_p = \sum_{C\in H_{G,\calS}} w_p\cdot \cost(p,\calS) = (1\pm 10\varepsilon) \sum_{C\in H_{G,\calS}} \cost(p,\calS).$$
The remaining entries of $u^{G,\calS}$ are $0$. Since $\|u^{G,\calS}\|_1 \leq \cost(G,\calS)$, we therefore have
\begin{equation}
\label{eq:lemq1}
\mathbb{E}_{\coreset}~\underset{\calS}{\text{sup}}\left[\left\vert\frac{\sum_{p\in \Omega} w_p \cdot u^{G,\calS}(p) - \|u^{G,\calS}\|_1}{\cost(G,\greedy)+\cost(G,\calS)}\right\vert~|~\calE_G\right]\cdot \pr_\coreset[\calE_G] \leq 10 \varepsilon.
\end{equation}

We now focus on term~\ref{eq:q2}. We distinguish between two cases.
If $\sum_{p\in \Omega} w_p \cdot u^{\calS}(p) \leq \|u^{\calS}\|_1$ then we have

\begin{equation}
\label{eq:lemq2}
\left\vert\frac{\sum_{p\in \Omega} w_p \cdot u^{\calS}(p) - \|u^{\calS}\|_1}{\cost(G,\greedy)+\cost(G,\calS)}\right\vert \leq \frac{\|u^{\calS}\|_1}{\cost(P,\greedy)+\cost(P,\calS)} \leq 1
\end{equation}

If $\sum_{p\in \Omega} w_p \cdot u^{\calS}(p) \geq \|u^{\calS}\|_1$ then we have
\begin{eqnarray*}
\sum_{p\in \Omega} w_p \cdot u^{G,\calS}_p
&=& \sum_{p\in \Omega} \frac{\cost(G,\greedy)}{\delta\cdot \cost(p,\greedy)} \cdot u^{G,\calS}_p\\
(Eq.~\ref{eq:ksize1}) &\leq & \sum_{C}\sum_{p\in \Omega\cap C \cap G} \frac{4k |C\cap G|\cdot \cost(G,\greedy)}{\delta\cdot \cost(G,\greedy)} \cdot u^{G,\calS}_p\\
&\leq & 4k\cdot \sum_{C}\sum_{p\in \Omega\cap C\cap G} \frac{|C\cap G|}{\delta} \cdot u^{G,\calS}_p\\
&\leq & 4k \cdot \|u^{\calS}\|_1,
\end{eqnarray*}
Therefore in this case
\begin{equation}
\label{eq:lemq3}
\left\vert\frac{\sum_{p\in \Omega} w_p \cdot u^{G,\calS}_p - \|u^{G,\calS}\|_1}{\cost(G,\greedy)+\cost(G,\calS)}\right\vert \leq \frac{4k\cdot \|u^{\calS}\|_1}{\cost(G,\greedy)+\cost(G,\calS)} \leq 4k
\end{equation}

Due to \cref{lem:ksize}, $\pr_\coreset[\overline{\calE_G}]\leq k\cdot \exp\left(-\frac{\varepsilon^2}{9\cdot k}\delta\right)$. Hence, if we set $\delta \geq 9 \varepsilon^{-2} k\log \frac{4k^2}{\varepsilon}$, we have $\pr[\overline{\calE_G}] \leq  \frac{\varepsilon}{4k}$. 
This implies together with Equations~\ref{eq:lemq2} and \ref{eq:lemq3}
\begin{equation}
\label{eq:lemq4}
\mathbb{E}_{\coreset}~\underset{\calS}{\text{sup}}\left[\left\vert\frac{\sum_{p\in \Omega} w_p \cdot q^{G,\calS}(p) - \|q^{G,\calS}\|_1}{\cost(G,\greedy)+\cost(G,\calS)}\right\vert~|~\overline{\calE_G}\right]\cdot \pr_\coreset[\overline{\calE_G}]  \leq 4k \cdot \frac{\varepsilon}{4k} \leq \varepsilon.
\end{equation}
The claim now follows by combining Equations~\ref{eq:lemq1} and~\ref{eq:lemq4} and rescaling $\varepsilon$.
\end{proof}

We now turn our attention to \cref{lem:qO}. Henceforth, we let $G\in G^O$. We first require an analogue of event $\calE_G$. 
We define event $\calE_{far,G}$ that for all clusters $C$ with
\begin{equation*}
\sum_{p\in C\cap G\cap \coreset} \frac{\cost(G,\greedy)}{\delta\cdot\cost(p,\greedy)} \cost(p,\greedy) = (1\pm \varepsilon)\cdot \cost(C\cap G,\greedy).
\end{equation*}

Furthermore, by definition of the groups, we have
\begin{equation}
\label{eq:ksizeO1}
\cost(G,\greedy) \leq 2k\cdot \cost(C\cap G,\greedy).
\end{equation}

We start by bounding the probability that $\calE_{far,G}$ fails to occur.

\begin{lemma}
\label{lem:eventEFar}
Event $\calE_{far,G}$ happens with probability at least $1 - k\exp(\frac{\eps^2}{5\cdot k}\cdot \delta)$.
\end{lemma}
\begin{proof}
Again, we aim to use Bernstein's Inequality. Let $p_j$ be the $j$th point in the sample $\coreset$ with respect to arbitrary but fixed ordering. Consider the random variable \newline $w_{p_j,C}  = \begin{cases} w_p \cdot \cost(p,\greedy) & \text{if } p_j=p\in C\cap G \\ 0 & \text{else}\end{cases}$. Then:
\begin{eqnarray*}
E[w_{p_j,C}^2]  &=& \sum_{p\in C\cap G} \left(\frac{\cost(G,\greedy)}{\delta\cdot \cost(p,\greedy)} \cdot \cost(p,\greedy)\right)^2 \cdot \pr[p \in  \coreset] \\
&=& \frac{\cost(G,\greedy)}{\delta^2} \cdot\sum_{p\in C\cap G}\cost(p,\greedy)  \\
&=& \frac{\cost(G,\greedy)}{\delta^2} \cost(C\cap G,\greedy)  \\
(Eq.~\ref{eq:ksizeO1}) &\leq & \frac{2k}{\delta^2}\cdot \cost^2(C\cap G,\greedy)
\end{eqnarray*}

Furthermore, we have by the same argument the following upper bound for the maximum value any of the $w_{p_j,C}$:
\begin{equation*}
M:= \max_{p\in C\cap G} \frac{\cost(G,\greedy)}{\delta\cdot \cost(p,\greedy)} \cdot \cost(p,\greedy) \leq \frac{2k}{ \delta}\cdot \cost(C\cap G,\greedy).
\end{equation*}

Combining both bounds with Bernstein's inequality now yields 
\begin{eqnarray*}
& &\mathbb{P}[|\cost(C\cap G\cap \coreset,\greedy) - \cost(C\cap G,\greedy)| \leq \varepsilon\cdot \cost(C\cap G,\greedy)] \\
&\leq & \exp\left(-\frac{\varepsilon^2 \cdot \cost^2(C\cap G,\greedy)}{2 \sum_{i=1}^\delta Var[X_i] + \frac{1}{3} M \cdot \varepsilon\cdot \cost(C\cap G,\greedy) }\right) \leq \exp\left(-\frac{\varepsilon^2}{5\cdot k} \cdot \delta\right)
\end{eqnarray*}

Reformulating, we now have 

\begin{equation*}
\sum_{p\in C\cap G\cap \coreset} \frac{\cost(G,\greedy)}{\delta\cdot\cost(p,\greedy)} \cost(p,\greedy) = (1\pm \varepsilon)\cdot \cost(C\cap G,\greedy).
\end{equation*}
Taking a union bound over all clusters yields the claim.
\end{proof}

\begin{lemma}
\label{lem:outerfar}
Condition on event $\calE_{far,G}$. Suppose $C\in F_{G,\calS}$.
Then 
$$\cost(C\cap G,\calS) + \sum_{p\in C\cap G\cap \coreset} w_{p} \cdot \cost(p,\calS) \leq \varepsilon\cdot \cost(C,\calS).$$
\end{lemma}
\begin{proof}
First, we fix a cluster $C \in \greedy$, and show that points of $C\cap G_{far, \calS}$ are very cheap compared to $\cost(C,\calS)$, assuming that $C\in F_{G,\calS}$.
Let $c$ be the center serving $p\in G_{far, \calS}\cap C$ in $\greedy$.
Let $C_{close}$ be the points of $C$ with cost at most $\left(\frac{2z}{\varepsilon}\right)^z \cdot \frac{\cost(C,c)}{|C|}$.
Consider an arbitrary point in $p'\in C_{close}$.
Due to the triangle inequality and $\cost(p,\calS) > 4^z \cdot \cost(p,c)$, we have $\dist(c,\calS) \geq \dist(p,\calS)  - \dist(p,c) \geq 4 \dist(p,c) - \dist(p,c) \geq \dist(p,c)$. Therefore
$\cost(c,\calS) \geq \left(\frac{4z}{\varepsilon}\right)^{2z} \cdot \frac{\cost(C,c)}{|C|}$. Using this and \cref{lem:weaktri} we now have for any $p'\in C_{close}$ 
\begin{eqnarray}
\nonumber
\cost(c,\calS) &\leq & (1+\varepsilon)\cdot \cost(p',\calS) + \left(\frac{2z+\varepsilon}{\varepsilon}\right)^{z-1}\cdot \cost(p',c) \\
\nonumber
&\leq & (1+\varepsilon)\cdot \cost(p',\calS) + \left(\frac{2z+\varepsilon}{\varepsilon}\right)^{z-1} \cdot \left(\frac{2z}{\varepsilon}\right)^{z} \cdot \frac{\cost(C,c)}{|C|}  \\
\nonumber
&\leq & (1+\varepsilon) \cdot \cost(p',\calS) +  \frac{\left(\frac{4z}{\varepsilon}\right)^{2z-1} \cdot \frac{\cost(C,c)}{|C|}}{\cost(p,c)} \cdot \cost(p,c) \\
\nonumber
&\leq & (1+\varepsilon) \cdot \cost(p',\calS) + \varepsilon \cdot \cost(p,c) \qquad \text{ since } p \in G\in G^O \\
\nonumber
&\leq & (1+\varepsilon) \cdot \cost(p',\calS) + \varepsilon \cdot \cost(c,\calS) \\
\label{eq:outerfar1}
\Rightarrow \cost(p',\calS) & \geq &  \frac{1-\varepsilon}{1+\varepsilon} \cdot \cost(c,\calS) 
\end{eqnarray}

We now bound $\cost(C,\calS)$ in terms of $\cost(C,c)$. We have due to Markov's inequality $|C\cap G|\leq \left(\frac{\varepsilon}{4z}\right)^{2z}$ and $|C_{close}| \geq (1-\varepsilon)\cdot |C|$ and therefore
\begin{eqnarray}
\label{eq:outerfar2}
\cost(C,\calS) &\geq & \cost(C_{close},\calS) = \sum_{p'\in C_{close}} \cost(p',\calS) \geq  |C_{close}|\cdot \frac{1-\varepsilon}{1+\varepsilon} \cdot cost(c,\calS) \\
\label{eq:outerfar3}
&\geq & |C_{close}| \cdot \frac{1-\varepsilon}{1+\varepsilon} \cdot \left(\frac{4z}{\varepsilon}\right)^{2z}  \cdot \frac{\cost(C,c)}{|C|} \geq \left(\frac{4z}{\varepsilon}\right)^{2z-1}  \cdot \cost(C,c) 
\end{eqnarray}

which yields for any $C\in G_{far,\calS}$.

\begin{eqnarray}
\nonumber
& &\cost(C\cap G ,\calS) = \sum_C \sum_{p\in C\cap G} \cost(p,\calS) \\
\nonumber
(\cref{lem:weaktri}) &\leq & \sum_{p\in C\cap G} (1+\varepsilon)\cdot \cost(c,\calS) + \left(\frac{2z+\varepsilon}{\varepsilon}\right)^{z-1}\cdot \cost(p,c)  \\
\nonumber
&\leq & |C\cap G| \cdot (1+\varepsilon)\cdot  \cost(c,\calS)  +  \left(\frac{2z+\varepsilon}{\varepsilon}\right)^{z-1}  \cdot \cost(C\cap G,c) \\
\label{eq:outerfar4}
(\text{Markov})&\leq & (1+\varepsilon)\cdot \left(\frac{\varepsilon}{2z}\right)^{2z} \cdot |C| \cdot \cost(c,\calS)  +  \left(\frac{2z+\varepsilon}{\varepsilon}\right)^{z-1}  \cdot \cost(C\cap G,c) \\
\nonumber
(\text{Markov}) &\leq & \frac{1+\eps}{1-\eps} \cdot \left(\frac{\varepsilon}{2z}\right)^{2z} \cdot |C_{close}| \cdot \cost(c,\calS)  +  \left(\frac{2z+\varepsilon}{\varepsilon}\right)^{z-1}  \cdot \cost(C\cap G,c) 
\end{eqnarray}
\begin{eqnarray}
\nonumber
(Eq.~\ref{eq:outerfar2})&\leq & \frac{(1+\varepsilon)^2}{(1-\varepsilon)^2}\cdot \left(\frac{\varepsilon}{2z}\right)^{2z}\cdot \cost(C,\calS)  +  \left(\frac{2z+\varepsilon}{\varepsilon}\right)^{z-1}  \cdot \cost(C\cap G,c) \\
(Eq.~\ref{eq:outerfar3}) & \leq & \frac{(1+\varepsilon)^2}{(1-\varepsilon)^2}\cdot \left(\frac{\varepsilon}{2z}\right)^{2z}\cdot \cost(C,\calS)  + \left(\frac{2z+\varepsilon}{\varepsilon}\right)^{z-1}  \cdot \left(\frac{\varepsilon}{4z}\right)^{2z-1} \cdot \cost(C,\calS) \\
\label{eq:outerfar5}
&\leq & \varepsilon \cdot \cost(C,\calS)  
\end{eqnarray}

What is left to show is that the weighted cost of the points in $G_{far, \calS} \cap \coreset$ can be bounded similarly.
For that, we use event $\calE_{far,G}$ to show that $\sum_{p\in G_{far, \calS}\cap C\cap \coreset} \frac{\cost(G,\A_0)}{\cost(p,\A_0)} \approx |G_{far, \calS}\cap C|$. 
We have for all clusters $C$ induced by $\greedy$

\begin{eqnarray}
\nonumber
\sum_{p\in C\cap G\cap \coreset} \frac{\cost(G,\greedy)}{\delta \cdot \cost(p,\greedy)} \cdot \left(\frac{2z}{\varepsilon}\right)^{2z} \cdot \frac{\cost(C,\greedy)}{|C|} &\leq & \sum_{p\in C\cap G\cap \coreset} \frac{\cost(G,\greedy)}{\delta\cdot\cost(p,\greedy)} \cost(p,\greedy) \\
\nonumber
& \leq &(1+\varepsilon) \cdot \cost(C\cap G,\greedy) \\
\nonumber
\Rightarrow \sum_{p\in C\cap G\cap\coreset} \frac{\cost(G_j,\greedy)}{\delta\cdot\cost(p,\greedy)} &\leq & (1+\varepsilon)\cdot \left(\frac{\varepsilon}{2z}\right)^{2z} \cdot |C| \frac{\cost(C\cap G,\greedy)}{\cost(C,\greedy)} \\
\label{eq:outerfar7}
& \leq  & (1+\varepsilon)\cdot \left(\frac{\varepsilon}{2z}\right)^{2z} \cdot |C|
\end{eqnarray}

Therefore, we have

\begin{eqnarray}
\nonumber
& &\cost(G_{far, \calS}\cap \coreset\cap C,\calS) = \sum_{p\in G_{far, \calS}\cap C} \frac{\cost(G,\greedy)}{\delta\cdot\cost(p,\greedy)}\cdot \cost(p,\calS) \\
\nonumber
(\cref{lem:weaktri}) &\leq & \sum_{p\in G_{far, \calS}\cap \coreset\cap C} \frac{\cost(G,\greedy)}{\delta\cdot\cost(p,\greedy)}\cdot \left((1+\varepsilon) \cdot \cost(c,\calS) + \left(\frac{2z+\varepsilon}{\varepsilon}\right)^{z-1}\cdot \cost(p,c)\right)  \\
\nonumber
&\leq & (1+\varepsilon)\cdot  \cost(c,\calS) \cdot \sum_{p\in G_{far, \calS}\cap \coreset \cap C} \frac{\cost(G,\greedy)}{\delta\cdot\cost(p,\greedy)} \\
\nonumber
(\calE_{far,G})& &+ \left(\frac{2z+\varepsilon}{\varepsilon}\right)^{z-1}  \cdot (1+\varepsilon) \cdot \cost(C\cap G,\greedy) \\
\nonumber
(Eq.~\ref{eq:outerfar7})&\leq & (1+\varepsilon)^2 \cdot  \cost(c,\calS) \cdot \left(\frac{\varepsilon}{2z}\right)^{2z} \cdot |C| +  \left(\frac{2z+\varepsilon}{\varepsilon}\right)^{z-1}  \cdot (1+\varepsilon) \cdot \cost(C\cap G,\greedy) \\
\label{eq:outerfar8}
&\leq & (1+\varepsilon)^2\cdot \left(\frac{\varepsilon}{2z}\right)^{2z} \cdot |C| \cdot \cost(c,\calS)  +  \left(\frac{2z+\varepsilon}{\varepsilon}\right)^{z-1}  \cdot \cost(C,c) \\
\nonumber
&\leq & \varepsilon \cdot \cost(C,\calS)  
\end{eqnarray}

where the steps following Equation~\ref{eq:outerfar8} are identical to those used to derive Equation~\ref{eq:outerfar5} from Equation~\ref{eq:outerfar4}.
Summing up Equations \ref{eq:outerfar5} and \ref{eq:outerfar8} and rescaling $\varepsilon$ by a factor $2$ yields the claim.
\end{proof}

\begin{proof}[Proof of \cref{lem:qO}]
Similar to the proof of \cref{lem:q}, we bound the expectation when conditioning on $\calE_{far,G}$ and when $\calE_{far,G}$ fails to hold:
\begin{eqnarray}
\nonumber
& &\mathbb{E}_{\coreset}~\underset{\calS}{\sup}\left[\left\vert\frac{\sum_{p\in \Omega} w_p \cdot u^{G,\calS}_p - \|u^{G,\calS}\|_1}{\cost(P^G,\greedy)+\cost(P^G,\calS)}\right\vert\right]  \\
\label{eq:qO1}
&= &\mathbb{E}_{\coreset}~\underset{\calS}{\sup}\left[\left\vert\frac{\sum_{p\in \Omega} w_p \cdot u^{G,\calS}_p - \|u^{G,\calS}\|_1}{\cost(P^G,\greedy)+\cost(P^G,\calS)}\right\vert \calE_{far,G}\right] \cdot \pr\left[\calE_{far,G}\right] \\
\label{eq:qO2}
& &+ \mathbb{E}_{\coreset}~\underset{\calS}{\sup}\left[\left\vert\frac{\sum_{p\in \Omega} w_p \cdot u^{G,\calS}_p - \|u^{G,\calS}\|_1}{\cost(P^G,\greedy)+\cost(P^G,\calS)}\right\vert \overline{\calE_{far,G}}\right] \cdot \pr\left[\overline{\calE_{far,G}}\right]
\end{eqnarray}

For term \ref{eq:qO1}, \cref{lem:outerfar} states that \\
\begin{eqnarray}
\nonumber
\sum_{p\in \Omega}w_p\cdot u_{p}^{G,\calS} + \|u^{G,\calS}\|_1 &=& \sum_{C\in F_{G,\calS}}\sum_{p\in \Omega \cap G}w_p\cdot u_{p}^{G,\calS} + \sum_{p\in C\cap G} \cost(p,\calS) \\
\label{eq:qO3}
&\leq & \sum_{C\in F_{G,\calS}} \varepsilon \cdot \cost(C\cap G,\calS) = \varepsilon \cdot \cost(P^G,\calS).
\end{eqnarray}

We now consider term \ref{eq:qO2}. If $\|u^{G,S}\|_1 > \sum_{p\in \Omega} w_p \cdot u^{G,\calS}_p$, we can bound $\frac{\sum_{p\in \Omega} w_p \cdot u^{G,\calS}_p - \|u^{G,\calS}\|_1}{\cost(P^G,\greedy)+\cost(P^G,\calS)}$ by $1$.
Otherwise, let $r_C= \underset{p\in C\cap G}{\max} \frac{\cost(p,\calS)}{\cost(p,\greedy}>4^z$ and let $p' = \underset{p\in C\cap G}{\text{argmax}} \frac{\cost(p,\calS)}{\cost(p,\greedy}$.
We have $\dist(c,\calS) \geq \dist(p',\calS) - \dist(p',c) \geq (r_C^{1/z} -1)\cdot \dist(p',c)$, which implies
$\frac{\cost(c,\calS)}{\cost(p',c)} \cdot 2^z \geq r_C$. Therefore
\begin{eqnarray*}
\sum_{p\in \coreset} w_p u^{G,\calS}_p &=& \sum_{C} \sum_{p\in \coreset\cap C} \frac{\cost(G,\greedy)}{\delta \cdot \cost(p,\greedy)} \cdot \cost(p,\calS)  \\
(Eq.~\ref{eq:ksizeO1}) &\leq & 4k\cdot \sum_{C} \underset{p\in C\cap G}{\max} \cost(C\cap G,\greedy) \cdot r_C \\
&\leq & 4k\cdot \sum_{C} \underset{p\in C\cap G}{\max} \cost(C\cap G,\greedy) \cdot 2^z \cdot \frac{\cost(c,\calS)}{\cost(p',\greedy)} \\
&\leq & 2^{z+2}k\cdot \sum_{C} \underset{p\in C\cap G}{\max} \frac{\cost(C,\greedy)}{\cost(p',\greedy)} \cdot \cost(c,\calS) \\
(Markov) &\leq & 2^{z+2}k\cdot \sum_{C} \left(\frac{\varepsilon}{4z}\right)^{2z}  |C|  \cdot \cost(c,\calS) \\
(\cref{lem:weaktri}) &\leq & 2^{2z+2}k\cdot \sum_{C} (\cost(C,\greedy) + \cost(C,\calS)) \\
&\leq &  2^{2z+2}k\cdot (\cost(P^G,\greedy) +\cost(P^G,\calS) )
\end{eqnarray*}
With this, we may bound the ratio $\frac{\sum_{p\in \Omega} w_p \cdot u^{G,\calS}_p - \|u^{G,\calS}\|_1}{\cost(P^G,\greedy)+\cost(P^G,\calS)}$ by $2^{2+2}k$.
The probability of $\overline{\calE_{far,G}}$ is at most $k\cdot \exp\left(-\frac{\varepsilon^2}{5\cdot k}\cdot \delta\right)$ due to \cref{lem:eventEFar}. Therefore setting $\delta > 5k\cdot \log \frac{k^2}{2^{2z+2}\cdot \varepsilon}$ yields
$$\mathbb{E}_{\coreset}~\underset{\calS}{\sup}\left[\left\vert\frac{\sum_{p\in \Omega} w_p \cdot u^{G,\calS}_p - \|u^{G,\calS}\|_1}{\cost(P^G,\greedy)+\cost(P^G,\calS)}\right\vert \overline{\calE_{far,G}}\right] \cdot \pr\left[\overline{\calE_{far,G}}\right] \leq 2^{2z+2}k\cdot \frac{\varepsilon}{2^{2z+2}\cdot k} \leq \varepsilon.$$
Summing this with Equation~\ref{eq:qO3} and rescaling $\varepsilon$ by a factor of $2$ yields the claim.
\end{proof}

%% file: appendix.tex
 \section{Lower bound for Arbitrary Powers in Euclidean Spaces}\label{ap:euclidean}
In this section, we generalize the lower bound to arbitrary powers $z
\neq 2$. The proof follows exactly the same steps as for $z=2$, except
that we make use of the following observation to handle $z \neq 2$:
 \begin{observation}
   \label{obs:p}
   For any $0 < a \leq 1$ any $b > 0$ and any $x \in [0,b]$ we have $b^a(1-x/b)
   \leq (b-x)^a \leq b^a(1-xa/b)$. For any $1 \leq a$ any $b > 0$ and any $x \in
   [0,b]$ we have $b^a(1-xa/b) \leq (b-x)^a \leq b^a(1-x/b)$.
 \end{observation}
 \begin{proof}
    For any $0 < a \leq 1$ any $b > 0$ and any $x \in [0,b]$, we have $(b-x)^a =
   b^a(1-x/b)^a = b^a \exp\left(-a \sum_{n=1}^\infty
     (x/b)^n/n\right)$. Since $a \leq 1$, this is at most $b^a
   \exp\left(-\sum_{n=1}^\infty (xa/b)^n/n\right) =
   b^a(1-xa/b)$. Also, since $0 \leq 1-x/b \leq 1$, it holds for any
   $0 < a \leq 1$ that $(1-x/b)^a \geq 1-x/b$.

   For any $1 \leq a$ any $b > 0$ and any $x \in [0,b]$, we have $(b-x)^a = b^a(1-x/b)^a = b^a \exp\left(-a \sum_{n=1}^\infty
     (x/b)^n/n\right)$. Since $a \geq 1$, this is at least $b^a
   \exp\left(-\sum_{n=1}^\infty (xa/b)^n/n\right) =
   b^a(1-xa/b)$. Also, since $0 \leq 1-x/b \leq 1$, it holds for any
   $1 \leq a$ that $(1-x/b)^a \leq 1-x/b$.
 \end{proof}

 The first step of our proof is again to argue that for any coreset
 using few points, there is a ``cheap'' clustering using a single
 center of unit norm:

\begin{lemma}
  \label{plem:highips}
  Let $r_1,\dots,r_\ell \in \R^{2d}$ and let $w_1,\dots,w_\ell \in
  \R^+$. There exists a unit vector $v$ such that $\sum_{i=1}^\ell w_i
  \min_{\xi \in \{-1,1\}} \|r_i - \xi v\|_2^z \leq \sum_{i=1}^\ell w_i
                                          (\|r_i\|_2^2 +
                 1)^{z/2} - 2 \min\{1,z/2\} \frac{\sum_{i=1}^\ell w_i  (\|r_i\|_2^2+1)^{z/2-1}\|r_i\|_2}{\sqrt{\ell}}$.
\end{lemma}

\begin{proof}
Consider the random vector $u = \sum_{i=1}^\ell w_i (\|r_i\|_2^2+1)^{z/2-1} \sigma_i r_i$ where the $\sigma_i$ are i.i.d. uniform Rademachers.
  We see that
  \begin{eqnarray*}
    \sum_{i=1}^\ell w_i (\|r_i\|_2^2+1)^{z/2-1} |\langle r_i, u\rangle| &=& \sum_{i=1}^\ell w_i (\|r_i\|_2^2+1)^{z/2-1} \left| \sum_{j=1}^\ell w_j  (\|r_j\|_2^2+1)^{z/2-1}\sigma_j \langle r_i, r_j \rangle\right| \\
                                             &=& \sum_{i=1}^\ell w_i
                                                 (\|r_i\|_2^2+1)^{z/2-1}
                                                 \left| \sum_{j=1}^\ell w_j  (\|r_j\|_2^2+1)^{z/2-1}\sigma_i \sigma_j \langle r_i, r_j \rangle\right| \\
                                             &\geq& \sum_{i=1}^\ell w_i  (\|r_i\|_2^2+1)^{z/2-1}\sum_{j=1}^\ell w_j  (\|r_j\|_2^2+1)^{z/2-1}\sigma_i \sigma_j \langle r_i, r_j \rangle \\
                                             &=& \|u\|_2^2.                                                
  \end{eqnarray*}
  We may then define the unit vector $v = u/\|u\|_2$ (with $v=0$ when $u=0$) and conclude that
$$
\sum_{i=1}^\ell w_i  (\|r_i\|_2^2+1)^{z/2-1}|\langle r_i, v \rangle| \geq \|u\|_2.
$$
Since $\E[\|u\|_2^2] = \sum_{i=1}^\ell w_i^2   (\|r_i\|_2^2+1)^{z-2}\|r_i\|_2^2$ we conclude that there must exist a unit vector $v$ with
$$
\sum_{i=1}^\ell w_i  (\|r_i\|_2^2+1)^{z/2-1}|\langle r_i, v \rangle| \geq \sqrt{\sum_{i=1}^\ell w_i^2  (\|r_i\|_2^2+1)^{z-2}\|r_i\|_2^2}.
$$
By Cauchy-Schwartz, we have:
$$
\sum_{i=1}^\ell | 1 \cdot w_i  (\|r_i\|_2^2+1)^{z/2-1}\|r_i\|_2 | \leq \sqrt{\sum_{i=1}^\ell w_i^2  (\|r_i\|_2^2+1)^{z-2}\|r_i\|_2^2 } \cdot \sqrt{\sum_{i=1}^\ell 1 } = \sqrt{\sum_{i=1}^\ell w_i^2  (\|r_i\|_2^2+1)^{z-2}\|r_i\|_2^2 } \cdot \sqrt{\ell}
$$
which finally implies
$$
\sum_{i=1}^\ell w_i  (\|r_i\|_2^2+1)^{z/2-1}|\langle r_i, v \rangle | \geq \frac{\sum_{i=1}^\ell w_i  (\|r_i\|_2^2+1)^{z/2-1}\|r_i\|_2}{\sqrt{\ell}}.
$$
For that unit vector $v$, consider $\sum_{i=1}^\ell w_i \min_{\xi
  \in \{-1,1\}} \|r_i -
\xi v\|_2^z$:
\begin{eqnarray*}
  \sum_{i=1}^\ell w_i \min_{\xi \in \{-1,1\}} \|r_i - \xi v\|_2^z &=& \sum_{i=1}^\ell w_i
                                          (\|r_i\|_2^2 + \|v\|_2^2  -2|
                                          \langle r_i,
                                                                            v\rangle|)^{z/2}
  \\
  &=& \sum_{i=1}^\ell w_i
                                          (\|r_i\|_2^2 + 1  -2|
                                          \langle r_i, v\rangle|)^{z/2}.
\end{eqnarray*}
By Observation~\ref{obs:p}, this is at most:
\begin{eqnarray*}
  &\leq& \sum_{i=1}^\ell w_i
                                          (\|r_i\|_2^2 +
                 1)^{z/2}\left(1 -\frac{2\min\{1,z/2\}|
                                          \langle r_i,
         v\rangle|}{\|r_i\|_2^2 + 1} \right) \\
  &\leq& \sum_{i=1}^\ell w_i
                                          (\|r_i\|_2^2 +
                 1)^{z/2} - \sum_{i=1}^\ell w_i 2 \min\{1,z/2\}
         |\langle r_i, v \rangle| (\|r_i\|_2^2 + 1)^{z/2-1} \\
  &\leq& \sum_{i=1}^\ell w_i
                                          (\|r_i\|_2^2 +
                 1)^{z/2} - 2 \min\{1,z/2\} \frac{\sum_{i=1}^\ell w_i  (\|r_i\|_2^2+1)^{z/2-1}\|r_i\|_2}{\sqrt{\ell}}.
\end{eqnarray*}
\end{proof}

We now extend this to create a cheap clustering using $k$ centers of
unit norm:

\begin{lemma}
  \label{plem:largereduction}
  Let $r_1,\dots,r_t \in \R^{2d}$ and let $w_1,\dots,w_t \in \R^+$. There
  exists a set of $k$ unit vectors $v_1,\dots,v_k$ such that
  $$
  \sum_{i=1}^t w_i \min_{j=1}^k \|r_i-v_j\|_2^z \leq \sum_{i=1}^t w_i
                                          (\|r_i\|_2^2 +
                 1)^{z/2} - \min\{1,z/2\} \sqrt{2k/t}  \sum_{i=1}^t
                   w_i
                   (\|r_i\|_2^2+1)^{z/2-1}\|r_i\|_2.
                   $$
                   and moreover, for every $v_j$, there is a $v_i$ such that $v_j = -v_i$.
\end{lemma}

\begin{proof}
  Partition $r_1,\dots,r_t$ arbitrarily into $k/2$ disjoint groups
  $G_1,\dots,G_{k/2}$ of at most $2t/k$ vectors each. For each group
  $G_j$, apply Lemma~\ref{plem:highips} to find a unit vector $u_j$
  with
  $$
  \sum_{r_i \in G_j}^\ell w_i
  \min_{\xi \in \{-1,1\}} \|r_i - \xi u_j\|_2^z \leq \sum_{r_i \in G_j} w_i
                                          (\|r_i\|_2^2 +
                 1)^{z/2} - 2 \min\{1,z/2\} \frac{\sum_{r_i \in G_j}
                   w_i
                   (\|r_i\|_2^2+1)^{z/2-1}\|r_i\|_2}{\sqrt{2t/k}}.$$
Let $v_{2j-1} = u_j$ and $v_{2j}=-u_j$. Since we always add both $u_j$ and $-u_j$ we conclude:
\begin{eqnarray*}
  \sum_{i=1}^t w_i \min_{j=1}^k \|r_i-v_j\|_2^z &\leq& \\
                                                                  \sum_{j=1}^{k/2}
                                                                  \sum_{r_i
                                                                  \in
                                                                  G_j}
                                                                  w_i
                                                                  \min_{\xi
                                                                  \in
                                                                  \{-1,1\}}
                                                                  \|r_i-\xi u_j\|_2^z
                                                &\leq&\\
  \sum_{i=1}^t w_i
                                          (\|r_i\|_2^2 +
                 1)^{z/2} - \min\{1,z/2\} \sqrt{2k/t}  \sum_{i=1}^t
                   w_i
                   (\|r_i\|_2^2+1)^{z/2-1}\|r_i\|_2.
\end{eqnarray*}
\end{proof}

We now use the orthogonality of the standard unit vectors
$e_1,\dots,e_d$ to argue that any clustering of them using unit norm
centers must be expensive:

 \begin{lemma}
  \label{plem:standardhard}
For any $d$, consider the point set $P=\{e_1,\dots,e_d\}$ in $\R^{2d}$. For any set of $k$ centers $c_1,\dots,c_k \in \R^{2d}$, all with unit norm and satisfying that for every $c_j$ there is an index $i$ such that $c_j = -c_i$, it holds that $\sum_{i=1}^d \min_{j=1}^k \|e_i - c_j \|_2^z \geq 2^{z/2} d - 2^{z/2} \cdot \max\{1,z/2\} \cdot\sqrt{dk}$.
\end{lemma}
\begin{proof}
  We see that
  \begin{eqnarray*}
    \sum_{i=1}^n \min_{j=1}^k \|e_i - c_j \|_2^p &=& \sum_{i=1}^d
                                                     \min_{j=1}^k
                                                     \left(\|e_i\|_2^2
                                                     + \|c_j \|_2^2
                                                     -2\langle e_i,
                                                     c_j \rangle \right)^{z/2} \\
                                                 &=& \sum_{i=1}^d
                                                     \left(2
                                                     -2 \max_{j=1}^k \langle e_i,
                                                     c_j \rangle
                                                     \right)^{z/2}.
  \end{eqnarray*}
  Since $c_1,\dots,c_k$ satisfy that for every $c_j$ there is an index $h$ with $c_j=-c_h$, it holds that $\max_{j=1}^k \langle e_i, c_j \rangle \geq 0$ for every $e_i$.
  By Cauchy-Schwartz, we have $|\langle e_i ,c_j \rangle| \leq 1$
  hence by Observation~\ref{obs:p}, the above is at least:
  \begin{eqnarray*}
    2^{z/2} d - 2^{z/2} \cdot \max\{1,z/2\} \cdot \sum_{i=1}^d \max_{j=1}^k  \langle e_i,
                                                     c_j \rangle.
    \end{eqnarray*}
  Now, for each $c_j$, define $\hat{c}_j$ to equal $c_j$, except that we set the $i$'th coordinate to $0$ if $j \neq \argmax_h \langle e_i, c_h \rangle$ or $i> d$. Then:
  \begin{eqnarray*}
    2^{z/2} d - 2^{z/2} \cdot \max\{1,z/2\} \cdot \sum_{i=1}^d \max_{j=1}^k  \langle e_i,c_j \rangle
    &=& 2^{z/2}d - 2^{z/2} \cdot \max\{1,z/2\} \cdot \sum_{i=1}^d \sum_{j=1}^k \langle e_i, \hat{c}_j \rangle \\
    &=& 2^{z/2} d - 2^{z/2} \cdot \max\{1,z/2\} \cdot \sum_{i=1}^d \langle e_i, \sum_{j=1}^k \hat{c}_j \rangle \\
    &\geq& 2^{z/2} d - 2^{z/2} \cdot \max\{1,z/2\} \cdot\|\sum_{j=1}^k \hat{c}_j\|_1.
  \end{eqnarray*}
  By Cauchy-Schwartz, we have $\|\sum_{j=1}^k \hat{c}_j\|_1 \leq \| \sum_{j=1}^k \hat{c}_j\|_2 \cdot \sqrt{d}$. Since the $\hat{c}_j$'s are orthogonal and have norm at most $1$, we have $\|\sum_{j=1}^k \hat{c}_j \|_2 \leq \sqrt{k}$. Thus we conclude $\sum_{i=1}^d \min_{j=1}^k \|e_i - c_j \|_2^z \geq 2^{z/2} d - 2^{z/2} \cdot \max\{1,z/2\} \cdot\sqrt{dk}$. 
\end{proof}

We also need a handle on the offset of any coreset. This is obtained
by considering a clustering using a single center that is orthogonal
to all points $e_1,\dots,e_d$ and all points of a coreset:

\begin{lemma}
  \label{plem:FandSquare}
  For any $d$, consider the point set $P=\{e_1,\dots,e_d\}$ in $\R^{2d}$. Let $r_1,\dots,r_t \in \R^{2d}$ and let $w_1,\dots,w_t \in \R^+$ be an $\eps$-coreset for $P$, using offset $\Delta$ and with $t < d$. Then we must have $\Delta + \sum_{i=1}^t w_i (\|r_i\|_2^2+1)^{z/2} \in (1 \pm \eps)2^{z/2}d$.
\end{lemma}

\begin{proof}
Since $t+d < 2d$ there exists a unit vector $v$ that is orthogonal to all $r_i$ and all $e_j$.
Consider placing all $k$ centers at $v$. Then the cost of clustering $P$ with these centers is $2^{z/2}d$. It therefore must hold that $\Delta + \sum_{i=1}^t w_i (\|r_i\|_2^2 + \|v\|_2^2 - 2\langle r_i ,v\rangle)^{z/2} = \Delta+\sum_{i=1}^t w_i (\|r_i\|_2^2+1)^{z/2} \in (1\pm \eps)2^{z/2}d$.
\end{proof}

\begin{lemma}
  \label{plem:uppernorm}
  For any $d$ and any $k > 1$, let $P=\{e_1,\dots,e_d\}$ in $\R^{2d}$. Let $r_1,\dots,r_t \in \R^{2d}$ and let $w_1,\dots,w_t \in \R^+$ be an $\eps$-coreset for $P$ with $t < d$, using offset $\Delta$. Then
  $$
  \sum_{i=1}^t w_i (\|r_i\|_2^2+1)^{z/2-1}\|r_i\|_2 \leq \frac{2 \eps 2^{z/2} d + \max\{1,z/2\} 2^{z/2} \sqrt{dk} }{\sqrt{2} \cdot \min\{1,z/2\}} \cdot \sqrt{t/k}.
  $$
\end{lemma}

\begin{proof}
  By Lemma~\ref{plem:largereduction}, we can find $k$ unit vectors
  $v_1,\dots,v_k$ such that
  $$
  \sum_{i=1}^t w_i \min_{j=1}^k \|r_i-v_j\|_2^z \leq \sum_{i=1}^t w_i
                                          (\|r_i\|_2^2 +
                 1)^{z/2} - \min\{1,z/2\} \sqrt{2k/t}  \sum_{i=1}^t
                   w_i
                   (\|r_i\|_2^2+1)^{z/2-1}\|r_i\|_2.
                   $$
                   Moreover, those vectors satisfy that for every $v_j$, there is an index $i$ such that $v_j = -v_i$.
By Lemma~\ref{plem:standardhard}, it holds that $\sum_{p
    \in P} \min_{j=1}^k \|p - v_j\|_2^z \geq 2^{z/2}d-2^{z/2} \cdot
  \max\{1, z/2\} \cdot \sqrt{dk}$. Since points $r_1,\dots,r_t$ with respective weights $w_1,\dots,w_t$ and offset $\Delta$ form an $\eps$-coreset for $P$, it follows from Observation~\ref{obs:p} that we must have 
\begin{eqnarray*}
  (1-\eps)2^{z/2}(d-\max\{1,z/2\} \cdot \sqrt{dk})&\leq& \\
  \Delta + \sum_{i=1}^t \min_{j=1}^k w_i \|r_i - v_j\|_2^z
                                                 &\leq& \\
  \Delta + \sum_{i=1}^t w_i
                                          (\|r_i\|_2^2 +
                 1)^{z/2} - \min\{1,z/2\} \sqrt{2k/t}  \sum_{i=1}^t
                   w_i
                   (\|r_i\|_2^2+1)^{z/2-1}\|r_i\|_2.
\end{eqnarray*}
By Lemma~\ref{plem:FandSquare}, this is at most:
$$
(1+\eps)2^{z/2}d- \min\{1,z/2\} \sqrt{2k/t}  \sum_{i=1}^t
                   w_i
                   (\|r_i\|_2^2+1)^{z/2-1}\|r_i\|_2.
$$
We have therefore shown that
\begin{eqnarray*}
  (1-\eps)2^{z/2}(d-\max\{1,z/2\} \cdot \sqrt{dk}) &\leq& (1+\eps)2^{z/2}d- \min\{1,z/2\} \sqrt{2k/t}  \sum_{i=1}^t
                   w_i
                                                         (\|r_i\|_2^2+1)^{z/2-1}\|r_i\|_2.
\end{eqnarray*}
Which implies:
\begin{eqnarray*}
  \min\{1,z/2\} \sqrt{2k/t} \cdot \sum_{i=1}^t w_i (\|r_i\|_2^2+1)^{z/2-1}\|r_i\|_2 &\leq& 2 \eps 2^{z/2} d + \max\{1,z/2\} 2^{z/2} \sqrt{dk} \Rightarrow \\
  \sum_{i=1}^t w_i(\|r_i\|_2^2+1)^{z/2-1}\|r_i\|_2 &\leq& \frac{2 \eps 2^{z/2} d + \max\{1,z/2\} 2^{z/2} \sqrt{dk} }{\sqrt{2} \cdot \min\{1,z/2\}} \cdot \sqrt{t/k}.
\end{eqnarray*}
\end{proof}

\begin{lemma}
  \label{plem:lowernorm}
  For any $0 < \eps < 1/2$ and any $k > 1$, let $d = k/(\min\{1,(z/2)^2\} 32^2 \eps^2)$ and let $P=\{e_1,\dots,e_d\}$ in $\R^{2d}$. Let $r_1,\dots,r_t \in \R^{2d}$ and let $w_1,\dots,w_t \in \R^+$ be an $\eps$-coreset for $P$, using offset $\Delta$. Then $$
  \sum_{h=1}^tw_h (\|r_h\|_2^2+1)^{z/2-1}\|r_h\|_2  \geq \frac{ 2^{z/2} d }{11 \max\{1,z/2\}\min\{1,z/2\}}.
  $$
 \end{lemma}

 \begin{proof}
   Consider the Hadamard basis $h_1,\dots,h_q$ on $q =
   1/(\min\{1,(z/2)^2\} 32^2 \eps^2)$ coordinates, i.e. the set of
   rows in the normalized Hadamard matrix. This is a set of $q$
   orthogonal unit vectors with all coordinates in
   $\{-1/\sqrt{q},1/\sqrt{q}\}$. All $h_i$ except $h_1$ have equally
   many coordinates that are $-1/\sqrt{q}$ and $1/\sqrt{q}$ and $h_1$
   have all coordinates $1/\sqrt{q}$. Now partition the first $d$
   coordinates into $k$ groups $G_1,\dots,G_{k}$ of $q$ coordinates
   each. For any $h_i$, consider the $k$ centers $v^i_1,\dots,v^i_k$
   obtained as follows: For each group $G_j$ of $q$ coordinates, copy
   $h_i$ into those coordinates to obtain $v^i_j$. We must have that
   $\sum_{h=1}^d \min_{j=1}^k \|e_h - v^i_j\|_2^p = \sum_{h=1}^d
   \min_{j=1}^k (\|e_h\|_2^2 + \|v^i_j\|_2^2 - 2\langle e_h, v^i_j
   \rangle)^{z/2}$. Since $k>1$, there is always a $j$ such that
   $\langle e_h, v^i_j \rangle = 0$. Moreover, for $i=1$, we have
   $\max_{j=1}^k \langle e_h,v^i_j \rangle = 1/\sqrt{q}$ (since all
   coordinates of $h_1$ are $1/\sqrt{d}$, and for $i \neq 1$, it holds for precisely half of all $e_h$ that $\max_{j=1}^k \langle e_h ,v^i_j \rangle = 1/\sqrt{q}$. Thus we have $\sum_{h=1}^d \min_{j=1}^k \|e_h - v^i_j\|_2^z \leq (d/2)2^{z/2} + (d/2)(2-1/\sqrt{q})^{z/2}$. By Observation~\ref{obs:p}, this is at most $(d/2)2^{z/2} + (d/2)2^{z/2}(1-\min\{1,z/2\}/(2\sqrt{q})) = d2^{z/2} - (d \min\{1,z/2\}/(4\sqrt{q}))2^{z/2} =d2^{z/2} - 8 \eps d 2^{z/2}$.
   Thus:
   \begin{eqnarray*}
     (1+\eps)(d2^{z/2} - 8 \eps d 2^{z/2}) &\geq& \Delta + \sum_{h=1}^t w_h(\|r_h\|_2^2 + 1 -2 \max_{j=1}^{k} \langle r_h, v^i_j \rangle)^{z/2} \\
     &\geq& \Delta + \sum_{h=1}^t w_h(\|r_h\|_2^2 + 1 -2 \max_{j=1}^{k} |\langle r_h, v^i_j \rangle|)^{z/2}
   \end{eqnarray*}
   By Observation~\ref{obs:p}, this is at least
   $$
   \Delta + \sum_{h=1}^t  w_h(\|r_h\|_2^2 + 1)^{z/2}  - 2 \max\{1,z/2\}\sum_{h=1}^t w_h \max_{j=1}^{k} |\langle r_h, v^i_j \rangle| (\|r_h\|_2^2+1)^{z/2-1}
   $$
   By Lemma~\ref{plem:FandSquare}, this is at least
   \begin{eqnarray*}
   (1-\eps)2^{z/2}d - 2 \max\{1,z/2\}\sum_{h=1}^t w_h \max_{j=1}^{k} |\langle r_h, v^i_j \rangle | (\|r_h\|_2^2+1)^{z/2-1}.
   \end{eqnarray*}
   We have thus shown
   \begin{eqnarray*}
    2 \max\{1,p/2\}\sum_{h=1}^t w_h \max_{j=1}^{k} |\langle r_h, v^i_j \rangle | (\|r_h\|_2^2+1)^{p/2-1}  &\geq& -2 \eps 2^{z/2} d + (1+\eps) 8 \eps d 2^{z/2}  \Rightarrow \\
     \sum_{h=1}^t w_h \max_{j=1}^{k} |\langle r_h, v^i_j \rangle | (\|r_h\|_2^2+1)^{z/2-1}  &\geq& \frac{3 \eps 2^{z/2} d}{\max\{1,z/2\}}.
   \end{eqnarray*}
   Now consider any $r_h$ with weight $w_h$. Collect the vectors $u^i_h$ such that $u^i_h = v^i_{j^*}$ where  $j^* = \argmax_j |\langle r_h , v^i_j \rangle|$.
Let $\sigma^i_h = \sign(\langle r_h , u^i_h \rangle)$.
   By construction, all these $q$ vectors are orthogonal (either disjoint support or distinct vectors from the Hadamard basis). By Cauchy-Schwartz, we then have $\langle w_h  (\|r_h\|_2^2+1)^{z/2-1} r_h, \sum_{i=1}^q \sigma^i_h u^i_h \rangle \leq w_h (\|r_h\|_2^2+1)^{z/2-1}\|r_h\|_2 \|\sum_{i=1}^q \sigma^i_h u^i_h\|_2 = w_h (\|r_h\|_2^2+1)^{z/2-1}\|r_h\|_2\sqrt{q}$. We then see that
   \begin{eqnarray*}
     \frac{3 \eps 2^{z/2} d q}{\max\{1,z/2\}} &\leq& \sum_{i=1}^q  \sum_{h=1}^t w_h \max_{j=1}^{k} |\langle r_h, v^i_j \rangle | (\|r_h\|_2^2+1)^{z/2-1}  \\
                &=& \sum_{h=1}^t \sum_{i=1}^q w_h \max_{j=1}^{k} |\langle r_h, v^i_j \rangle | (\|r_h\|_2^2+1)^{z/2-1} \\
                &=& \sum_{h=1}^t \langle  w_h (\|r_h\|_2^2+1)^{z/2-1} r_h , \sum_{i=1}^q  \sigma^i_hu_h^i \rangle \\
     &\leq&  \sum_{h=1}^tw_h (\|r_h\|_2^2+1)^{z/2-1}\|r_h\|_2\sqrt{q}.
   \end{eqnarray*}
   We have thus shown
   $$
   \sum_{h=1}^tw_h (\|r_h\|_2^2+1)^{z/2-1}\|r_h\|_2 \geq \frac{3 \eps 2^{z/2} d \sqrt{q}}{\max\{1,z/2\}} = \frac{3 \cdot 2^{z/2} d }{32 \max\{1,z/2\} \min\{1,z/2\}} \geq \frac{ 2^{z/2} d }{11 \max\{1,z/2\}\min\{1,z/2\}}. 
   $$
 \end{proof}

 We finally combine it all:

  \begin{theorem}\label{thm:lb-euclidean}
   For any $0 < \eps < 1/2$ and any $k$, let $d = k/(\min\{1,(z/2)^2\} 32^2 \eps^2)$ and let $P=\{e_1,\dots,e_d\}$ in $\R^{2d}$. Let $r_1,\dots,r_t \in \R^{2d}$ and let $w_1,\dots,w_t \in \R^+$ be an $(\eps,k,z)$-coreset for $P$, using offset $\Delta$. Then $t = \Omega\left( \frac{k}{\eps^2 \max\{1, z^4\}}\right)$.
 \end{theorem}

 \begin{proof}
   Combining Lemma~\ref{plem:uppernorm} and Lemma~\ref{plem:lowernorm}, we get $\frac{ 2^{z/2} d }{11 \max\{1,z/2\}\min\{1,z/2\}} \leq \sum_{h=1}^tw_h (\|r_h\|_2^2+1)^{z/2-1}\|r_h\|_2  \leq \frac{2 \eps 2^{z/2} d + \max\{1,z/2\} 2^{z/2} \sqrt{dk} }{\sqrt{2} \cdot \min\{1,z/2\}} \cdot \sqrt{t/k}$. That is,
   \begin{eqnarray*}
     t &\geq& \frac{k \cdot 2 \min\{1, (z/2)^2\}\cdot d^2\cdot 2^z}{(2 \eps 2^{z/2} d + \max\{1,z/2\} 2^{z/2} \sqrt{dk})^2 11^2 \max\{1,(z/2)^2\} \min\{1,(z/2)^2\} }.
   \end{eqnarray*}
   We have $\sqrt{dk} = d \min\{1,(z/2)\} 32 \eps$.
   Asymptotically, the whole bound thus becomes:
   $$
   t = \Omega\left( \frac{k}{\eps^2 \max\{1, z^4\}}\right).
   $$
 \end{proof}